\documentclass[11pt]{article}

\usepackage[T1]{fontenc}
\usepackage{lmodern}
\usepackage{microtype}
\usepackage[margin=1in]{geometry}
\usepackage{amsmath,amssymb,amsthm,mathtools,bm}
\usepackage{aliascnt}
\usepackage{braket}
\usepackage{dsfont}
\usepackage{booktabs}
\usepackage{enumitem}
\usepackage[normalem]{ulem}
\usepackage{array,longtable}
\usepackage{authblk}

\usepackage{xcolor}
\usepackage{tikz}
\usetikzlibrary{arrows.meta,positioning}
\usepackage[colorlinks=true,allcolors=blue!55!black]{hyperref}
\usepackage[nameinlink,capitalise,noabbrev]{cleveref}
\usepackage[numbers, sort&compress]{natbib}
\crefname{equation}{Eq.}{Eqs.}
\Crefname{equation}{Eq.}{Eqs.}

\allowdisplaybreaks
\setlist{nosep,leftmargin=*}
\numberwithin{equation}{section}

\newtheorem{theorem}{Theorem}[section]

\newaliascnt{lemma}{theorem}
\newtheorem{lemma}[lemma]{Lemma}
\aliascntresetthe{lemma}

\newaliascnt{proposition}{theorem}
\newtheorem{proposition}[proposition]{Proposition}
\aliascntresetthe{proposition}

\newaliascnt{corollary}{theorem}
\newtheorem{corollary}[corollary]{Corollary}
\aliascntresetthe{corollary}

\theoremstyle{definition}
\newaliascnt{definition}{theorem}
\newtheorem{definition}[definition]{Definition}
\aliascntresetthe{definition}

\newaliascnt{assumption}{theorem}

\aliascntresetthe{assumption}

\theoremstyle{remark}
\newaliascnt{remark}{theorem}
\newtheorem{remark}[remark]{Remark}
\aliascntresetthe{remark}

\crefname{assumption}{assumption}{assumptions}
\Crefname{assumption}{Assumption}{Assumptions}

\newcommand{\supp}{\operatorname{supp}}

\newcommand{\dist}{\operatorname{dist}}
\newcommand{\eps}{\varepsilon}
\newcommand{\cA}{\mathcal A}
\newcommand{\cU}{\mathcal U}
\newcommand{\cY}{\mathcal Y}
\newcommand{\bbE}{\mathbb E}
\newcommand{\bbP}{\mathbb P}

\newcommand{\norm}[1]{\left\lVert #1\right\rVert}
\newcommand{\abs}[1]{\left|#1\right|}
\newcommand{\ind}{\mathds{1}}
\newcommand{\bigO}{\mathcal O}

\newcommand{\ceil}[1]{\left\lceil #1\right\rceil}
\newcommand{\floor}[1]{\left\lfloor #1\right\rfloor}
\newcommand{\ad}{\operatorname{ad}}
\newcommand{\Tr}{\operatorname{Tr}}
\newcommand{\Jan}{J_{\rm an}}

\renewcommand{\Re}{\mathrm{Re}}
\renewcommand{\Im}{\mathrm{Im}}

\title{Emergent Prethermal Symmetries for Scalable Hamiltonian Learning}

\author[,1,2]{Myeongjin Shin\thanks{\href{mailto:hanwoolmj@kaist.ac.kr}{\texttt{hanwoolmj@kaist.ac.kr}}}}
\author[,3,4]{Junseo Lee\thanks{\href{mailto:junseolee@fas.harvard.edu}{\texttt{junseolee@fas.harvard.edu}}}}
\author[,5,6,8]{Iman Marvian\thanks{\href{mailto:iman.marvian@duke.edu}{\texttt{iman.marvian@duke.edu}}}}
\author[,5,7,8]{Yu Tong\thanks{\href{mailto:yu.tong@duke.edu}{\texttt{yu.tong@duke.edu}}}}

\affil[1]{School of Computing, Korea Advanced Institute of Science and Technology, Daejeon 34141, Republic of Korea}
\affil[2]{Institute for Quantum Information and Matter, California Institute of Technology, Pasadena, CA 91125, USA}
\affil[3]{Harvard Quantum Initiative, Harvard University, Cambridge, MA 02138, USA}
\affil[4]{The NSF AI Institute for Artificial Intelligence and Fundamental Interactions, Cambridge, MA 02139, USA}
\affil[5]{Department of Electrical and Computer Engineering, Duke University, Durham, NC 27708, USA}
\affil[6]{Department of Physics, Duke University, Durham, NC 27708, USA}
\affil[7]{Department of Mathematics, Duke University, Durham, NC 27708, USA}
\affil[8]{Duke Quantum Center, Duke University, Durham, NC 27701, USA}

\date{September 22, 2026}

\begin{document}
\maketitle

\begin{abstract}
Learning an interacting many-body Hamiltonian at the Heisenberg limit generally requires control that preserves local information for evolution times of order $1/\varepsilon$ to attain $\varepsilon$ precision. To achieve this, existing protocols often rely on trusted many-qubit operations or increasingly rapid control pulses, creating a precision barrier when gates involve unknown multi-qubit interaction or have finite duration. We show that neither resource is necessary for geometrically local Hamiltonians. Our protocol applies only static single-qubit fields, whose strength is independent of system size and grows polylogarithmically with $1/\varepsilon$. These fields generate emergent prethermal symmetries that suppress thermalization for the entire learning experiment while retaining informative symmetry-preserving dynamics. On a $d$-dimensional lattice of $n$ qubits, this enables learning every coefficient to accuracy $\varepsilon$, with failure probability at most $\delta$, using total evolution time $\bigO(\log^d(1/\eps)\log^2(n/\delta)/\eps).$ Thus, our learning protocol attains the information-theoretically optimal \(1/\varepsilon\) dependence up to polylogarithmic factors using product-state preparation, single-qubit measurements, and non-adaptive experiments. More broadly, our results establish emergent prethermal symmetries as a resource for quantum learning, opening a route to Heisenberg-limited characterization of many-body systems where fast or trusted many-qubit control is unavailable.
\end{abstract}

\tableofcontents

\section{Introduction}\label{sec:intro}
In this work, we study  learning the Hamiltonian of a quantum system from real-time evolution.  A Hamiltonian specifies how constituents of a quantum system interact with each other. Learning it is therefore a basic inverse problem in quantum science.  This problem underlies quantum sensing and metrology, where unknown fields are encoded as Hamiltonian parameters, and the calibration, diagnosis, and verification of quantum processors and analog simulators~\cite{giovannetti2011advances,degen2017quantum,GranadeFerrieWiebeCory2012robust,wiebe2014b,StilckFranca2024,hangleiter2024robustlylearninghamiltoniandynamics}.

A broad literature has developed learning protocols from dynamical data, steady states, and Gibbs states~\cite{BoixoSomma2008parameter,HolzapfelEtAl2015scalable,BaireyAradEtAl2019learning,KrastanovZhouEtAl2019stochastic,li2020hamiltonian,che2021learning,ZubidaYitzhakiEtAl2021optimal,wilde2022learnH,Caro_2024,HaahKothariTang2022optimal,gu2022practical,bairey2020learning,BakshiLiuMoitraTang2024learning,yu2023robust,MobusBluhmCaroEtAl2023dissipation}. Recent advances extend Hamiltonian learning to fermionic and bosonic systems, as well as beyond fixed Hamiltonian ansätze~\cite{MiraniHayden2024learning,NiLiYing2024quantum,LiTongNiGefenYing2023heisenberg,MobusBluhmGefenEtAl2025heisenberg,ma2024learningkbodyhamiltonianscompressed,HuMaGongEtAl2025,bakshi2024structure,SinhaTong2025improved,Rosati2025quantum,Zhao2025learning,ChenCotlerHuang2025quantum}. Along with theoretical advances, significant experimental progress has also been made for this important problem~\cite{hangleiter2024robustlylearninghamiltoniandynamics,GuoEtAlDuan2025HamLearn,franceschetto2025hamiltonian}. Of particular importance are protocols that attain Heisenberg-limited scaling~\cite{HuangTongFangSu2023learning,dutkiewicz2024advantage,LiTongNiGefenYing2023heisenberg,MiraniHayden2024learning,NiLiYing2024quantum,HuMaGongEtAl2025,MobusBluhmGefenEtAl2025heisenberg,ma2024learningkbodyhamiltonianscompressed,bakshi2024structure,brahmachari2026}.  Here the Heisenberg limit means that estimating the Hamiltonian to error $\eps$ requires total evolution time proportional to $1/\eps$, up to logarithmic factors. This dependence on $\eps$ is information-theoretically optimal and quadratically improves on the standard $1/\eps^2$ scaling.

This improvement in evolution time comes with an important operational cost. Dutkiewicz et al. proved that, for broad classes of many-body Hamiltonians, including Hamiltonians that thermalize according to the eigenstate thermalization hypothesis, learning without control is necessarily standard-quantum-limited, whereas controlled evolution can attain the Heisenberg limit~\cite{dutkiewicz2024advantage}.  Thus, for generic thermalizing systems, quantum control is not merely a convenience but the resource that preserves information at long times.  On the other hand, Hamiltonian learning is needed precisely when the system Hamiltonian is not yet known well enough to support complicated, trusted control.  For example, entangling gates in trapped ions depend on calibrated collective motional modes, while neutral-atom gates depend on calibrated Rydberg excitation and blockade interactions~\cite{bruzewicz2019trapped,saffman2010rydberg}.  These two observations create a tension: Heisenberg-limited Hamiltonian learning requires control, and yet it is most useful only when the available control is limited.

This tension motivates us to limit the control operations we use to those that do not rely on detailed knowledge of the Hamiltonian to be learned. For example, single-qubit rotations can often be calibrated locally: transition frequencies are determined by Ramsey spectroscopy, pulse areas by repeated resonant rotations, and rotation axes by the phase of a coherent drive, without first reconstructing the full interqubit Hamiltonian~\cite{harty2014highfidelity,sheng2018highfidelity}.  This motivates treating static and dynamical single-qubit control as the experimentally trusted resource.

Restricting control to single qubits, however, does not by itself remove the obstruction.  The Hamiltonian-reshaping protocol of Huang et al. uses only single-qubit gates, but its control interval must scale as $\bigO(\sqrt{\eps})$ to reach precision $\eps$~\cite{HuangTongFangSu2023learning}.  However, real gates have a nonzero duration. Once the required interval falls below this gate time, the idealized reshaping sequence can no longer be implemented and the protocol encounters a precision floor.  A natural question is therefore:
\begin{quote}
    \textit{Is Heisenberg-limited Hamiltonian learning still possible using only static single-qubit control whose strength grows mildly with the target precision?}
\end{quote}

This work provides an affirmative answer to this question. Our starting point is the connection between thermalization and learning. The lower bound of Ref.~\cite{dutkiewicz2024advantage} identifies thermalization as the mechanism that erases the long-time local information needed for Heisenberg scaling.  From this perspective, existing control-based protocols can be viewed as preventing thermalization over the time window needed for learning.  In particular, Hamiltonian reshaping decouples the system into small patches and produces an exponential number of local conservation laws, preserving useful coherent dynamics up to time $\bigO(1/\eps)$ when gates are applied with an interval of $\bigO(\sqrt{\eps})$ \cite{HuangTongFangSu2023learning}. 

In this work, we will engineer prethermal symmetry to delay thermalization beyond the time window required for learning. Prethermalization describes a long but finite regime in which approximate conservation laws constrain the dynamics and thereby delay thermalization~\cite{gring2012relaxation,neyenhuis2017observation,mallayya2019prethermalization}. The system need not be frozen during this regime, as local observables can continue to exhibit nontrivial coherent dynamics, while processes that violate the approximate conservation laws remain small. This phenomenon can arise from a large separation of energy scales, which makes violating the approximate conservation laws energetically costly \cite{else2017prethermal}.

Abanin, De Roeck, Ho, and Huveneers proposed a mathematical framework to rigorously analyze the prethermalization phenomenon \cite{abanin2017rigorous}. In particular, for a local Hamiltonian with a large energy scale $\nu$ multiplying an integer-spectrum charge, they constructed a recursion process, which we shall refer to as the ADHH recursion, to produce a quasi-local change of frame in which the effective Hamiltonian commutes with a dressed charge, up to a quasi-exponentially small remainder~\cite{abanin2017rigorous}.  Consequently, this charge is approximately conserved for a time of order $\exp(\Omega(\nu/\log^3\nu))$. This almost-exponential prethermalization window provides ample time to extract information from coherent dynamics with moderately large $\nu$.  In our setting, the applied static field provides the charge, and its dressed conservation law is an \emph{emergent prethermal symmetry}. Choosing $\nu=\operatorname{polylog}(1/\eps)$ makes this symmetry persist through the entire $\bigO(1/\eps)$ learning experiment.  We note that in addition to the desirable polylogarithmic dependence on precision, this parameter choice is also independent of the system size. This idea of recursively constructing an effective Hamiltonian shares many similarities with the Schrieffer-Wolff transformation~\cite{BravyiDiVincenzoLoss2011schrieffer,schrieffer1966relation}.

We use this prethermal window to learn all coefficients of a geometrically local many-body Hamiltonian.  We apply random static single-qubit fields to create simple charge sectors whose symmetry-preserving dynamics encode local transition energies. These energies are then estimated through experimental protocols similar to Ramsey interferometry, using product initial states and single-qubit measurements. A modified version of the robust phase estimation algorithm is used for classical post-processing~\cite{KimmelLowYoder2015robust}. A stable classical inverse map, which involves iteratively solving an optimization problem with rigorous convergence guarantee, recovers the Hamiltonian.  Throughout this process, the required field strength is independent of system size and grows only polylogarithmically with $1/\eps$.  Finite achievable field strength can therefore still impose a precision limit, but the dependence is exponentially milder than the gate-time barrier above.  Up to polylogarithmic overheads, the protocol retains the main features of Ref.~\cite{HuangTongFangSu2023learning}: Heisenberg-limited total evolution time, logarithmic dependence on system size, product-state preparation, single-qubit measurements, non-adaptive experiments, and efficient classical post-processing, despite using a weaker form of control in the sense explained above.

\subsection{Problem setup}
In what follows we will describe the detailed set up of the quantum system and the locality assumptions we use. We consider a Hamiltonian on a $d$-dimensional physical lattice $\Lambda=[-\mathsf{L},-\mathsf{L}+1,\cdots,\mathsf{L}]^d$, with  size $|\Lambda|=n$, where $\mathsf{L}$ is a positive integer. The Hamiltonian takes the form 
\begin{equation}
 H(\lambda)=\sum_{a\in\cA}\lambda_aP_a,
 \quad \lambda\in\mathbb R^M,
 \label{eq:intro-H}
\end{equation}
where each $P_a$ is a distinct, non-identity, and one- or multi-qubit, Pauli operator, and $\mathcal{A}$ is the set of Pauli term indices, with $|\mathcal{A}|=M$. Throughout this work we assume that the Hamiltonian is $q$-local and involve interaction range $r_0$, i.e., 
\begin{equation}
\label{eq:locality-bounds-H}
    |\mathrm{supp}(P_a)|\leq q,\quad \mathrm{diam}(\mathrm{supp}(P_a))\leq r_0,\quad \forall a\in \cA.
\end{equation}
We assume that $d,q,r_0>0$ are constants that are independent of the system size $n$. These assumptions ensure that the total number of terms is $M= |\mathcal{A}| = \bigO(n)$. Our goal is to obtain an estimate $\hat{\lambda}$ such that
\begin{equation}
 \norm{\widehat\lambda-\lambda}_\infty\le\eps,
 \label{eq:intro-goal}
\end{equation}
where $\|\cdot\|_\infty$ denotes the $\ell^\infty$-norm.

We measure the cost of our learning protocol primarily by its \emph{total evolution time}, defined as the cumulative time for which the system evolves under the unknown Hamiltonian across all experimental runs. Specifically, if the protocol performs \(N_s\) experiments with evolution times \(T_1,T_2,\ldots,T_{N_s}\), then the total evolution time is \(\sum_{i=1}^{N_s}T_i\). We will also report the \emph{number of experiments} needed.

\subsection{Main result}
\label{sec:main-result}

Our main result is a parallel protocol for learning all coefficients of a geometrically local Hamiltonian using static single-qubit fields, product initial states, and single-qubit measurements.  The protocol does not require preparing an eigenstate of the unknown Hamiltonian or implementing time-dependent many-body control. Instead, it divides the lattice into a small number of batches of well separated sites and probes every site in one batch during the same physical evolution. 

For each batch and each repetition, we choose independent random Pauli axes and signs for the single-qubit fields. Every site in the batch is called a ``defect''. The field  has full strength within disjoint balls \(B_R(k)\) around the defects, is weakened by a factor of two at every defect, and is zero outside the union of all balls \(B_R(k)\). We prepare a product state, evolve under the unknown Hamiltonian together with this static field, and measure two single-qubit Pauli observables at every defect.  Robust frequency estimation~\cite{KimmelLowYoder2015robust,LiTongNiGefenYing2023heisenberg} converts these measurements into estimates of local transition energies.  Finally, a classical fitting procedure compares the estimated energies with a forward model evaluated on small neighborhoods and returns the Hamiltonian coefficients.

The defect batches are obtained by coloring the lattice so that sites in the same batch are separated by distance $\bigO(\log(1/\eps))$. Consequently, only $\bigO(\log^d(1/\eps))$ batches are needed in $d$ spatial dimensions.  This is the source of the polylogarithmic overhead in the parallel protocol.  The formal statement appears later as \Cref{thm:main-para} and we state its main content here.

\begin{theorem}[Main result, informal]
\label{thm:overview-main-result}
Consider the class of geometrically local Hamiltonians described above, with fixed spatial dimension, interaction weight, and interaction range, and with $\norm\lambda_\infty\le1$.  For every target accuracy $\eps>0$ and failure probability $\delta\in(0,1)$, static single-qubit control together with single-qubit preparation and measurement suffices to produce an estimate $\widehat\lambda$ satisfying $\norm{\widehat\lambda-\lambda}_\infty\le\eps$ with probability at least $1-\delta$.  It is sufficient to use a field strength 
\begin{equation}
    \nu\ge\nu_{\min}:=\bigO\!\left(\max\left\{\log^d(1/\eps), \log(1/\eps)\left[1+\log\log(1/\eps)\right]^3\right\}\right)
\end{equation}
The total evolution time is
\begin{equation}
 \bigO\left({\log^d(1/\eps)\log^2(n/\delta)}{/\eps}\right).
 \label{eq:overview-main-time}
\end{equation}
{The protocol uses $\bigO(\log^{d+2}(1/\eps)\log^2(n/\delta))$ independent non-adaptive global experiments.} The classical post-processing time is polynomial in $n$, $1/\eps$, and $\log(1/\delta)$.
\end{theorem}
 
We note that we saturate the optimal Heisenberg limit up to a polylogarithmic factor in $\eps$.  Moreover, parallelization results in a logarithmic dependence of the total evolution time on the system size $n$.  A central point of the result is that the required field strength does not grow with $n$, even though the Hamiltonian acts on an extensive system and has $\Theta(n)$ spectral norm. In \Cref{tab:comparison_HL_methods} we give a detailed comparison with the most relevant Heisenberg-limited Hamiltonian learning methods.
Throughout this work, classical computational cost is measured by the number of arithmetic operations, with each operation assigned unit cost. 

\begin{remark}[Robustness to SPAM errors]
The protocol is robust to a certain amount of state preparation and measurement (SPAM) errors. All experimentally obtained data enter the Hamiltonian reconstruction procedure only through the robust frequency estimation subroutine, which is a variant of robust phase estimation~\cite{KimmelLowYoder2015robust}. These methods tolerate a constant amount of deviation of the measurement probabilities from the ideal values, which can be a result of SPAM noise, and retain the Heisenberg-limited scaling despite such deviation. 
As a result, the overall cost scaling of our protocol remains unchanged in the presence of  measurement probabilities errors due to imperfect SPAM below a constant threshold.
\end{remark}

Although the ADHH construction~\cite{abanin2017rigorous} provides the prethermal normal form, our algorithmic guarantees are not direct consequences of its results. In \Cref{thm:uniform-normal-form}, we establish a uniform analytic continuation of the ADHH recursion to non-Hermitian interactions; combined with the Cauchy estimates in \Cref{lem:analytic-gap-bound}, this controls the Taylor truncations used in our locality and parallelization arguments. We also quantify how the resulting effective transition energies respond to perturbations of the underlying Hamiltonian coefficients in \Cref{lem:row-derivatives}, providing the derivative bounds that lead to the stable inverse guarantee of \Cref{thm:local-inverse}. Finally, \Cref{lem:D-taylor-expansion-commuting-support} combines the locality and charge symmetry of the effective Hamiltonian, and \Cref{lem:parallel-code-factorization} uses this structure to identify the defect product states as eigenstates of the truncated effective Hamiltonian. In the parallel analysis below, \Cref{lem:parallel-stopped-dynamics,thm:parallel-Ramsey} bound the change of each measured signal under the ADHH unitary transformation and relate it to the corresponding effective transition energy.\footnote{In all of our proofs, we avoid using \cite[Theorem~3.3]{abanin2017rigorous} because of our inability to independently verify one of the steps in its proof.} These ingredients use new techniques that go beyond those developed in the ADHH analysis.

\subsection{Comparison with previous works}
\label{sec:comparison}

In \Cref{tab:comparison_HL_methods} we compare our result with representative
Heisenberg-limited Hamiltonian learning protocols.  

\begin{table}[ht]
    \centering
    \begin{small}
    \begin{tabular}{c|cccc}
    \hline
     Hamiltonian learning protocol    &  Single-qubit & $1/\tau$ scaling & $\nu$ scaling & Many-body scalable \\
    \hline
    Huang et al. \cite{HuangTongFangSu2023learning}    & \textcolor{blue}{Yes} & \textcolor{red}{$\bigO(1/\sqrt{\eps})$} & - & \textcolor{blue}{Yes} \\
    Dutkiewicz et al. \cite{dutkiewicz2024advantage} & \textcolor{red}{No} & - & \textcolor{blue}{$\bigO(1)$} & \textcolor{blue}{Yes} \\
    Bakshi et al. \cite{bakshi2024structure} & \textcolor{red}{No} & \textcolor{blue}{$\bigO(1)$} & - & \textcolor{blue}{Yes} \\
    Brahmachari et al. \cite{brahmachari2026}  & \textcolor{blue}{Yes} & - & \textcolor{blue}{$\bigO(1)$} & \textcolor{red}{No} \\
    \textbf{This work} & \textcolor{blue}{Yes} & - & \textcolor{blue}{$\widetilde{O}(1)$} & \textcolor{blue}{Yes} \\
    \hline
    \end{tabular}
  \end{small}
    \caption{Comparison of selected Heisenberg-limited methods for Hamiltonian learning from prior work and the present work. Here, $\tau$ denotes the time interval between successive pulsed-control operations, whereas $\nu$ denotes the characteristic strength or frequency scale of continuous or static control. The final column indicates whether each method can learn geometrically local many-body Hamiltonians on $n$ qubits in $\mathrm{poly}(n)$ time. For the method in the final row, the factor denoted by $\widetilde{O}(1)$ in the $\nu$-scaling is, more precisely,
    $\bigO(\max\{\log^d(1/\eps), \log(1/\eps)[1+\log\log(1/\eps)]^3\}).$}
\label{tab:comparison_HL_methods}
\end{table}

To the best of our knowledge, our protocol is the only Heisenberg-limited method in this comparison that simultaneously uses only single-qubit control, scales to geometrically local many-body Hamiltonians, and requires a control scale that is independent of $n$ and only polylogarithmic in $1/\eps$. Huang et al.~\cite{HuangTongFangSu2023learning} use single-qubit control and achieve many-body scalability, but the inverse control interval grows as $\bigO(1/\sqrt{\eps})$.  The scalable protocols of Dutkiewicz et al.~\cite{dutkiewicz2024advantage} and Bakshi et al.~\cite{bakshi2024structure} instead require multi-qubit control.

Brahmachari et al.~\cite{brahmachari2026}, a precursor to the present work, introduced a Heisenberg-limited learning protocol based on static single-qubit fields whose strength is independent of $\eps$.  Their analysis applies to systems of fixed size, $n=\bigO(1)$, and leaves a scalable many-body extension open.  Here we obtain such an extension by combining randomized single-qubit fields, a prethermal approximation that specifically tracks locality, and parallel learning using multiple well-separated defects.

Related dynamical-approximation results also illustrate why independence from $n$ does not follow from conventional global-propagator bounds.  At a qualitative level, our use of strong static fields may seem similar to the quantum Zeno dynamics induced by a strong field, but an important difference is that in quantum Zeno effect the field strength typically needs to scale with the system size, whereas in prethermalization setting and our learning protocol, the field strength is independent of the system size.\footnote{As an example, the ``eternal approximation'' of \cite[Theorem~2]{Burgarth2022oneboundtorulethem} has an error bound that is uniform in time, but it depends on the ratio between the global Hamiltonian norm and the control frequency.  At fixed approximation error, satisfying this bound for an extensive Hamiltonian with $\norm{H}=\Theta(n)$ requires a control frequency of order $n$, or equivalently a control period of order $1/n$. Similarly, the error bound for a recently proposed learning protocol based on the quantum Zeno effect depends on the global Hamiltonian norm~\cite{franceschetto2025hamiltonian}.  At fixed Zeno error and for evolution time $T$, the stated bound entails a kick rate proportional to $\norm{H}(1+T\norm{H})$.  Thus, for fixed $T$ and $\norm{H}=\Theta(n)$, the prescribed kick rate is quadratic in $n$. In the short-time regime $T\norm{H}=\bigO(1)$, it is linear in $n$.}

Our proof applies the local-norm estimates of~\cite{abanin2017rigorous} to a field confined to a neighborhood of the measured qubit. The error from the unitary transformation is then controlled by the volume of that neighborhood. We also show that the exterior evolves almost identically for the two states of the measured qubit, up to their relative phase. Together, these estimates control the measured signals through evolution times \(O(1/\eps)\), with a field strength independent of \(n\). The resulting frequencies retain the information needed for stable coefficient recovery.

\subsection{Outlook}

Prethermalization may also provide a natural framework for Hamiltonian certification and property testing~\cite{wiebe2014b,GaoEtAl2026certification,BluhmEtAl2026certifying,bluhm2024hamiltonian,franca2025learning}. In our setting, the emergent prethermal symmetry and the associated long-lived local oscillations could provide diagnostic tools for certification: rather than reconstructing every coefficient, one could test whether measured transition frequencies and approximate conservation laws agree with those predicted by a target Hamiltonian.  A related direction is Hamiltonian learning through high-frequency driving, which guarantees a long prethermal window governed by a quasi-local effective Hamiltonian~\cite{abanin2017rigorous}. Adapting our local spectroscopy and analytic-continuation techniques could provide rigorous and scalable learning guarantees throughout this window. Moreover, the underlying strategy suggests extensions beyond qubits and Pauli interactions. Developing suitable prethermal charges and spectroscopy protocols for qudits, general spin systems, and fermionic or bosonic lattices would connect this approach with recent advances in Hamiltonian learning for interacting fermions and bosons~\cite{MiraniHayden2024learning,LiTongNiGefenYing2023heisenberg} and clarify how broadly prethermalization can serve as a resource for scalable learning and certification. Beyond learning and certification, the analytic tools developed here may be useful for studying other emergent phenomena in prethermal systems.  Techniques developed in this work, such as analytic continuation and Taylor truncation of the effective Hamiltonian, provide a general framework for obtaining quantitative guarantees throughout the prethermal window.  Possible applications include prethermal Hilbert-space fragmentation and other forms of constrained or nonergodic dynamics~\cite{yoshinaga2022emergence,ghosh2023prethermal,stephen2024ergodicity}.

\section{Technical overview}
\label{sec:tenchical-overview}

We will now explain the main ideas behind \Cref{thm:overview-main-result}.  We will describe in more detail the learning protocol, and explains why it works.  In particular, we will show why static fields turn an extensive many-body learning problem into a collection of local spectroscopy experiments, why the resulting classical inverse problem remains well conditioned, and why sufficiently separated experiments can be run during the same evolution.

The unknown parameters $\lambda_a$ describe the strengths of the local interactions $P_a$ in $H(\lambda)=\sum_{a\in\mathcal A}\lambda_aP_a$.  Learning these parameters from dynamics is difficult for two related reasons.  First, a local measurement generally depends on many noncommuting terms of the Hamiltonian. Second, even when an experiment produces useful spectral information, one must still determine whether this information can be inverted stably to recover every coefficient.

Locality prevents these difficulties from worsening with the volume in the way that a naive operator-norm estimate would suggest. The global operator norm $\norm{H}$ may grow as $\bigO(n)$, as is natural for an extensive Hamiltonian, but the scale relevant to our analysis is instead the local interaction strength $J$ (defined in \Cref{eq:Ham-local-norm-bound}): roughly speaking, $J$ bounds the total strength of the terms touching any one site.  For geometrically local Hamiltonians, $J$ is independent of $n$.  This distinction is crucial to obtaining a static field strength that is independent of system size $n$.

Our protocol proceeds in three stages.  We first add a static single-qubit field with one deliberately weakened site, which, as mentioned previously, we call a defect. The field is applied only within a neighborhood of the defect. We prepare the defect in a superposition of its two local field eigenstates and the other controlled qubits in fixed field eigenstates. Although the spins outside the neighborhood continue evolving, their effect on the two components of this superposition is almost the same.  A prethermalization result shows that, for a long time, the unknown Hamiltonian only weakly disturbs the symmetry imposed by the field.  The measured signal therefore behaves approximately as a single-frequency oscillation, whose frequency contains an effective
transition energy.

We next repeat this experiment with random local Pauli axes.  Different choices of axes probe different linear combinations of the unknown coefficients, and we recover the coefficients by fitting the measured transition energies. At finite field strength these relations are nonlinear, which presents considerable difficulty for direct analysis. Fortunately, the strong-field limit provides a simple relation between coefficients and transition energies as described in \Cref{eq:bare-linear}. By quantifying how far the non-linear relation deviates from this linear one, we will show that the fitting procedure recovers the actual coefficients to desired accuracy, and can be solved efficiently using a classical iterative algorithm with exponential convergence.

Finally, we perform many defect experiments at the same time.  This step is not automatic: with several defects, the field no longer isolates a two-dimensional subspace for each one.  We resolve this by separating the defects in space and showing that processes capable of coupling two defects appear only at high order.  Their total effect then decreases exponentially with the defect separation.

Throughout the paper, the spatial dimension $d$, interaction weight $q$, interaction range $r_0$, and local Pauli terms $\{P_a\}_{a\in\mathcal{A}}$ are fixed.  We normalize energy units so that $\norm\lambda_\infty\le1$. Analytic continuation will be an important technical tool for analyzing various approximation errors, and for this purpose we will  sometimes replace $H$ by a non-Hermitian interaction. 

\subsection{Static fields and local spectroscopy}

Here we will describe the specific protocol we use. For a defect \(k\), we will perform a series of experiments indexed by $\ell$. Experiments across all defects are then indexed by $\rho=(k,\ell)$. For each experiment indexed by $\rho$, we choose independent random signed Pauli axes \(\tau_{j,z}^{\rho}\), and set \(Q_j^\rho=(I-\tau_{j,z}^\rho)/2\). On \(B=B_R(k)\), we apply a half-strength field at \(k\) and full strength at the other sites. The field is zero on \(C=\Lambda\setminus B\). This gives us the control Hamiltonian $H_{\rm ctrl}^\rho$, the full generator of the dynamics $G_\rho^B$, and the charge operator $N_\rho^B$, which gives rise to the approximate prethermal symmetry.
\begin{equation}
 \begin{aligned}
 H_{\rm ctrl}^\rho&=\frac12\tau_{k,z}^\rho+
 \sum_{j\in B\setminus\{k\}}\tau_{j,z}^\rho,\quad 
 G_\rho^B=H(\lambda)+\nu N_\rho^B,
 \quad N_\rho^B=Q_k^\rho+2\sum_{j\in B\setminus\{k\}}Q_j^\rho.
 \end{aligned}
 \label{eq:overview-charge}
\end{equation}
The physical Hamiltonian \(H(\lambda)-\nu H_{\rm ctrl}^\rho\) differs from \(G_\rho^B\) only by a scalar. Let \(\phi_{B,0}^\rho\) be the local-frame vacuum on \(B\), and \(\phi_{B,1}^\rho\) its defect flip. The entire charge-zero and charge-one sectors are \(\ket{\phi_{B,b}^\rho}\otimes\mathcal H_C\), \(b=0,1\). Prepare their equal superposition on \(B\), tensored with a fixed branch-independent exterior state. A product state suffices.

We analyze this experiment using an iterative unitary transformation, called the ADHH recursion and introduced below in \Cref{sec:prethermal-effective-Ham}. After \(s\) steps, it gives
\begin{equation}
 Y_{\rho,s}G_\rho^B Y_{\rho,s}^\dagger
 =\nu N_\rho^B+D_{\rho,s}^{\rm ball}+V_{\rho,s}^{\rm ball},
 \qquad [D_{\rho,s}^{\rm ball},N_\rho^B]=0.
 \label{eq:overview-normal-form}
\end{equation}
Here \(Y_{\rho,s}\) is unitary, \(D_{\rho,s}^{\rm ball}\) preserves the charge imposed by the field, and \(V_{\rho,s}^{\rm ball}\) is a small remainder. The fact that the field is confined to \(B\) enables us bound both \(\|Y_{\rho,s}-I\|\) and \(\|V_{\rho,s}^{\rm ball}\|\) using \(|B|\), rather than the full system size. Specifically, these bounds are \(CJ|B|/\nu\) and \(CJ|B|(2/3)^s\), respectively. The construction is controlled up to a maximum number of steps
\begin{equation}
 n_*(\nu)=\Theta\!\left(
 \frac{\nu/J}{[1+\log(\nu/J)]^3}\right).
 \label{eq:overview-nstar}
\end{equation}
We choose \(s\le n_*(\nu)\) from the desired accuracy and use the same number of steps for all experiments and trial Hamiltonians.

To describe the frequency that we wish to extract, it is useful to introduce a reference model in which the field extends to every qubit. Its charge is \(N_\rho=Q_k^\rho+2\sum_{j\ne k}Q_j^\rho\). Let \(\ket{\phi_{j,b}^\rho}\) be the eigenstate of \(Q_j^\rho\) with eigenvalue \(b\in\{0,1\}\). The eigenvalues zero and one of \(N_\rho\) each correspond to a single product state:
\begin{equation}
 \ket{\Phi_0^\rho}=\bigotimes_j\ket{\phi_{j,0}^\rho},
 \qquad \ket{\Phi_1^\rho}=\ket{\phi_{k,1}^\rho}
 \bigotimes_{j\ne k}\ket{\phi_{j,0}^\rho}.
 \label{eq:overview-doublet}
\end{equation}
Applying the same \(s\)-step construction with charge \(N_\rho\) gives an interaction \(\widehat D_\rho^\Lambda\) commuting with \(N_\rho\). These two states are therefore eigenstates of \(\widehat D_\rho^\Lambda\), and we define their half transition energy by
\begin{equation}
 E_\rho^\Lambda(\lambda)=\frac12\left(
 \bra{\Phi_1^\rho}\widehat D_\rho^\Lambda\ket{\Phi_1^\rho}
 -\bra{\Phi_0^\rho}\widehat D_\rho^\Lambda\ket{\Phi_0^\rho}\right).
 \label{eq:overview-full-gap}
\end{equation}
This model is auxiliary in the sense that we only use it as an intermediate step for comparison, and the experiment still applies the field only within \(B\). To relate the two models, we use locality. An interaction process that can distinguish the field configurations must reach from the defect to the exterior of \(B\). Such processes require many local interaction terms, and their contribution decreases exponentially with the radius. As a result, the exterior evolution is almost the same for the two defect states, apart from the phase determined by \(E_\rho^\Lambda\).

We will perform an experiment with the static field restricted to a ball, as described above and in more detail in \Cref{sec:experiment}. Measuring two single-qubit observables on defect $k$ yields a complex signal \(S_\rho(t)\), similar to the signal one gets in Ramsey interferometry experiments. Let \(K\ge0\) be the integer specifying how many orders of the locality expansion are kept; we explain this expansion in the next subsection. The following informal statement will be proved below as \Cref{thm:Ramsey-approximation}.

\begin{theorem}[Informal Ramsey approximation]
For \(R\ge(K+1)r_0\), every real \(t\), and \(V_R=\max_k|B_R(k)|\), the physical defect signal $S_\rho(t)$ satisfies
\begin{equation}
 \abs{S_\rho(t)-e^{i[\nu+2E_\rho^\Lambda(\lambda)]t}}
 \le \frac{C_1JV_R}{\nu}
 +C_2JV_R|t|(2/3)^s+C_3J|t|2^{-K}.
\end{equation}
All constants are independent of the system size.
\end{theorem}
The three errors come from the unitary transformation, the remainder after \(s\) steps, and the influence of terms extending outside \(B\). The single-frequency form of $e^{i[\nu+2E_\rho^\Lambda(\lambda)]t}$ allows us to apply robust frequency estimation to estimate the half transition energy $E_\rho^\Lambda(\lambda)$. To estimate it to accuracy \(\eta>0\), robust frequency estimation  requires evolution times at most \(O(1/\eta)\) and a small constant signal error over those times. Choosing \(s=K=\ceil{C\log(eJ/\eta)}\), \(R=(K+1)r_0\), and
\begin{equation}
 \nu\ge\nu_{\rm dyn}=CJ\max\{V_R,\ s[1+\log(s+1)]^3\}
 \label{eq:overview-dynamic-field}
\end{equation}
meets these requirements. The two terms in the maximum ensure, respectively, a small error from \(Y_{\rho,s}\) and that the chosen \(s\) is admissible. We justify the frequency-estimation step below in \Cref{subsec:frequency-estimation}. Neither the field strength nor the chosen number of recursion steps needs to grow with \(n\).

\subsection{Local approximation of the transition energies}

The experiments described above provide estimates of $E_\rho^\Lambda(\lambda)$ for many field patterns $\rho$.  To recover $\lambda$, the classical algorithm takes a trial coefficient vector $\mu$ and predicts the transition energies that would be observed for $H(\mu)$.  The map from $\mu$ to these predicted values is the \emph{forward map}. As defined in \Cref{eq:overview-full-gap}, however, $E_\rho^\Lambda(\mu)$ is obtained from an effective Hamiltonian on the full lattice.  Directly evaluating it would require treating the entire many-body system. Fortunately, as we will prove rigorously in \Cref{sec:local-approx}, the transition energies can be well-approximated locally, and therefore it is enough to work on a neighborhood of the defect. Below we will provide a high-level overview of the proof.

For a ball with radius $R$ centered around $k$, which we will refer to as a \emph{patch} and denote by $B_R(k)$, let $H_{B_R(k)}(\mu)$ contain only the candidate interactions supported inside the ball.  We run the same \(s\)-step prethermal construction \emph{on this smaller system} and denote its predicted half transition energy by $F_{\rho,R}(\mu)$. The integer \(s\le n_*(\nu)\) is chosen once and used for every field pattern $\rho$, patch $B_R(k)$, and trial vector $\mu$.  This common choice lets us regard the construction as one fixed analytic function of the Hamiltonian coefficients.

\paragraph{A bookkeeping parameter.}
To compare the full system with the patch, we introduce a scalar $z$ that turns on the unknown Hamiltonian continuously.
\begin{equation}
 g_\rho^\Lambda(z)=E_\rho^\Lambda(zH(\mu)),
 \quad
 g_\rho^{B}(z)=F_{\rho,R}(z\mu).
 \label{eq:overview-z-functions}
\end{equation}
At $z=0$ the unknown interaction is absent, while the coefficient of $z^p$ describes contributions involving $p$ interaction terms.   We show in \Cref{thm:uniform-normal-form} that the ADHH recursion used to construct the prethermal effective Hamiltonian is well defined even for complex $z$ in a disc of radius two, and that the resulting effective interactions remain bounded by $\bigO(J)$ there.

\paragraph{Approximation through connectedness}
Let us consider the Taylor expansion in $z$ of $\widehat{D}_\rho^\Lambda(z\mu )$ obtained through the ADHH recursion for $zH(\mu)$ (note that $zH(\mu)=H(z\mu)$). As we will prove in \Cref{lem:connected-word-expansion}, a nonzero order-$p$ contribution can be built only from a connected cluster of $p$ overlapping Hamiltonian terms.  Since each original term has range at most $r_0$, such a chain can travel distance at most $\bigO(pr_0)$.  We formalize these connected clusters as \emph{connected $p$-words}: ordered products of $p$ Pauli operators each sharing the same support as a Hamiltonian interaction term, whose overlap graph is connected (see \Cref{def:p-word}).  If a connected $p$-word affects the transition at $k$, then its support must contain $k$. Connectedness then ensures that this connected $p$-word only involves terms that are within  distance $pr_0$ of the defect.

It follows that the Taylor expansions of the full-system $\Lambda$ and patch $B_R(k)$ transition energies agree through degree $K$ whenever $R\ge(K+1)r_0$. For degrees beyond $K$, analyticity can be used to bound the coefficients. More specifically, Cauchy's formula (\Cref{lem:Cauchy-general}) on the radius-two disc bounds the coefficient of $z^p$ by $\bigO(J2^{-p})$.  The unmatched tail therefore decays exponentially in $K$.

\begin{theorem}[Informal version of \Cref{thm:patch-locality}]
If $R\ge(K+1)r_0$, then, uniformly over every pattern $\rho$ and every trial point $\mu$ in the unit $\ell^\infty$-neighborhood of $\lambda$,
\begin{equation}
 \abs{E_\rho^\Lambda(\mu)-F_{\rho,R}(\mu)}
 \le C_{\rm P}J2^{-K}.
 \label{eq:overview-patch-error}
\end{equation}
The constant is independent of the system size $n$ and the pattern $\rho$.
\end{theorem}

Taking $K=\Theta(\log(J/\eps))$ and $R=\Theta(r_0\log(J/\eps))$ makes the patch error $\bigO(\eps)$. For the rest of this section, we will keep $R$ fixed, and abbreviate $F_{\rho,R}$ simply as $F_{\rho}$. The measured data can therefore be written as
\begin{equation}
 y=F(\lambda)+e,
 \quad
 F(\mu):=(F_{\rho}(\mu))_\rho,
 \quad
 \norm e_\infty\le\eta_{\rm freq}+C_{\rm P}J2^{-K}.
 \label{eq:overview-data-model}
\end{equation}
Here $\eta_{\rm freq}$ is the experimental frequency-estimation error, while $C_{\rm P}J2^{-K}$ is the deterministic error from replacing the full lattice by a patch.  Thus we can use classical algorithm to fit the patch map $F$rather than the exact full-system energies, and \Cref{eq:overview-patch-error} accounts for the difference. This enables us to design an efficient classical algorithm for solving the optimization problem, and makes it much easier to analyze its stability.

\paragraph{Controlling how the predictions change.}
A bound on the predicted energies is not by itself enough to analyze the fit.  We also need to know how rapidly they change when one or more trial coefficients are varied.  The same complex extension discusssed above gives these derivative bounds.  Applying Cauchy's formula in the coefficient directions shows that the finite-field correction and its first two derivatives are small.  If $F_0$ denotes the $\nu\to\infty$ limit and $\mathcal R=F-F_0$, then each $\rho$ obeys (using $J=\bigO(1)$)
\begin{equation}
 \norm{D\mathcal R_\rho(\mu)}_1=\bigO(1/\nu),
 \quad \sup_a\sum_b\abs{\partial_a\partial_bF_\rho(\mu)}
 =\bigO(1/\nu),
 \label{eq:overview-derivatives}
\end{equation}
uniformly over the same coefficient neighborhood. The precise statement is in  \Cref{lem:row-derivatives}.  The first bound controls the change in the linear response of the data, while the second controls the Hessian of the least-squares objective used below.

Although $|B_R(k)|=\bigO(R^d)$, the classical algorithm does not diagonalize the entire patch, which would cost $2^{\bigO(R^d)}$.  Instead, it evaluates the Taylor series and its derivatives by enumerating connected clusters only through logarithmic degree.  \Cref{prop:classical-row-evaluation} and Appendix~\ref{app:classical-evaluation} show that the resulting total post-processing time is polynomial in $n$, $1/\eps$, and $\log(1/\delta)$.

\subsection{Recovering the Hamiltonian coefficients}

We now discuss using the transition-energy estimates to recover Hamiltonian coefficients. For a trial vector $\mu$, $F(\mu)=(F_\rho(\mu))_{\rho}$ is the vector of predicted patch energies for each $\rho$, and $y$ is the vector of estimated transition energies obtained from the experiments.  We choose $\mu$ by minimizing the squared discrepancy
\begin{equation}
 \mathcal L(y,\mu)
 =\frac{1}{2m}\norm{F(\mu)-y}_2^2.
 \label{eq:overview-loss}
\end{equation}
There are two questions to address.  Does a small change in the measured transition energies correspond to a small change in the recovered coefficients?  And can the minimizer be found efficiently from a simple initial guess? The former concerns how errors in the estimated transition energies propagate to the recovered coefficients, whereas the latter concerns the efficiency of classical post-processing. Both questions are governed by the derivatives of $F$.

\paragraph{The strong-field limit.}
It is helpful to first consider $\nu=\infty$.  In this limit the correction from the prethermal change of basis vanishes, and the predicted transition energies depend linearly on the coefficients:  $F_0(\mu)=X\mu$.

For a Pauli term $P_a$, let $\chi(a)$ be its support indicator and write $P_a\sim\beta^\rho$ when its nonidentity Pauli axes agree with the random axes chosen in experiment $\rho$.  A term contributes to the transition at $k$ only if it contains $k$ and its Pauli axes match the chosen field axes. When it contributes, its sign is the parity of the random signs on its support.  Explicitly,
\begin{equation}
 X_{\rho,a}
 =-\ind[k\in\supp(P_a)]\ind[P_a\sim\beta^\rho]
   (-1)^{s^\rho\cdot\chi(a)}.
 \label{eq:overview-Walsh-entry}
\end{equation}
We refer to these parity signs $(-1)^{s^\rho\cdot\chi(a)}$ as \emph{Walsh characters}.  For a fixed local choice of Pauli axes, they form rows of a Walsh-Hadamard matrix.  This observation leads to a useful cancellation: when we average over the random signs and axes, cross terms between distinct
Hamiltonian coefficients vanish. Formally, we can express this cancellation by considering $X^{\mathsf T} X$ as follows:
\begin{equation}
 \Gamma:=\frac1mX^{\mathsf T}X,
 \quad
 \bbE[\Gamma]
 =\Sigma:=\operatorname{diag}(p_a3^{-p_a})_{a\in\mathcal A},
 \label{eq:overview-Gamma}
\end{equation}
where $p_a=|\supp(P_a)|$.  The matrix $\Gamma$ measures how well changes in the coefficients can be distinguished from changes in the strong-field data.  Its mean $\Sigma$ is diagonal, with diagonal entries bounded away from zero because $p_a\le q$.  Since each interaction overlaps only constantly many others, $m=\Theta(\log(n/\delta))$ random patterns per site are enough for $\Gamma$ to remain close to $\Sigma$ with probability at least $1-\delta$.  As a result, we can use the explicit form of $\Sigma$ to show that both \(\Gamma\) and \(\Gamma^{-1}\) have bounded \(\ell^\infty\to\ell^\infty\) operator norms with large probability (see \Cref{thm:Walsh-conditioning}).

\paragraph{Returning to finite field strength.}
We will next use the above observations for the $\nu\to\infty$ limit to study the finite-$\nu$ case, in which case the map $F(\mu)$ is nonlinear.  Let $\mathsf J(\mu)=DF(\mu)$ denote its Jacobian.  We have $\mathsf J(\mu)=X+D\mathcal R(\mu)$.  The derivative bound in \Cref{eq:overview-derivatives}, together with the fact that a coefficient appears only in nearby patches which ensures the sparsity of $D\mathcal R(\mu)$, gives
\begin{equation}
 \norm{\frac1m\mathsf J(\mu)^{\mathsf T}\mathsf J(\mu)-\Gamma}
 _{\infty\to\infty}
 \le C V_R\left(\frac{J}{\nu}+\frac{J^2}{\nu^2}\right),
 \label{eq:overview-Gram-perturbation}
\end{equation}
where $V_R=\bigO(R^d)$. The Hessian of the least-squares loss contains this Jacobian product and a second term involving the second derivatives of $F$.  The second derivative bound in \Cref{eq:overview-derivatives} controls the latter term as well.  Hence a field $\nu\ge\nu_{\rm opt}=\bigO(JV_R)=\bigO(\log^d(1/\eps))$ makes the entire Hessian uniformly close to $\Gamma$. We note that this parameter choice is independent of $n$. 

The same matrix $\Gamma$ is useful in our optimization algorithm.  Rather than updating $\mu$ using the unscaled gradient, we multiply the gradient by $\Gamma^{-1}$, which compensates for the different sensitivities of the coefficients.  This fixed rescaling is the preconditioner in our algorithm.  If $f(y,\mu)=\nabla_\mu\mathcal L(y,\mu)$, one iteration is $\mathcal T_y(\mu)=\mu-\Gamma^{-1}f(y,\mu)$.  This algorithm has the additional advantage that it does not require evaluating the second-order derivatives of $F(\mu)$, and still achieves exponential convergence, as we will see in the theorem below:

\begin{theorem}[Informal version of \Cref{thm:local-inverse}]
Let $m\ge C\log(n/\delta)$, $\nu\ge\nu_{\rm opt}=\bigO(\log^d(1/\eps))$, and $y_*=F(\lambda)$.  With probability at least $1-\delta$ over the random control patterns, there is a constant $r_y>0$ such that, whenever $\norm{y-y_*}_\infty\le r_y$, the loss in \Cref{eq:overview-loss} has a unique minimizer $\mu^\star(y)$ in $\cU=\{\mu:\norm{\mu-\lambda}_\infty\le1\}$. It obeys, for some $C_{\rm stab}>0$ that is independent of $n$,
\begin{equation}
 \norm{\mu^\star(y)-\lambda}_\infty
 \le C_{\rm stab}\norm{y-y_*}_\infty.
\end{equation}
Moreover, $\mathcal T_y$ maps $\cU$ into itself and is a contraction in $\ell^\infty$ with contraction factor at most $1/3$.  Consequently, for every $\mu^{(0)}\in\cU$,
\begin{equation}
 \norm{\mu^{(t)}-\mu^\star(y)}_\infty
 \le3^{-t}\norm{\mu^{(0)}-\mu^\star(y)}_\infty.
\end{equation}
\end{theorem}

\paragraph{Initialization and accuracy.}
The contraction guarantee is useful only if we can start inside $\cU$.  Our normalization $\norm\lambda_\infty\le1$ makes the all-zero vector a valid starting point, answering the initialization question raised above.  We also give an explicit data-dependent linear initializer in \Cref{eq:linear-initializer,eq:initializer-bound}; it can give a better first estimate while remaining inside $\cU$ under the same field and error
conditions.

It remains to account for the two errors in \Cref{eq:overview-data-model}.  We choose the frequency-estimation accuracy and the patch radius so that the experimental and local approximation errors (the latter bounded in \Cref{eq:overview-patch-error}) are both $\bigO(\eps)$.  The stability theorem then converts an $\ell^\infty$ error in the measured transition energies into an $\ell^\infty$ error in the Hamiltonian coefficients.  This gives the sequential learning protocol in \Cref{thm:main}.   Its total evolution time is \(\bigO({n\log^2(n/\delta)}/{\eps})\), and it returns $\widehat\lambda$ satisfying $\norm{\widehat\lambda-\lambda}_\infty\le\eps$.  A sufficient combined field threshold is
\begin{equation}
    \nu\ge\nu_{\min}:=\max\{\nu_{\rm dyn},\nu_{\rm opt}\}=\bigO\!\left(\max\left\{\log^d(1/\eps),\log(1/\eps)\left[1+\log\log(1/\eps)\right]^3\right\}\right)
\end{equation}
which does not grow with the system size.  The total evolution time still depends on the system size because the experiment needs to be repeated for every choice of the defect qubit sequentially. Therefore it is natural to consider a strategy to parallelize.

\subsection{Learning at many sites in parallel}

The sequential (one-defect-at-a-time) procedure repeats essentially the same spectroscopy experiment at each site.  Since the preparation, control, and measurement are all local, it is natural to ask whether many sites can be probed during the same evolution. We choose a set \(\mathcal D\) of defects whose radius-\(R\) neighborhoods are disjoint, apply half-strength fields at the defects and full-strength fields elsewhere in those neighborhoods, and measure all defects simultaneously. The field is zero outside the neighborhoods. As in the single-defect experiment, we compare the measured frequencies with a reference model whose field extends to every qubit. Writing \(Q_j^{\mathcal D,\ell}\) for the local excitation projectors in experiment \(\ell\), this reference model has charge
\begin{equation}
 N_{\mathcal D,\ell}=\sum_{k\in\mathcal D}Q_k^{\mathcal D,\ell}+2\sum_{j\notin\mathcal D}Q_j^{\mathcal D,\ell}.
 \label{eq:overview-parallel-charge}
\end{equation}
The eigenvalue one of \(N_{\mathcal D,\ell}\) is now degenerate:  with one defect, the charge-one sector contains a single state, whereas with several defects, it contains $|\mathcal D|$ states, one for an excitation at each defect.  A charge-preserving effective Hamiltonian may mix these states.  Thus prethermal charge conservation by itself no longer isolates a two-level system at each site.

\paragraph{Utilizing spatial separation.}
The remedy is to place defects far enough apart.  To quantify ``far enough,'' we again introduce the bookkeeping parameter $z$.  We will stop the ADHH recursion at order $n_{\rm par}\le n_*(\nu)$ (the reason for this is explained at the end of \Cref{subsec:parallel-stability}).  For fixed $(\mathcal D,\ell)$, let $\widehat D_{\mathcal D,\ell}(z):=D_{n_{\rm par}}(zH(\lambda))$ be the interaction obtained by stopping the ADHH recursion with $N_{\mathcal D,\ell}$ after $n_{\rm par}$ steps. Holding $n_{\rm par}$ fixed as $z$ varies, we write down the Taylor expansion
\begin{equation}
 \widehat D_{\mathcal D,\ell}(z)=\sum_{p\ge1}z^p\widehat D_{\mathcal D,\ell}^{(p)}.
 \label{eq:overview-parallel-Taylor}
\end{equation}
As discussed for the sequential protocol, the coefficient of $z^p$ contains connected $p$-words (\Cref{def:p-word}) built from $p$ Pauli operators, each having exactly the same support as a local interaction term.  For such a connected $p$-word to connect two regions, those terms must form an overlapping chain; its diameter is therefore at most $pr_0$.  If $L_{\mathcal D}=\min_{k\ne k'\in\mathcal D}\dist(k,k')$, and $P_{\mathcal D}=\ceil{L_{\mathcal D}/r_0}$, then no term of degree $p<P_{\mathcal D}$ can reach two defects.  The terms that can do so start only at degree $P_{\mathcal D}$.  Analyticity on a fixed complex disc and Cauchy's estimate bound the combined contribution of these high-degree terms by
\begin{equation}
 \norm{\sum_{p\ge P_{\mathcal D}}\widehat D_{\mathcal D,\ell}^{(p)}}_{\kappa}\le CJ2^{-P_{\mathcal D}},
 \label{eq:overview-parallel-tail}
\end{equation}
independently of $n$, $|\mathcal D|$, and the defect locations.

To see what this buys us,  we consider product states obtained by independently flipping any subset $A$ of the defects while leaving every other site unflipped, which we denote by $\ket{\Phi_A^{\mathcal{D},\ell}}$.   Let $\mathcal H_{\rm def}$ be the span of these states.  If we retain only degrees below $P_{\mathcal D}$, no effective interaction can involve two defects.  Every interaction term in the effective Hamiltonian also conserves the charge (see \Cref{lem:D-taylor-expansion-commuting-support}).  Consequently, an effective interaction term below degree $P_{\mathcal D}$ cannot couple $\ket{\Phi_A^{\mathcal{D},\ell}}$ to $\ket{\Phi_{A'}^{\mathcal{D},\ell}}$ for some different $A'\subset \mathcal{D}$, because doing so either violates charge conservation or involves flipping two defects. It cannot couple $\ket{\Phi_A^{\mathcal{D},\ell}}$ to a state orthogonal to $\mathcal H_{\rm def}$ either, because doing so requires flipping at least two defects to compensate for the $+2$ charge change induced by flipping a non-defect site. As a result, every state $\ket{\Phi_A^{\mathcal{D},\ell}}$ is an eigenstate of the truncated effective Hamiltonian. These eigenstates belong to the reference model with fields on every qubit. To connect their transition energies to the experiment, we again compare the exterior evolution for the two states of a measured defect. The fields around the other defects act outside its neighborhood. A correction that depends on those fields must connect two separated neighborhoods, so locality makes it small. We carry out this argument below in \Cref{lem:parallel-stopped-dynamics,thm:parallel-Ramsey}; together with the eigenstate calculation of \Cref{lem:parallel-code-factorization}, also proved below, it gives the following result.
\begin{theorem}
\label{thm:overview-parallel-error}
Let $\widehat D_{\mathcal D,\ell}^{(<P_{\mathcal D})}=\sum_{p<P_{\mathcal D}}\widehat D_{\mathcal D,\ell}^{(p)}$. Its restriction to the defect subspace has the form
\begin{equation}
    \widehat D_{\mathcal D,\ell}^{(<P_{\mathcal D})}\big|_{\mathcal H_{\rm def}} =c_{\mathcal D}I +2\sum_{k\in\mathcal D} E_{\mathcal D,\ell;k}^{(<P_{\mathcal D})}Q_k^{\mathcal D,\ell},
\end{equation}
and it does not couple \(\mathcal H_{\rm def}\) to its orthogonal complement. For the experiment, choose integers \(K,P\ge1\), where \(K\) controls the error from truncating the neighborhood of one defect and \(P\) controls the effect of fields near other defects. Assume \(R\ge(K+1)r_0\) and \(L_{\mathcal D}\ge2R+(P+1)r_0\). Prepare each defect in an equal superposition of its two local field eigenstates and all other qubits in their zero-excitation states. If \(S_{\mathcal D,\ell;k}(t)\) denotes the complex Ramsey signal measured at \(k\), then
\begin{equation}
    \abs{S_{\mathcal D,\ell;k}(t) -e^{i\left(\nu+2E_{\mathcal D,\ell;k}^{(<P_{\mathcal D})}\right)t}} \le \frac{C_1JV_R}{\nu} +C_2JV_R|t|\left[(2/3)^{n_{\rm par}}+(P+1)2^{-P}\right] +C_3J|t|2^{-K}.
 \label{eq:overview-parallel-signal-bound}
\end{equation}
for every real \(t\). The additional term \((P+1)2^{-P}\) accounts for fields near the other defects. All constants are independent of \(n,\mathcal D,\ell\), the stopping order, and the actual field strength.
\end{theorem}
For half transition energy accuracy \(\eta\), choose \(n_{\rm par}=K=P=\ceil{C_{\rm par}(d+1)\log(eJ/\eta)}\), \(R=(K+1)r_0\), and \(L_{\mathcal D}\ge2R+(P+1)r_0\). As in the sequential experiment, the field condition \(\nu\ge CJ\max\{V_R,n_{\rm par}[1+\log(n_{\rm par}+1)]^3\}\) keeps the error from the unitary transformation small and permits the chosen number of steps. With sufficiently large constants, all terms in the signal error remain below a small constant through times \(O(1/\eta)\). Thus the required distance between defects grows only logarithmically with \(1/\eta\).

We will show that the same classical Hamiltonian recovery algorithm used in the sequential setting can still be applied here.  We will perform the same iterative optimization procedure to fit the patch half-transition-energy map $F_{\rho,R}^{\rm par}$ (defined in \Cref{eq:parallel-patch-map}) obtained by stopping the ADHH recursion on \(B_R(k)\) at order \(n_{\rm par}\). The transition energies $E_{\mathcal{D},\ell;k}^{(<P_{\mathcal{D}})}$ we extract from the experimentally generated Ramsey signals using robust frequency estimation are good approximations of $F_{\rho,R}^{\rm par}$, for reasons explained below. Under the conditions $R\ge(K+1)r_0$ and $R<L_{\mathcal D}$, \Cref{prop:parallel-patch} shows that the difference between $E_{\mathcal{D},\ell;k}^{(<P_{\mathcal{D}})}$ and $F_{\rho,R}^{\rm par}$ is $\mathcal{O}(J2^{-K})$. Choosing the Taylor coefficient-matching degree $K=\Theta(\log(J/\eps))$ used in the local-approximation, makes this error $\bigO(\eps)$.  Fitting the patch half-transition-energy map $\{F_{\rho,R}^{\rm par}\}_{\rho}$ is sufficient for recovering the Hamiltonian coefficients, because it has the same limiting Jacobian $X$ (in the limit of $\nu\to\infty$) as the original patch map \(F_{\rho,R}(\mu)\) (defined in \Cref{eq:patch-gap}), while the finite-field correction and derivative bounds \Cref{lem:scalar-correction,lem:row-derivatives} used in the conditioning analysis hold uniformly for every stopping order $n_{\rm par}\le n_*(\nu)$. Hence the same stability and convergence analysis \Cref{thm:Walsh-conditioning,lem:Jacobian-perturbation} applies to the new patch map $\{F_{\rho,R}^{\rm par}\}_{\rho}$.

\paragraph{Processing the defects in parallel.}
It remains to cover every site while maintaining the required separation. We color the $d$-dimensional lattice so that sites of one color are separated by $L=\Theta(\log(J/\eps))$.  There are $\bigO(L^d)=\bigO(\log^d(J/\eps))$ color classes.  A single global experiment for one class and one repetition (indexed by $\ell$) produces data for the local Ramsey signal for every defect in that class, while the  random pattern (consisting of field directions and signs) restricted to each patch $B_R(k)$ is the same as in the sequential protocol.  Therefore the Jacobian bound in \Cref{eq:overview-Gram-perturbation} and optimization analysis apply without change. Combining this parallel procedure with the error bound in \Cref{thm:overview-parallel-error} gives the resource bounds stated in \Cref{thm:overview-main-result}, and proved formally in \Cref{thm:main-para}.  In particular, all defects of one color are processed during the same set of evolutions, which removes the linear factor in $n$ from the sequential procedure evolution time.

Let us summarize where the errors enter.  The prethermal approximation relates the physical dynamics to the charge-preserving effective dynamics. Frequency estimation introduces experimental error.  Spatial and Taylor truncations replace the full system by a finite patch and decouple separated defects.  Finally, the finite-$\nu$ derivative bounds control the stability of the inverse problem.  We bound each contribution using local quantities before choosing the parameters.  This is why the required field strength and the stability and convergence guarantees are independent of the system size $n$.

The remainder of the paper develops these steps in order.  We introduce the prethermal construction and its analytic continuation in \Cref{sec:prelim}, describe the Ramsey experiment in \Cref{sec:experiment}, prove the locality reduction in \Cref{sec:local-approx}, analyze the inverse problem in \Cref{sec:optimization}, and establish the parallel protocol in \Cref{sec:parallelization}.

\section{Preliminaries}
\label{sec:prelim}

In this section we will introduce the basic notations and definitions, introduce the ADHH recursion used to construct the prethermal effective Hamiltonian, and extend it to potentially non-Hermitian interactions.

\subsection{Geometrically local Hamiltonians}
\label{subsec:ham-class}

We define a ball of radius $R$ around any site $z\in \Lambda$ to be
\begin{equation}
 B_R(x):=\{z\in\Lambda:\dist(x,z)\le R\},
 \label{eq:ball-prelim}
\end{equation}
where $\dist(\cdot,\cdot)$ denotes the Manhattan distance on the lattice. We also define the maximal volume of a radius-$R$ ball as $V_R:=\max_{x\in\Lambda}|B_R(x)|$, which is upper bounded by $V_R\le(2R+1)^d$.

In order to apply the analysis in \cite{abanin2017rigorous}, we will need the following local norm for interactions across the quantum system.

\begin{definition}[Interaction]
\label{def:interaction-family}
An interaction is a family \(Z=(Z_X)_{X\subseteq\Lambda}\), with \(Z_X=0\) unless \(X\) is nonempty and nearest-neighbor connected. The operator \(Z_X\), called the component with label \(X\), is a complex linear combination of Pauli operators satisfying \(\supp(Z_X)\subseteq X\). The label \(X\) can therefore be larger than the actual support of \(Z_X\). Write \(\mathfrak I_\Lambda\) for the complex vector space of these families. We also use \(Z\) for the associated operator \(\sum_X Z_X\), with each component extended by the identity outside \(X\). An interaction is Hermitian if \(Z_X^\dagger=Z_X\) for every \(X\).
\end{definition}

\begin{definition}[Local norm]
     For an interaction \(Z\) as in \Cref{def:interaction-family}, the local norm is a norm of the family \((Z_X)_X\), defined by
\begin{equation}
 \norm Z_\kappa :=\sup_{x\in\Lambda}\sum_{\substack{X\ni x\\ {X\text{connected}}}}e^{\kappa|X|}\norm{Z_X}.
 \label{eq:local-norm}
\end{equation}
\end{definition}
The operations on interactions below specify the label of each output component. If such a component acts nontrivially on only a subset of its label \(X\), we still use \(X\) in the local norm. For example, a component \(Z_X=A_x\otimes I_{X\setminus\{x\}}\), with \(x\in X\), contributes with weight \(e^{\kappa|X|}\), even when its actual support is \(\{x\}\). Norms without a locality subscript refer to the usual operator norm.

We introduce two parameters $\zeta$ and $D_{\rm ov}$ to respectively describe the maximal number of Hamiltonian terms that act on any given qubit, and the maximal number of Hamiltonian terms that overlap with any given term.
\begin{equation}
    \zeta:=\max_{x\in\Lambda}\abs{\{a:x\in\supp(P_a)\}}, \quad D_{\rm ov}:=\max_a\abs{\{b:\supp(P_a)\cap\supp(P_b)\ne\varnothing\}}.
    \label{eq:zeta-Dov}
\end{equation}
Recall that $q,r_0,d$ are all constants that are independent of $n$. Consequently, these parameters are  upper bounded by $\zeta\leq  \mathcal O(3^q(2r_0+1)^{d(q-1)})=\mathcal{O}(1)$, $D_{\rm ov} \le q\zeta = \mathcal{O}(1)$, and are independent of $n$.

For each Pauli term, choose \(x_a\in\supp(P_a)\) and set \(X_a:=B_{r_0}(x_a)\). Since \(\Lambda\) is a box, \(X_a\) is nearest-neighbor connected, contains \(\supp(P_a)\), and has at most \((2r_0+1)^d\) sites. Since \(x\in X_a\) implies \(x_a\in B_{r_0}(x)\),
\begin{equation}
    \widetilde\zeta:=\max_{x\in\Lambda}|\{a:x\in X_a\}|\le V_{r_0}\zeta\le(2r_0+1)^d\zeta=\mathcal O(1).
\end{equation}
Fix these labels once and represent \(H(\mu)\) by the interaction \(H_X(\mu):=\sum_{a:X_a=X}\mu_aP_a\). We therefore have
\begin{equation}
    \label{eq:Ham-local-norm-bound}
    \|H\|_{\kappa_0}  {\leq}\sup_{x\in\Lambda}\sum_{a:{X_a}\ni x}e^{\kappa_0{|X_a|}}\leq {\widetilde\zeta} e^{\kappa_0 {(2r_0+1)^d}}=: J=\bigO(1),
\end{equation}
where $\kappa_0>0$ is an arbitrarily chosen constant. With the above definition one can readily show, for any coefficient vector $\xi=(\xi_a)_{a\in\mathcal{A}}$,
\begin{equation}
    \norm{\sum_a\xi_aP_a}_{\kappa_0}\le J\norm\xi_\infty.
 \label{eq:coeff-to-local}
\end{equation}

\subsection{The prethermal effective Hamiltonian}
\label{sec:prethermal-effective-Ham}

In this section, we will introduce the effective Hamiltonian in \cite{abanin2017rigorous} that describes the prethermal behavior of quantum systems. Our learning methods requires applying static field $\nu N$ , where the $N$ can be understood as an artificial \emph{charge operator}. It helps us decompose the Hilbert space into a direct sum of its eigenspaces, each corresponding to an integer eigenvalue. For a Hamiltonian $H$ that is locally small compared to $\nu$, \cite{abanin2017rigorous} introduces the recursion, which we shall refer to as the \emph{ADHH recursion}, to construct an effective Hamiltonian whose leading interaction commutes with \(N\). For a field supported in a finite region, we will prove below in \Cref{thm:ADHH-dynamics} that the dynamics generated by \(\nu N+D_s\) approximates the physical dynamics. The input $H$ can be our original Hamiltonian in \Cref{eq:intro-H}. In \Cref{subsec:uniform-order}, we denote a generic, possibly non-Hermitian, input by $Z$ when extending the ADHH recursion to a complex domain. The total Hamiltonian and the charge operator are respectively
\begin{equation}
 G=\nu N+H, \quad
 N=\sum_xN_x, \quad
 \sigma(N_x)\subset\mathbb Z.
 \label{eq:generic-G}
\end{equation}

To define the ADHH recursion, we need the following superoperator that extracts the charge-$\omega$ component of an operator.
\begin{equation}
    \mathcal P_{N,\omega}(W) :=\frac{1}{2\pi}\int_0^{2\pi} e^{-i\omega\theta}e^{i\theta N}We^{-i\theta N}\,d\theta.
    \label{eq:charge-projection}
\end{equation}
Following the averaging notation of~\cite{abanin2017rigorous}, define $\langle W\rangle_N:=\mathcal P_{N,0}(W).$ Since conjugation with $N$ does not increase the support, we have the following inequalities 
\begin{equation}
    \norm{\langle W\rangle_N}_{\kappa_0}\le\norm W_{\kappa_0}, \quad \norm{W-\langle W\rangle_N}_{\kappa_0}\le2\norm W_{\kappa_0}.
    \label{eq:projection-norms}
\end{equation}
With this definition we initialize the ADHH recursion with
\begin{equation}
    D:=\langle H\rangle_N, \quad V:=H-D.
\label{eq:DV-generic}
\end{equation}
In other words, $D$ is the part of $H$ that maps within each eigenspace of $N$, while $V$ is the part that maps across eigenspaces of $N$. This definition ensures $[D,N]=0$. We also define the following $\mathcal S_N(\cdot)$ as an inverse of $[\cdot,\nu N]$ on the charge-changing subspace, with $T:=2\pi/\nu$:
\begin{equation}
    \mathcal S_N(V) :=\sum_{\omega\ne0}\frac{\mathcal P_{N,\omega}(V)}{\nu\omega} = -\frac{i}{T}\int_0^T\int_0^t e^{is\nu N} V e^{-is\nu N} ds\,dt, \quad [\mathcal S_N(V),\nu N]=-V.
    \label{eq:homological-inverse}
\end{equation}

With the above notation, we are ready to define the ADHH recursion. Starting from $H_0=H$, each ADHH step constructs $D_j,V_j,A_j$, and $H_{j+1}$.
\begin{equation}
\label{eq:ADHH-recursion}
   D_j=\langle H_j\rangle_N,\quad V_j=H_j-D_j,\quad A_j=\mathcal S_N(V_j), \quad H_{j+1}=e^{A_j}(\nu N+H_j)e^{-A_j}-\nu N.
\end{equation}

At the interaction level, the charge projections and \(\mathcal S_N\) act componentwise, and \(N\) has singleton components \(N_{\{x\}}=N_x\). Commutators and conjugations are defined as follows to retain the component information
\begin{align}
    ([A,W])_S&:=\sum_{X\cup Y=S}[A_X,W_Y],\quad \operatorname{ad}_A(W):=[A,W]\\
    e^AWe^{-A}&:=e^{\operatorname{ad}_A}W =\sum_{m\ge0}\frac{\operatorname{ad}_A^m(W)}{m!}.
    \label{eq:interaction-operations}
\end{align}
Commutators with disjoint labels vanish, so all nonzero labels remain connected. Summing the resulting components gives the usual operator operations. All local-norm bounds below use this interaction recursion. Our $A_j$ has the opposite sign from the generator used in the conjugation convention of Ref.~\cite{abanin2017rigorous}; this lets us retain $[A_j,\nu N]=-V_j$ and $e^{A_j}(\cdot)e^{-A_j}$ consistently throughout. A sequence of local norms $\|\cdot\|_{\kappa_j}$ is also defined, with a slowly decreasing locality parameter
\begin{equation}
\label{eq:kappa-j-defn}
    \kappa_j=\kappa_0/[1+\log(j+1)],
\end{equation}
which allows a controlled loss of locality at each conjugation step. The resulting quasi-local $H_j$, for large $j$, offers an effective description of the prethermal dynamics, and gives rise to approximate conservation laws that are valid for long time.

In the above, viewing the Hilbert space as a direct sum of eigenspaces of \(N\), \(D_j\) extracts the block-diagonal part of \(H_j\), while \(V_j\) extracts the off-diagonal part.  It is shown in \cite[Eq.~(4.2) and Eq.~(4.10)]{abanin2017rigorous} that,
\begin{equation}
 \norm{V_j}_{\kappa_j}\le CJ\left(2/3\right)^j, \quad
 \norm{A_j}_{\kappa_j}\le {C\pi J(2/3)^j}/{\nu},
\end{equation}
where \(C>0\) is a constant.  Thus both \(V_j\) and \(A_j\) shrink geometrically in the running local norm, with \(A_j\) carrying the additional small factor \(\nu^{-1}\).  The reason for defining \(A_j\) in this way is that it solves the equation $[A_j,\nu N]=-V_j$, so that the leading off-diagonal term is canceled under conjugation.  Indeed, we can informally study this effect using the Baker–Campbell–Hausdorff formula, which gives
\begin{align}
    H_{j+1}+\nu N
    &=e^{A_j}(\nu N+H_j)e^{-A_j} \\
    &\approx \nu N+D_j+V_j+[A_j,\nu N]+[A_j,H_j] \\
    &=\nu N+D_j+[A_j,H_j].
\end{align}
The original off-diagonal term \(V_j\) is therefore removed, and the new off-diagonal part is generated by commutators involving \(A_j\). Since \(\norm{A_j}_{\kappa_j}\lesssim \nu^{-1}\norm{V_j}_{\kappa_j}\) and \(\norm{H_j}_{\kappa_j}\lesssim J\), these commutator terms are smaller than \(V_j\) by a factor of order \(J/\nu\), up to the allowed locality losses.  The rigorous analysis in \cite{abanin2017rigorous} shows this leads to the geometric contraction \(\norm{V_{j+1}}_{\kappa_{j+1}}\le (2/3)\norm{V_j}_{\kappa_j}\).

An equivalent way to define the ADHH recursion is as follows: We define
\begin{equation}
 \gamma_j(W):=e^{A_j}We^{-A_j}, \quad
 \alpha_j(W):=\int_0^1e^{tA_j}We^{-tA_j}\,dt,
 \label{eq:gamma-alpha}
\end{equation}
then the ADHH recursion in \Cref{eq:ADHH-recursion} can be written as
\begin{equation}
    H_{j+1}=\gamma_j(D_j)+\bigl(\gamma_j(V_j)-V_j\bigr)-\bigl(\alpha_j(V_j)-V_j\bigr).
    \label{eq:recursion-no-N}
\end{equation}

We will use the following lemma from \cite{abanin2017rigorous} in our analysis:
\begin{lemma}[ADHH local conjugation estimate;~{\cite[Lemma 4.1]{abanin2017rigorous}}]\label{lem:ADHH-conjugation}
Let \(\kappa_j>\kappa_{j+1}>0\), set \(\delta\kappa_j=\kappa_j-\kappa_{j+1}\), and let \(Q,Z\) be interactions.  If $ 3\norm Q_{\kappa_j}\le\delta\kappa_j$, then
\begin{equation}
 \norm{e^QZe^{-Q}-Z}_{\kappa_{j+1}} \le
 \frac{18\norm Z_{\kappa_j}\norm Q_{\kappa_j}}{\delta\kappa_j\kappa_{j+1}}.
 \label{eq:ADHH-conjugation-bound}
\end{equation}
\end{lemma}
Thus, the local change induced by conjugation is linear in the local norm of conjugation interaction $Q$. We will use this inequality to establish the analytic continuation of the ADHH recursion in~\Cref{sec:analytic-cont}.

We use ``ADHH recursion'' in \Cref{eq:ADHH-recursion} for the iterative construction of $H_j,D_j,V_j,A_j$. We reserve ``final ADHH decomposition'' of \Cref{thm:ADHH-normal-form} for the transformed Hamiltonian $Y(\nu N+H)Y^\dagger=\nu N+\widehat D+\widehat V$ obtained at step $n_*$. The theorem below guarantees this final decomposition. First, we define the quantities
\begin{equation}
    \nu_0:=\frac{54\pi \bigl(\norm D_{\kappa_0}+2\norm V_{\kappa_0}\bigr)}{\kappa_0^2}, \quad
    n_*:=\floor{\frac{\nu/\nu_0}{[1+\log(\nu/\nu_0)]^3}}-2.
    \label{eq:ADHH-nu0-nA}
\end{equation}
Note that $\norm D_{\kappa_0}+2\norm V_{\kappa_0}=\mathcal{O}(J)$, where $J$ is our upper bound for $\|H\|_{\kappa_0}$ as defined in \Cref{eq:Ham-local-norm-bound}.

\begin{theorem}[Final ADHH decomposition;~{\cite[Theorem 3.1]{abanin2017rigorous}}]\label{thm:ADHH-normal-form} 
Assume \(\nu\ge\nu_0\), \(\nu\ge9\pi\norm V_{\kappa_0}/\kappa_0\) and \(n_*\ge1\). Then for \(G=\nu N+D+V\), with \([D,N]=0\), there exist Hermitian interactions \(\widehat H\), \(\widehat D\), and \(\widehat V\), and a unitary \(Y\), such that
\begin{equation}
\label{eq:ADHH-decomposition-final-order}
    Y(\nu N+H)Y^\dagger = \nu N+\widehat H = \nu N+\widehat D+\widehat V,
\end{equation}
with the following properties:
\begin{enumerate}[label=(\roman*)]
\item For $\widehat{D}$ we have $\widehat D=\langle\widehat H\rangle_N$, and $[\widehat D,N]=0$.

\item Define the decay rate $\kappa_{n_*}:=\kappa_0\bigl(1+\log(n_*+1)\bigr)^{-1}$. Then
\begin{equation}
 \norm{\widehat D-D}_{\kappa_{n_*}} \le {C\nu_0}/{\nu}, \quad
 \norm{\widehat V}_{\kappa_{n_*}} \le (2/3)^{n_*} \norm V_{\kappa_0}.
\end{equation}

\item The unitary \(Y\) is close to the identity and quasilocal in
the sense that, for every interaction \(Z\),
\begin{equation}
 \norm{YZY^\dagger-Z}_{\kappa_{n_*}} \le {C\nu_0\norm Z_{\kappa_0}}/{\nu}.
\end{equation}
\end{enumerate}
\end{theorem} 

In the above the Hermitian interactions  $\widehat{H},\widehat{D},\widehat{V}$ and unitary $Y$ can be explicitly constructed from  the ADHH recursion: 
\begin{align}\label{eq:ADHH-expression}
    \widehat{H} = H_{n_*},\quad \widehat{D} = D_{n_*},\quad \widehat{V} = V_{n_*},\quad Y = e^{A_{n_*-1}}\cdots e^{A_1}e^{A_0}.
\end{align}
We will often refer to $\nu N+\widehat{D}$ or $\widehat{D}$ as the \emph{effective Hamiltonian}.

We will also use the version of Lie-Robinson bound in~\cite{abanin2017rigorous}:
\begin{lemma}[ADHH Lieb-Robinson estimate;~{\cite[Lemma 5.1]{abanin2017rigorous}}]
\label{lem:ADHH-LR}
Let \(Z=\sum_SZ_S\) be Hermitian with \(\norm Z_{2\kappa}<\infty\).  If \(A\) and \(B\) are supported in \(X\) and \(Y\), respectively, then
\begin{equation}
    \norm{[A,e^{itZ}Be^{-itZ}]} \le{2}\norm A\norm B e^{-\kappa\left[\dist(X,Y)-v|t|\right]}\min\{|X|,|Y|\},
    \label{eq:ADHH-LR}
\end{equation}
where $v=C(d)\kappa^{-(d+2)}e^{\kappa}\norm Z_{2\kappa}$. The constant depends only on the spatial dimension.
\end{lemma}

The next theorem compares the propagators generated by \(H+\nu N_B\) and \(\nu N_B+D_s\), and then deduces the corresponding estimate for observables. We use the ADHH recursion and its local-norm estimates from \cite{abanin2017rigorous}, but avoid using \cite[Theorem~3.3]{abanin2017rigorous}.

\begin{theorem}
\label{thm:ADHH-dynamics}
Let \(N_B\) be an onsite integer-valued charge vanishing outside \(B\). Fix \(\kappa_0>0\), let \(H\) be Hermitian with \(\norm H_{\kappa_0}\le\Jan\). As stated below in \Cref{thm:uniform-normal-form}, for $\bar\nu_0$ and $n_*(\nu)$ defined in \Cref{eq:uniform-nu0,eq:new-n*} respectively, we assume \(\nu\ge\max\{\bar\nu_0,54\pi\Jan/\kappa_0\}\) and \(n_*(\nu)\ge1\).  Fix an integer \(1\le s\le n_*(\nu)\) and  let \(D_j,V_j,A_j\) be generated by the ADHH recursion \Cref{eq:ADHH-recursion} with input \(H\) and charge \(N_B\). We also define \(Y_s=e^{A_{s-1}}\cdots e^{A_0}\). Put \(G:=H+\nu N_B\) and \(G_0:=\nu N_B+D_s\). Then, for every real \(t\):
\begin{enumerate}[label=(\roman*)]
    \item The propagators satisfy
        \begin{equation}
            \norm{e^{-iGt}-e^{-iG_0t}} \le C\Jan|B|\left(\nu^{-1}+|t|(2/3)^s\right).
            \label{eq:ADHH-propagator-comparison}
        \end{equation}
        
    \item For every bounded operator \(O\),
        \begin{equation}
        \norm{e^{iGt}Oe^{-iGt}-e^{iG_0t}Oe^{-iG_0t}}\le C\norm O\Jan|B|\left(\nu^{-1}+|t|(2/3)^s\right).
            \label{eq:ADHH-local-dynamics}
        \end{equation}
\end{enumerate}
The constants $C$ are independent of \(B,\nu,t,s\), and the system size.
\end{theorem}

\begin{proof}
We use the component labels from \Cref{def:interaction-family} and the operations in \Cref{eq:interaction-operations}. Fix \(X\) with \(X\cap B=\varnothing\). Since \((H_j)_X\) is supported within \(X\), it commutes with \(N_B\). Consequently,
\begin{equation}
    (D_j)_X =\mathcal P_{N_B,0}\bigl((H_j)_X\bigr) =\frac1{2\pi}\int_0^{2\pi}e^{i\theta N_B}(H_j)_Xe^{-i\theta N_B}\,d\theta=(H_j)_X.    
\end{equation}
It follows that \((V_j)_X=0\). The map \(\mathcal S_{N_B}\) in \Cref{eq:homological-inverse} also acts componentwise, so \((A_j)_X=\mathcal S_{N_B}((V_j)_X)=0\). Thus every nonzero component of \(A_j\) or \(V_j\) has a label intersecting \(B\). For any interaction \(W\) with this property,
\begin{equation}
    \norm W
    \le \sum_{X:X\cap B\ne\varnothing}\norm{W_X}
    \le \sum_{x\in B}\sum_{X\ni x}\norm{W_X}
    \le |B|\norm W_{\kappa_j}.
    \label{eq:finite-ball-global-norm}
\end{equation}
We now use the bounds on \(A_j\) and \(V_j\) stated below in \Cref{eq:uniform-A-bound,eq:uniform-V-bound}. Their proof is given below in Appendix~\ref{app:prelim-proofs} from the ADHH recursion and local conjugation estimate, without using the present dynamical comparison; the forward reference is therefore not circular. Since \(A_j^\dagger=-A_j\), integration of \(e^{uA_j}A_j\) over \(0\le u\le1\) gives \(\norm{e^{A_j}-I}\le\norm{A_j}\). Telescoping the product defining \(Y_s\), using unitarity of each factor, and summing \(\sum_{j<s}(2/3)^j\le3\), we obtain
\begin{equation}
    \norm{Y_s-I}\le\sum_{j<s}\norm{A_j}\le C\Jan|B|/\nu,
    \qquad \norm{V_s}\le C\Jan|B|(2/3)^s.
    \label{eq:finite-ball-dressing}
\end{equation}

Iterating \Cref{eq:ADHH-recursion} gives \(Y_s GY_s^\dagger=G_0+V_s\), hence
\begin{equation}
    e^{-iGt}=Y_s^\dagger e^{-i(G_0+V_s)t}Y_s.    
\end{equation}
Duhamel's formula gives
\begin{equation}
    e^{-i(G_0+V_s)t}-e^{-iG_0t}=-i\int_0^t e^{-i(G_0+V_s)(t-u)}V_s e^{-iG_0u}\,du.    
\end{equation}
The propagators in the integrand are unitary, so the norm of this difference is at most \(|t|\norm{V_s}\). To bound the effect of \(Y_s\), write
\begin{equation}
    e^{-iGt}-e^{-iG_0t} = Y_s^\dagger\bigl(e^{-i(G_0+V_s)t}-e^{-iG_0t}\bigr)Y_s+(Y_s^\dagger-I)e^{-iG_0t}Y_s+e^{-iG_0t}(Y_s-I).  
\end{equation}
Unitary invariance of the norm gives
\begin{equation}
    \norm{e^{-iGt}-e^{-iG_0t}} \le |t|\norm{V_s}+2\norm{Y_s-I}\le C\Jan|B|\left(\nu^{-1}+|t|(2/3)^s\right),    
\end{equation}
which proves (i). For (ii), put \(U=e^{-iGt}\) and \(U_0=e^{-iG_0t}\). The identity
\begin{equation}
    U^\dagger OU-U_0^\dagger OU_0=(U^\dagger-U_0^\dagger)OU+U_0^\dagger O(U-U_0)
\end{equation}
implies \(\norm{U^\dagger OU-U_0^\dagger OU_0} \le2\norm O\norm{U-U_0}\). Combining this with (i) proves (ii), after absorbing the numerical factor into \(C\).
\end{proof}

\subsection{Analytic continuation of the ADHH recursion}
\label{sec:analytic-cont}

We will extend the ADHH recursion beyond Hermitian Hamiltonian $H$. This extension is needed in our algorithm analysis.  Specifically, we replace the Hermitian $H$ by a possibly non-Hermitian $Z$. The resulting bounds allow us to employ Cauchy integrals to bound high-order derivatives. We define the set of coefficient vectors for \Cref{eq:intro-H} that we want to study 
\label{subsec:uniform-order}
\begin{equation}
    \cU:=\{\mu\in\mathbb R^M:\norm{\mu-\lambda}_\infty\le 1\}.
    \label{eq:cU-defn}
\end{equation}
By \Cref{eq:coeff-to-local} we therefore have
\begin{equation}
    \norm{H(\mu)}_{\kappa_0}\le2J,
    \quad \mu\in\cU.
    \label{eq:real-basin-local}
\end{equation}
We will study analyticity in $Z$ in a complex ball defined below
\begin{equation}
    \Jan:=4J, \quad
    \mathfrak B_{\rm an}:=\{Z:\norm Z_{\kappa_0}\le\Jan\}.
    \label{eq:analytic-ball}
\end{equation}
The initial $\nu_0$ in~\Cref{thm:ADHH-normal-form} is therefore replaced by 
\begin{equation}
    {\bar\nu_0}:={270\pi\Jan}/{\kappa_0^2}.
    \label{eq:uniform-nu0}
\end{equation}
and $n_*$ is replaced by
\begin{equation}
\label{eq:new-n*}
    n_*:=\floor{\frac{\nu/{\bar\nu_0}}{[1+\log(\nu/{\bar\nu_0})]^3}}-2
\end{equation}
We choose the field strength $\nu$  to satisfy {\(\nu\ge\bar\nu_0\),} $n_* \ge 1$, and $\nu \ge 54\pi J_{\rm an}/\kappa_0$. We will use these parameters uniformly for all $Z\in\mathfrak B_{\rm an}$. We show that the local norm of $D_j,V_j,A_j$ is still bounded with the same expression even if the initial interaction $H_0=Z$ in the ADHH recursion \Cref{eq:ADHH-recursion} is non-Hermitian.

\begin{theorem}[Uniform analytic bounds for the ADHH recursion]
\label{thm:uniform-normal-form}
We assume $\nu>0$ satisfies {\(\nu\ge\bar\nu_0\)}, 
\begin{equation}
\label{eq:nu-requirement}
    \nu \ge 54\pi J_{\rm an}/\kappa_0,
\end{equation}
and $n_* \ge 1$, with $n_*$ defined in \Cref{eq:new-n*}. The ADHH recursion \Cref{eq:ADHH-recursion} is well defined for every \(Z\in\mathfrak B_{\rm an}\), including non-Hermitian \(Z\), and obeys
\begin{align}
    \norm{D_j}_{\kappa_j}&\le C\Jan,\label{eq:uniform-D-bound}\\
    \norm{V_j}_{\kappa_j}&\le C\Jan(2/3)^j,\label{eq:uniform-V-bound}\\
    \norm{A_j}_{\kappa_j}&\le {C\Jan(2/3)^j}/{\nu},
    \label{eq:uniform-A-bound}
\end{align}
where $0\le j\le n_*$. The maps \(Z\mapsto D_j(Z),V_j(Z),A_j(Z)\) are analytic on the interior of \(\mathfrak B_{\rm an}\).  
\end{theorem}

The proof of this theorem largely follows the original proof in \cite{abanin2017rigorous} for the Hermitian case, and is provided in Appendix~\ref{app:prelim-proofs} for completeness. We note that these estimates do not necessarily hold beyond $j=n_{*}$. This is because the locality loss parameter $\delta\kappa_j$ in~\Cref{lem:ADHH-conjugation} decreases with $j$, and when $j$ is too large we will not be able to guarantee $3\|A_j\|_{\kappa_j}\leq \delta\kappa_j$. We will prove the above analytic continuation of the ADHH recursion in estimates in~Appendix~\ref{app:prelim-proofs}. The key is the local conjugation estimate in~\Cref{lem:ADHH-conjugation}, whose proof does not require Hermiticity.

The analyticity of the quantities involved, in particular that of $D_j$, and the above bounds will allow us to use the Cauchy integral estimates below (see \cite[Chap.~I, Sec.~1, Thm.~1.6]{range1998holomorphic}.):

\begin{lemma}[Cauchy integral estimates]
\label{lem:Cauchy-general}
Let \(f\) be analytic on a neighborhood of the closed disc \(\abs z\le r\).  Then
\begin{equation}
    f^{(p)}(0)=\frac{p!}{2\pi i}\oint_{|z|=r}\frac{f(z)}{z^{p+1}}\,dz,
    \quad
    |f^{(p)}(0)|\le\frac{p!}{r^p}\sup_{|z|=r}|f(z)|.
    \label{eq:Cauchy-one-variable}
\end{equation}
More generally, if \(f(z_1,\ldots,z_s)\) is analytic on a polydisc \(|z_j|\le r_j\), then
\begin{equation}
    \abs{\partial_1^{p_1}\cdots\partial_s^{p_s}f(0)}
    \le \frac{p_1!\cdots p_s!}{r_1^{p_1}\cdots r_s^{p_s}} \sup_{|z_j|\le r_j}|f(z_1,\ldots,z_s)|.
    \label{eq:Cauchy-polydisc}
\end{equation}
\end{lemma}

\section{Experiment}
\label{sec:experiment}

We will run experiments to gather a set of transition-energy estimates (which we shall define in \Cref{eq:full-gap}) for different applied static fields. We call $k$ a \emph{defect site} because the applied static field at $k$ has half the strength used on the other sites of \(B=B_R(k)\). We also set the field to be zero outside \(B\). For each defect site \(k\in\Lambda\) and repetition \(\ell=1,\ldots,m\), let
\begin{equation}
    \rho=(k,\ell)
    \label{eq:row-index}
\end{equation}
index one transition-energy estimate, which will correspond to an element in the target vector in the optimization procedure to be introduced in~\Cref{sec:optimization}.  Independently for repetition index $\ell$ and site $k$, draw
\begin{equation}
    s_j^\rho\in\{0,1\}, \quad \beta_j^\rho\in\{X,Y,Z\}
    \label{eq:random-pattern}
\end{equation}
uniformly.  These parameters allow us to fix a signed local Pauli operator for qubit $j$,
\begin{equation}
    \tau_{j,z}^\rho:=(-1)^{s_j^\rho}\sigma_j^{\beta_j^\rho}.
    \label{eq:tau-z}
\end{equation}
We will view this operator as the new Pauli-$z$ operator for a changed local basis, hence the $z$ subscript in its notation. Choose eigenvectors \(\ket{\phi_{j,0}^\rho}\), \(\ket{\phi_{j,1}^\rho}\) satisfying
\begin{equation}
    \tau_{j,z}^\rho\ket{\phi_{j,b}^\rho}=(-1)^b\ket{\phi_{j,b}^\rho},
    \quad b\in\{0,1\},
    \label{eq:local-eigenstates}
\end{equation}
and set
\begin{equation}
    Q_j^\rho:=\frac{I-\tau_{j,z}^\rho}{2}.
    \label{eq:Q-row}
\end{equation}
With the notation introduced above, we can define the static control that we will apply:
\begin{equation}
    H_{\rm ctrl}^\rho=\sum_{j\in B\setminus\{k\}}\tau_{j,z}^\rho+\frac{\tau_{k,z}^\rho}{2},
 \label{eq:Hctrl}
\end{equation}
so the physical Hamiltonian corresponding to transition-energy estimate indexed by \(\rho\) is
\begin{equation}
    H_{\rm tot}^\rho(\mu)=H(\mu)-\nu H_{\rm ctrl}^\rho.
    \label{eq:Htot}
\end{equation}
We also use the reference Hamiltonian obtained by extending the field to every qubit, with half strength at \(k\) and full strength at every \(j\ne k\). Its charge is \(N_\rho\), defined below. This is an auxiliary model for defining the transition energy to be reconstructed; the experimental field remains confined to \(B\):
\begin{equation}
    N_\rho:=Q_k^\rho+2\sum_{j\ne k}Q_j^\rho.
    \label{eq:N-rho}
\end{equation}

Define the physical charge \(N_\rho^B\) and the Hamiltonians \(G_\rho^B,G_\rho\) by
\begin{align}
    N_\rho^B&:=Q_k^\rho+2\sum_{j\in B\setminus\{k\}}Q_j^\rho,\\
    G_\rho^B(\mu)&:=H(\mu)+\nu N_\rho^B=H_{\rm tot}^\rho(\mu)+\nu(|B|-1/2)I,\\
    G_\rho(\mu)&:=H(\mu)+\nu N_\rho.
    \label{eq:G-rho}
\end{align}
The equality for \(G_\rho^B\) shows that it generates the same observable dynamics as \(H_{\rm tot}^\rho\). The Hamiltonian \(G_\rho\) describes
the reference field on all of \(\Lambda\).

For this auxiliary integer-valued charge $N_\rho$, the \(N_\rho=0\) and \(N_\rho=1\) eigenspaces are one dimensional.  Their normalized states are
\begin{equation}
\label{eq:Phi0_Phi1}
    \ket{\Phi_0^\rho}=\bigotimes_{j\in\Lambda}\ket{\phi_{j,0}^\rho},
    \quad  \ket{\Phi_1^\rho}=\ket{\phi_{k,1}^\rho}\bigotimes_{j\ne k}\ket{\phi_{j,0}^\rho}.
\end{equation}
Consequently,
\begin{equation}
    \ket{\Phi_+^\rho}:=\frac{\ket{\Phi_0^\rho}+\ket{\Phi_1^\rho}}{\sqrt2}=\frac{\ket{\phi_{k,0}^\rho}+\ket{\phi_{k,1}^\rho}}{\sqrt2}\bigotimes_{j\ne k}\ket{\phi_{j,0}^\rho}
    \label{eq:Phi-plus}
\end{equation}
is a product state. The charge-one eigenspace of \(N_\rho\) is one dimensional. In contrast, because \(N_\rho^B\) vanishes outside \(B\), each of its charge-zero and charge-one eigenspaces includes arbitrary states on \(\Lambda\setminus B\). The effect of those exterior states will be treated below in \Cref{thm:Ramsey-approximation}.

\subsection{Local Ramsey signal}
\label{subsec:Ramsey-signal}

Similar to a Ramsey interferometry experiment, we will introduce the observables and initial states to generate what we call the \emph{local Ramsey signal} corresponding to a defect site $k$, which enables us to extract effective transition-energy information through robust frequency estimation. On the defect qubit $k$ define the two single-qubit Hermitian operators
\begin{equation}
\label{eq:X-Y-quadrature}
    \tau_{k,x}^\rho:=\ket{\phi_{k,0}^\rho}\!\bra{\phi_{k,1}^\rho}+\ket{\phi_{k,1}^\rho}\!\bra{\phi_{k,0}^\rho}, \quad
    \tau_{k,y}^\rho:=i\ket{\phi_{k,0}^\rho}\!\bra{\phi_{k,1}^\rho}-i\ket{\phi_{k,1}^\rho}\!\bra{\phi_{k,0}^\rho}.
\end{equation}
They are single-qubit Pauli observables that correspond to Pauli-$x$ and $y$ operators in a $\rho$-dependent frame.  With
\begin{equation}
    \tau_{k,+}^\rho:=\tau_{k,x}^\rho+i\tau_{k,y}^\rho=2\ket{\phi_{k,1}^\rho}\!\bra{\phi_{k,0}^\rho},
    \label{eq:O-plus}
\end{equation}
we have \(\tau_{k,+}^\rho\ket{\Phi_0^\rho}=2\ket{\Phi_1^\rho}\) and \(\tau_{k,+}^\rho\ket{\Phi_1^\rho}=0\).

Each experiment for a transition-energy estimate indexed $\rho$ involves preparing \(\ket{\Phi_+^\rho}\), evolving for time \(t\) under \(H_{\rm tot}^\rho(\lambda)\), and measuring \(\tau_{k,x}^\rho\) or \(\tau_{k,y}^\rho\). These experiments allow us to estimate, by averaging over measurement results,
\begin{equation}
    S_\rho(t):=\bra{\Phi_+^\rho}e^{iG_\rho^B(\lambda)t}\tau_{k,+}^\rho e^{-iG_\rho^B(\lambda)t}\ket{\Phi_+^\rho}.
    \label{eq:signal}
\end{equation}
Here \(H_{\rm tot}^\rho(\lambda)\) is replaced by \(G_\rho^B(\lambda)\) because the scalar shift does not change the signal.

Since the eigenspaces of \(N_\rho\) with eigenvalues zero and one are one dimensional, an operator commuting with \(N_\rho\) acts as a scalar on each of them. Fix \(1\le s\le n_*(\nu)\), run \Cref{eq:ADHH-recursion} with charge \(N_\rho\) and input \(H(\mu)\) for exactly \(s\) steps, and denote the resulting \(D_s\) by \(\widehat D_\rho^\Lambda(\mu)\). The integer \(n_*(\nu)\) is the maximum admissible order; the chosen \(s\) is held fixed as \(\mu\) varies. We define the half transition energy of \(\widehat D_\rho^\Lambda(\mu)\) between \(\ket{\Phi_0^\rho}\) and \(\ket{\Phi_1^\rho}\) by
\begin{equation}
    E_\rho^{\Lambda}(\mu) :=\frac12\left( \bra{\Phi_1^\rho}\widehat D_\rho^{\Lambda}(\mu)\ket{\Phi_1^\rho} -\bra{\Phi_0^\rho}\widehat D_\rho^{\Lambda}(\mu)\ket{\Phi_0^\rho} \right). \label{eq:full-gap}
\end{equation}
The quantity $2E_\rho^{\Lambda}(\mu)$ is the effective transition energy between $\ket{\Phi_0^\rho}$ and $\ket{\Phi_1^\rho}$. Indeed, both states are eigenstates of $\widehat D_\rho^{\Lambda}(\mu)$ because this interaction commutes with $N_\rho$ and the relevant charge sectors are non-degenerate.

We define the functional used to extract the effective half transition energy from $\widehat D_\rho^{\Lambda}(\mu)$, which is valid for any operator $W$,
\begin{equation}
    \ell_\rho(W):=\frac12\left(\bra{\Phi_1^\rho}W\ket{\Phi_1^\rho}-\bra{\Phi_0^\rho}W\ket{\Phi_0^\rho}\right).
    \label{eq:transition-functional}
\end{equation}
Because both $\ket{\Phi_0^\rho}$ and $\ket{\Phi_1^\rho}$ are eigenstates of $N_\rho$, the functional vanishes on operators that map one charge sector to another. Note that because $\ket{\Phi_1^\rho} = \tau_{k,x}^\rho\ket{\Phi_0^\rho}$,
\begin{equation}
    \label{eq:transition-functional-Xk}
    \ell_\rho(W) =\frac12 \bra{\Phi_0^\rho}( \tau_{k,x}^\rho W  \tau_{k,x}^\rho-W)\ket{\Phi_0^\rho}.
\end{equation}
The functional therefore also vanishes on every term in $W$ whose support does not contain $k$. This fact provides us a way to upper bound $|\ell_\rho(W)|$ using the local norm of $W$, for any operator $W$:
\begin{equation}
    |\ell_\rho(W)| \le\sum_{X\ni k}\norm{W_X} \le\norm W_\kappa.
    \label{eq:ell-local-bound-general-main}
\end{equation}
As a special case, we then have, for any $\mu$ such that $\mu\in \cU$,
\begin{equation}
\label{eq:half-gap-bound}
    |E_\rho^\Lambda(\mu)|\leq \|\widehat D_\rho^\Lambda(\mu)\|_{\kappa_{{s}}}\leq CJ_{\rm an},
\end{equation}
for a constant $C$,
where the last inequality comes from \Cref{eq:uniform-D-bound}.

We next relate \(S_\rho(t)\) to \(E_\rho^\Lambda(\lambda)\). The eigenspaces of \(N_\rho^B\) with eigenvalues zero and one contain arbitrary states on \(\Lambda\setminus B\). We will represent the restriction of \(D_{\rho,s}^{\rm ball}\) to these eigenspaces by operators \(K_0,K_1\), and bound \(K_1-K_0-2F_{\rho,R}(\lambda)I\). The function \(F_{\rho,R}\) is defined below in \Cref{eq:patch-gap} by running the recursion on \(B\). The proof also uses the expansion into connected \(p\)-words proved later in \Cref{sec:local-approx}. That expansion and its coefficient bounds follow from the recursion
without a dynamical comparison, so this forward reference is not circular.

\begin{theorem}[Finite-ball Ramsey approximation]
\label{thm:Ramsey-approximation}
Assume \(\nu\ge\max\{\bar\nu_0,54\pi\Jan/\kappa_0\}\) and \(n_*(\nu)\ge1\), the same conditions as in \Cref{thm:uniform-normal-form}. Fix \(1\le s\le n_*(\nu)\), an integer \(K\ge0\), and \(R\ge(K+1)r_0\). Both \(E_\rho^\Lambda\) in \Cref{eq:full-gap} and \(F_{\rho,R}\), defined below in \Cref{eq:patch-gap}, are evaluated at this same order \(s\).

For each \(\rho=(k,\ell)\), set \(B=B_R(k)\), \(C=\Lambda\setminus B\), and, for \(b=0,1\), define
\begin{equation}
    \ket{\phi_{B,b}^\rho} :=\ket{\phi_{k,b}^\rho} \bigotimes_{j\in B\setminus\{k\}}\ket{\phi_{j,0}^\rho}, \qquad \ket{\phi_{B,+}^\rho} :=\frac{\ket{\phi_{B,0}^\rho}+\ket{\phi_{B,1}^\rho}}{\sqrt2}.
    \label{eq:ball-Ramsey-states}
\end{equation}
Let \(\sigma_C\) be any density matrix on \(C\), and set \(\varrho_{+,\sigma}:=\ket{\phi_{B,+}^\rho}\bra{\phi_{B,+}^\rho}
\otimes\sigma_C\). Define
\begin{equation}
    S_{\rho,\sigma}(t):=\Tr\!\left(\varrho_{+,\sigma}e^{iG_\rho^B(\lambda)t} \tau_{k,+}^\rho e^{-iG_\rho^B(\lambda)t}\right).
    \label{eq:Ramsey-general-signal}
\end{equation}
For \(\sigma_C=\bigotimes_{j\in C}
\ket{\phi_{j,0}^\rho}\bra{\phi_{j,0}^\rho}\), this is \(S_\rho(t)\) from \Cref{eq:signal}. For every real \(t\):

\begin{enumerate}[label=(\roman*)]
    \item The approximation using \(E_\rho^\Lambda(\lambda)\) satisfies
    \begin{equation}
        \abs{S_{\rho,\sigma}(t)-e^{i[\nu+2E_\rho^\Lambda(\lambda)]t}} \le \frac{C_1\Jan V_R}{\nu} +C_2\Jan V_R|t|(2/3)^s+C_3\Jan|t|2^{-K}.
        \label{eq:Ramsey-approximation}
    \end{equation}
    
    \item The approximation using \(F_{\rho,R}(\lambda)\) satisfies
    \begin{equation}
        \abs{S_{\rho,\sigma}(t)-e^{i[\nu+2F_{\rho,R}(\lambda)]t}} \le \frac{C_1\Jan V_R}{\nu}+C_2\Jan V_R|t|(2/3)^s+C_3\Jan|t|2^{-K}.
        \label{eq:Ramsey-patch-approximation}
    \end{equation}
\end{enumerate}
The constants \(C_1,C_2,C_3\) are independent of the system size, \(\rho\), \(\sigma_C\), and the admissible stopping order.
\end{theorem}

\begin{proof}
We first construct the effective Hamiltonian with charge \(N_\rho^B\). Run the full-lattice ADHH recursion with input \(H(\lambda)\) and charge \(N_\rho^B\) for \(s\) steps, and denote its output \(D_s\) by \(D_{\rho,s}^{\rm ball}\). The superscript indicates that the \emph{field} is confined to the ball: the input Hamiltonian and the recursion still act on the full lattice, including couplings across the boundary of \(B\). In particular, this output need not equal the patch output \(\widehat D_{\rho,R}(\lambda)\). The latter uses only \(H_B(\lambda):=\sum_{a:\supp(P_a)\subseteq B}\lambda_aP_a\), as defined formally below in \Cref{eq:patch-H}. Set
\begin{equation}
    G_{\rm eff}:=\nu N_\rho^B+D_{\rho,s}^{\rm ball}.
    \label{eq:ball-effective-generator}
\end{equation}

The eigenspaces with eigenvalues zero and one of \(N_\rho^B\) are \(\operatorname{span}\{\ket{\phi_{B,b}^\rho}\}\otimes\mathcal H_C\), for \(b=0,1\), because the defect has charge weight one, the other sites of \(B\) have weight two, and the exterior has weight zero. Their spectral projections are \(\Pi_b=\ket{\phi_{B,b}^\rho}\bra{\phi_{B,b}^\rho}\otimes I_C\). Since \([D_{\rho,s}^{\rm ball},N_\rho^B]=0\), both eigenspaces are invariant under \(G_{\rm eff}\).

For an operator \(W\) on the full system, define its partial matrix element over \(B\) by
\begin{equation}
    \mathcal M_b(W) :=(\bra{\phi_{B,b}^\rho}\otimes I_C) W(\ket{\phi_{B,b}^\rho}\otimes I_C)=\Tr_B\!\left[ (\ket{\phi_{B,b}^\rho}\bra{\phi_{B,b}^\rho}\otimes I_C)W\right].
    \label{eq:ball-partial-matrix-element}
\end{equation}
Thus \(\mathcal M_b(W)\) is an operator on \(C\): the partial trace is taken after inserting the projector onto the specified state of \(B\). The conditional exterior Hamiltonians are \(K_b:=\mathcal M_b(D_{\rho,s}^{\rm ball})\). For every \(\ket\psi\in\mathcal H_C\), invariance of these eigenspaces gives
\begin{equation}
    G_{\rm eff}(\ket{\phi_{B,b}^\rho}\otimes\ket\psi)=\ket{\phi_{B,b}^\rho}\otimes(\nu b I_C+K_b)\ket\psi.
    \label{eq:ball-sector-generator}
\end{equation}

We next compare the difference of these exterior Hamiltonians with the patch transition energy. Define
\begin{align}
    \mathcal T_B(W)&:=\frac12\bigl(\mathcal M_1(W)-\mathcal M_0(W)\bigr),\\
    \norm{\mathcal T_B(W)}&\le\sum_{X\ni k}\norm{W_X}\le\norm W_\kappa.
    \label{eq:ball-transition-functional}
\end{align}
To see the first inequality, the two states in \Cref{eq:ball-Ramsey-states} differ only at \(k\), so \(\mathcal T_B(W_X)=0\) whenever \(k\notin X\). Each partial matrix element is operator-norm contractive, giving \(\norm{\mathcal T_B(W_X)}\le\norm{W_X}\) for the remaining components. In particular, \(K_1-K_0=2\mathcal T_B(D_{\rho,s}^{\rm ball})\).

Fix \(\mu\in\cU\). To make the coefficient comparison explicit, let \(D_{\rho,s}^{\rm ball}(z;\mu)\) denote \(D_s\) from the recursion with charge \(N_\rho^B\) and input \(zH(\mu)\), and introduce the two analytic functions
\begin{align}
    \mathcal F_\rho^{\rm ball}(z)&:=\mathcal T_B\bigl(D_{\rho,s}^{\rm ball}(z;\mu)\bigr)=\sum_{p\ge0}T_p z^p,\\
    f_{\rho,R}(z)&:=F_{\rho,R}(z\mu)=\sum_{p\ge0}c_p z^p.
    \label{eq:ball-patch-Taylor-series}
\end{align}
The first function takes values in operators on \(C\); the second is scalar, and we compare the first with \(f_{\rho,R}(z)I_C\). All recursions in this comparison stop at the same fixed \(s\). Since \(\norm{H(\mu)}_{\kappa_0}\le2J\), patch restriction is contractive, and \(\Jan=4J\), both inputs lie in the analytic ball for \(|z|<2\). Consequently, \Cref{thm:uniform-normal-form}, together with \Cref{eq:ball-transition-functional,eq:ell-local-bound-general-main}, gives
\begin{equation}
    \max\left\{\norm{\mathcal F_\rho^{\rm ball}(z)},\abs{f_{\rho,R}(z)}\right\}\le C\Jan, \qquad |z|<2.
    \label{eq:ball-patch-analytic-bound}
\end{equation}

We claim that \(T_p=c_p I_C\) for \(0\le p\le K\). Both constant coefficients vanish because the interaction output is zero at zero input. For positive degrees, use the connected-word expansion of \Cref{lem:connected-word-expansion}, proved below. Its induction applies also with charge \(N_\rho^B\): onsite charge conjugations preserve the support of each Pauli letter, and nonzero commutators join overlapping words. Thus a degree-\(p\) contribution is built from \(p\) input Pauli terms with a connected overlap graph.

If such a contribution survives \(\mathcal T_B\), the union of its input supports contains \(k\). Since each input support has diameter at most \(r_0\), connectedness places this entire union within \(B_{pr_0}(k)\). For \(p\le K\), all its input terms are therefore retained in \(H_B(\mu)\), and every intermediate operator used to construct that contribution acts within \(B\). On any such operator \(W\), the charge actions satisfy
\begin{equation}
    e^{i\theta N_\rho^B}We^{-i\theta N_\rho^B}
    =e^{i\theta N_{\rho,B}}We^{-i\theta N_{\rho,B}}
    =e^{i\theta N_\rho}We^{-i\theta N_\rho}.
    \label{eq:ball-local-charge-agreement}
\end{equation}
Here \(N_{\rho,B}\), defined below in \Cref{eq:patch-charge}, is extended by the identity on \(C\). The last equality holds because the additional onsite charges in \(N_\rho\) act outside \(B\) and commute with \(W\). Hence \(\mathcal P_{N,\omega}\) and \(\mathcal S_N\) agree on every intermediate operator for these three choices of \(N\). Using the recursion \Cref{eq:recursion-no-N}, the same ordered input terms produce the same coefficient in the recursions with inputs \(zH(\mu)\) and \(zH_B(\mu)\). Terms whose input supports miss \(k\) contribute zero to both transition functionals. Finally, every surviving low-degree term acts as the identity on \(C\), so its partial matrix-element difference is precisely its scalar patch contribution times \(I_C\). This proves the claimed equality of coefficients.

Cauchy's estimate applied on circles of radius \(r<2\), followed by \(r\uparrow2\), gives \(\norm{T_p},\abs{c_p}\le C\Jan2^{-p}\). Evaluating at \(z=1\) and using the matching coefficients yields
\begin{equation}
    \norm{\mathcal F_\rho^{\rm ball}(1)-F_{\rho,R}(\mu)I_C}\le\sum_{p>K}\bigl(\norm{T_p}+\abs{c_p}\bigr)\le2C\Jan\sum_{p>K}2^{-p}=2C\Jan2^{-K}.
    \label{eq:ball-patch-Taylor-tail}
\end{equation}
Taking \(\mu=\lambda\) and absorbing numerical factors into the constant, we obtain
\begin{equation}
    \norm{K_1-K_0-2F_{\rho,R}(\lambda)I_C}\le C\Jan2^{-K}.
    \label{eq:ball-conditional-difference}
\end{equation}

We now compute the effective signal for the initial state \(\varrho_{+,\sigma}\) in the theorem. For the original preparation \(\ket{\Phi_+^\rho}\), the exterior density matrix is \(\sigma_C=\ket{\phi_{C,0}^\rho}\bra{\phi_{C,0}^\rho}\), where \(\ket{\phi_{C,0}^\rho}:=\bigotimes_{j\in C}\ket{\phi_{j,0}^\rho}\). \Cref{eq:ball-sector-generator} implies
\begin{equation}
    e^{-iG_{\rm eff}t}(\ket{\phi_{B,b}^\rho}\otimes\ket\psi)=e^{-i\nu bt}\ket{\phi_{B,b}^\rho}\otimes e^{-iK_b t}\ket\psi.
    \label{eq:ball-sector-evolution}
\end{equation}
On the direct sum of these two eigenspaces, \(\tau_{k,+}^\rho\) is \(2\ket{\phi_{B,1}^\rho}\bra{\phi_{B,0}^\rho}\otimes I_C\). The initial superposition has coherence \(1/2\) between the two branches, which cancels this factor of two. Taking the expectation therefore gives
\begin{equation}
    S_{\rm eff}(t):=\Tr\!\left(\varrho_{+,\sigma}e^{iG_{\rm eff}t}\tau_{k,+}^\rho e^{-iG_{\rm eff}t}\right)=e^{i\nu t}\Tr_C\!\left(\sigma_Ce^{iK_1t}e^{-iK_0t}\right).
    \label{eq:ideal-Ramsey-computation}
\end{equation}

Write \(\delta:=2F_{\rho,R}(\lambda)\) and \(\Delta:=K_1-K_0-\delta I_C\). The operators \(K_b\) are Hermitian and \(\delta\) is real. Duhamel's formula gives
\begin{equation}
    e^{iK_1t}-e^{i(K_0+\delta I_C)t}=i\int_0^t e^{iK_1(t-u)}\Delta
    e^{i(K_0+\delta I_C)u}\,du.
 \label{eq:ball-exterior-Duhamel}
\end{equation}
All propagators in the integral are unitary. Multiplying on the right by \(e^{-iK_0t}\), and using \(\norm{\sigma_C}_1=1\), yields
\begin{equation}
    \abs{S_{\rm eff}(t)-e^{i(\nu+\delta)t}}\le\norm{e^{iK_1t}e^{-iK_0t}-e^{i\delta t}I_C}\le |t|\norm\Delta\le C\Jan|t|2^{-K}.
    \label{eq:ball-effective-tone-error}
\end{equation}
Thus only the deviation of \(K_1-K_0\) from a scalar enters the bound; no estimate of the possibly extensive \(\norm{K_0}\) is needed.

For \(S_{\rho,\sigma}(t)\) in \Cref{eq:Ramsey-general-signal}, apply \Cref{thm:ADHH-dynamics} with charge \(N_\rho^B\) and observable \(\tau_{k,+}^\rho\), whose norm is two. Taking its expectation in \(\varrho_{+,\sigma}\) gives
\begin{equation}
    \abs{S_{\rho,\sigma}(t)-S_{\rm eff}(t)}\le C\Jan|B|\left(\nu^{-1}+|t|(2/3)^s\right).
    \label{eq:ball-physical-signal-error}
\end{equation}
Combining this with \Cref{eq:ball-effective-tone-error} and \(|B|\le V_R\) proves (ii), namely \Cref{eq:Ramsey-patch-approximation}.

It remains to prove (i), which uses \(E_\rho^\Lambda(\lambda)\). \Cref{eq:ball-local-charge-agreement} shows why the extra fields outside \(B\) do not change the low-degree contributions to the transition energy at \(k\). Applying the same coefficient-matching and Cauchy coefficient estimate to \(E_\rho^\Lambda(z\mu)\) gives
\begin{equation}
    \abs{E_\rho^\Lambda(\lambda)-F_{\rho,R}(\lambda)}\le C\Jan2^{-K},
    \label{eq:ball-auxiliary-gap-comparison}
\end{equation}
This estimate will also be stated as \Cref{thm:patch-locality} below; its derivation uses only the recursion and the analytic estimates, not the present signal bounds. Consequently,
\begin{equation}
    \abs{e^{i[\nu+2E_\rho^\Lambda(\lambda)]t}-e^{i[\nu+2F_{\rho,R}(\lambda)]t}}\le2|t|\abs{E_\rho^\Lambda(\lambda)-F_{\rho,R}(\lambda)}.
\end{equation}
The triangle inequality proves (i). Thus \(N_\rho\) is used to define the scalar \(E_\rho^\Lambda\), while the estimate for the physical evolution uses \(N_\rho^B\) and the operators \(K_0,K_1\) on \(C\).
\end{proof}

\subsection{Robust frequency estimation for transition energies}

\label{subsec:frequency-estimation}

\Cref{thm:Ramsey-approximation} guarantees that we can approximately generate a single-frequency signal. The robust frequency estimation procedure~\cite{LiTongNiGefenYing2023heisenberg,ma2024learningkbodyhamiltonianscompressed,MobusBluhmGefenEtAl2025heisenberg,brahmachari2026} (a variant of robust phase estimation~\cite{KimmelLowYoder2015robust}) enables us to robustly extract the frequency with Heisenberg-limited scaling. We will restate~{\cite[Theorem~4.1]{MobusBluhmGefenEtAl2025heisenberg}} below, while adapting it to the notation in our paper:

\begin{lemma}[Robust frequency estimation;~{\cite[Theorem~4.1]{MobusBluhmGefenEtAl2025heisenberg}}]
\label{lem:robust-frequency}
    Let $\omega\in[\omega_0-\Phi,\omega_0+\Phi]$, and let $\widehat c(t)$ and $\widehat s(t)$ be independent random variables such that $|\widehat c(t)-\cos(\omega t)|<1/\sqrt{8}$ with probability at least $2/3$, and $|\widehat s(t)-\sin(\omega t)|<1/\sqrt{8}$ with probability at least $2/3$, for any queried $t$. Then for $\eta_\omega>0$, with independent samples $\widehat c(t_1),\widehat c(t_2),\cdots,\widehat c(t_{N_{\rm time}})$ and $\widehat s(t_1),\widehat s(t_2),\cdots,\widehat s(t_{N_{\rm time}})$, with
    \begin{equation}
        N_{\rm time}=\bigO(\log^2(\Phi/\eta_\omega)), \quad T=t_1+t_2+\cdots+t_{N_{\rm time}}=\Theta(1/\eta_\omega), \quad \max_j t_j=\bigO(1/\eta_\omega) ,
    \end{equation}
    and $t_j\geq 0$, we can construct  an estimator $\hat{\omega}$ such that $\sqrt{\mathbb{E}[|\hat{\omega}-\omega|^2]}\leq \eta_\omega$.
\end{lemma}
In the original Theorem~4.1 in \cite{MobusBluhmGefenEtAl2025heisenberg} $\omega_0=0$, but setting it to an arbitrary value does not change the estimation algorithm or its analysis.

Note that because of the effective half transition energy bound \Cref{eq:half-gap-bound}, the frequency $\omega$ that we want to estimate lies in the interval $[\nu-2CJ_{\rm an},\nu+2CJ_{\rm an}]$, we can therefore set $\Phi=2CJ_{\rm an}=\bigO(1)$. In order to estimate the effective half transition energy $E_\rho^\Lambda(\lambda)$, we can simply generate the approximately single-frequency signal $S_\rho(t)$ in \Cref{eq:signal}, and take its real part as $\widehat c(t)$, and its imaginary part as $\widehat s(t)$. These two parts can each be estimated to $1/\sqrt{8}$ precision through Monte Carlo averaging. Note that the estimates are not required to be unbiased. This allows us to estimate the frequency $\nu + 2E_\rho^\Lambda(\lambda)$ in~\Cref{thm:Ramsey-approximation} to any precision $\eta_\omega>0$. Because we know $\nu$,  $E_\rho^\Lambda(\lambda)$ can then be readily estimated. To obtain half transition energy accuracy \(\eta\), the estimator uses times at most \(T_{\max}=C_T/\eta\), for a fixed constant \(C_T\). Choose \(s=K=\ceil{C_0\log(eJ/\eta)}\) and \(R=(K+1)r_0\), with \(C_0\) sufficiently large. A sufficient field condition, using \Cref{eq:Ramsey-approximation},  is
\begin{align}
    \nu\ge\nu_{\rm dyn}(\eta)
    &:=CJ\max\{V_R,\ s[1+\log(s+1)]^3\} \label{eq:field-strength-requirement} \\
    &=O\!\left(J\max\left\{[\log(eJ/\eta)]^d,
    \log(eJ/\eta)[1+\log\log(eJ/\eta)]^3\right\}\right).
\end{align}
The constant \(C\) includes the thresholds \(\bar\nu_0\) and \(54\pi\Jan/\kappa_0\). The term involving \(s\) ensures \(s\le n_*(\nu)\), and the volume term makes \(C_1\Jan V_R/\nu\) smaller than a fixed constant. The remaining errors in \Cref{eq:Ramsey-approximation} are also smaller than a fixed constant for \(|t|\le T_{\max}\), because their exponential decay in \(s,K\) dominates the factors \(V_R\) and \(T_{\max}\). We choose these constants so that the deterministic signal error, bounded via \Cref{eq:Ramsey-approximation}, is strictly below \(1/\sqrt8\), which leaves room for the sampling error in \Cref{lem:robust-frequency}. The selected \(s\) is kept fixed as \(\nu\) or the trial coefficients vary.

From this we obtain the following corollary
\begin{corollary}
    \label{cor:estimate-half-gap}
    Assume that $\nu\ge\nu_{\rm dyn}(\eta)$, where \(\nu_{\rm dyn}(\eta)\), \(s\), \(K\), and \(R\) are given in the preceding parameter choice and \Cref{eq:field-strength-requirement} For each $\rho$, the effective half transition energy $E_\rho^\Lambda(\lambda)$ can be estimated to precision $\eta$ with probability at least $1-\delta_{\rm est}$, using $\bigO(\log^2(1/\eta)\log(1/\delta_{\rm est}))$ independent non-adaptive experiments, $\bigO(\log(1/\delta_{\rm est})/\eta)$ total evolution time, and $\bigO(1/\eta)$ evolution time for each experiment.
\end{corollary}
{To obtain the stated confidence, apply \Cref{lem:robust-frequency} with root-mean-square frequency error \(\eta\). Markov's inequality gives frequency error at most \(2\eta\), hence half-transition-energy error at most \(\eta\), with probability at least \(3/4\). Taking the median of an odd number \(\bigO(\log(1/\delta_{\rm est}))\) of independent repetitions reduces the failure probability to \(\delta_{\rm est}\) by the Chernoff bound.}
In our learning protocol, we will estimate the effective half transition energy associated with each $\rho$, and denote it by $y_\rho$. We require
\begin{equation}
 \abs{y_\rho-E_\rho^{\Lambda}(\lambda)}\le\eta_{\rm freq}.
 \label{eq:full-gap-data-error}
\end{equation}
We can ensure that this bound holds for all $\rho$ with probability at least $1-\delta_{\rm freq}$, by setting $\delta_{\rm est} = \delta_{\rm freq}/(nm)$ in~\Cref{cor:estimate-half-gap}. We therefore have an {\(\bigO(\log(nm/\delta_{\rm freq}))\)} overhead in total evolution time and the number of experiments for estimating all $y_\rho$ accurately.

\section{Local approximation}
\label{sec:local-approx}

In this section we will show that an effective half transition energy $E_\rho^\Lambda(\mu)$, with $\rho=(k,\ell)$, can be approximated as a function of Hamiltonian parameters that are close to $k$. Proofs of the lemmas and theorems in this section can be found in Appendix~\ref{app:patch-proofs}.

\subsection{Patch effective transition energy}
\label{subsec:patch-map}

We consider a ball of radius $R$ around defect $k$, and let
\begin{equation}
    B:=B_R(k), \quad
    H_B(\mu):=\sum_{\supp(P_a)\subseteq B}\mu_aP_a,
    \label{eq:patch-H}
\end{equation}
For a general interaction, define patch restriction componentwise by
\begin{equation}
    (\mathcal E_B Z)_X :=2^{-|X\setminus B|}\operatorname{Tr}_{X\setminus B}(Z_X) \otimes I_{X\setminus B}, \qquad
    \norm{\mathcal E_B Z}_\kappa\le\norm Z_\kappa.
    \label{eq:interaction-patch-restriction}
\end{equation}
The component map is an average of unitary Pauli conjugations on \(X\setminus B\), hence is operator-norm contractive, also for complex inputs. It keeps Pauli strings supported in \(B\) and annihilates the others, so \(\mathcal E_B H(\mu)\) represents exactly \(H_B(\mu)\). The recursion on \(B\) starts from \(\mathcal E_B Z\) and retain the ambient labels \(X\subseteq\Lambda\); its local norms use these labels, even when \(X\not\subseteq B\). Support in \(B\) for an interaction means that each component acts only within \(B\). This property is preserved by the recursion, whose corresponding operator version can therefore be evaluated on the patch Hilbert space. We will call $B$ a patch, and also restrict the charge to it,
\begin{equation}
    N_{\rho,B}:=Q_k^\rho+2\sum_{\substack{j\in B\\j\ne k}}Q_j^\rho.
    \label{eq:patch-charge}
\end{equation}
With the above definitions, we can run exactly the same ADHH recursion \Cref{eq:ADHH-recursion} for \({s}\) steps on \(\nu N_{\rho,B}+H_B(\mu)\). Denote its final charge-preserving interaction by $\widehat D_{\rho,R}(\mu):=D_{{s}}$. The patch effective half transition energy is
\begin{equation}
    F_{\rho,R}(\mu) :=\frac12\left(\bra{\Phi_1^\rho}\widehat D_{\rho,R}(\mu)\ket{\Phi_1^\rho}-\bra{\Phi_0^\rho}\widehat D_{\rho,R}(\mu)\ket{\Phi_0^\rho}\right),
    \label{eq:patch-gap}
\end{equation}
where an operator local to the patch $B$ is extended by the identity outside \(B\).  We abbreviate \(F_\rho:=F_{\rho,R}\) when \(R\) is fixed.

For an interaction $Z$ that we feed into the ADHH recursion, define the zeroth-order full-lattice half transition energy as
\begin{equation}
    E_{\rho,0}^{\Lambda}(Z):=\ell_\rho(Z).
    \label{eq:bare-half-gap}
\end{equation}
Here by ``zeroth-order'' we mean the limit we obtain by considering $\nu\to\infty$. We note that $E_{\rho,0}^{\Lambda}(Z)=\ell_\rho(\langle Z\rangle_{N_\rho})=\ell_{\rho}(D_{\rho,0}^\Lambda)$. We also adopt the analogous definition for $E_{\rho,0}^{B}(Z)$ on a patch $B$.

The lemma below tells us that the zeroth-order half transition energy deviates from the effective half transition energy by $\bigO(1/\nu)$, and this is true for both the full lattice $\Lambda$ and any patch $B$.
\begin{lemma}
\label{lem:scalar-correction}
For every complex interaction \(Z\in\mathfrak B_{\rm an}\), there exists a constant $C_Q$ such that
\begin{equation}
    \abs{E_\rho^{\Lambda}(Z)-E_{\rho,0}^{\Lambda}(Z)}, \abs{E_\rho^{B}(Z)-E_{\rho,0}^{B}(Z)} \le {C_Q\Jan^2}/{\nu}.
    \label{eq:scalar-correction-full}
\end{equation}
The bound holds for every patch $B$ and full domain $\Lambda$.
\end{lemma}
If we restrict $Z$ to be supported on a patch $B$, the following lemma tells us that the effective half transition energy computed from the full-lattice ADHH recursion is identical to that computed from the patch ADHH recursion. This is because, in this setting, the ADHH iterates $D_j,V_j,A_j$ do not propagate outside the patch $B$.

\begin{lemma}\label{lem:exact-embedding}
If \(Z\) is supported in \(B_R(k)\), then the effective half transition energies of the full domain $\Lambda$ and the patch domain $B_R(k)$ are equal:
\begin{equation}
    E_\rho^{\Lambda}(Z)=E_\rho^{B_R(k)}(Z).
    \label{eq:exact-embedding}
\end{equation}
\end{lemma}
\begin{proof}
For convenience, write $B=B_R(k)$. We decompose the charge operator $N_{\rho}$ as
\begin{equation}
    N_\rho=N_{\rho,B}+N_{\rho,B^c},
    \label{eq:N-split-embedding}
\end{equation}
where $N_{\rho,B}$ is the sum of the single-qubit Pauli operators in $N_\rho$ supported in $B$, and $N_{\rho,B^c}$ is the sum of those supported in $B^c$. If \(W\) is supported in \(B\), then \([W,N_{\rho,B^c}]=0\). Hence
\begin{equation}
    e^{i\theta N_{\rho}} W e^{-i\theta N_{\rho}}=e^{i\theta N_{\rho,B}} W e^{-i\theta N_{\rho,B}} .
\end{equation}
It follows, by induction on the iteration index $j$, that for every $j$ and every operation in the ADHH recursion \Cref{eq:ADHH-recursion}, replacing $N_{\rho}$ by $N_{\rho,B}$ has no effect. In particular, $\widehat D_\rho^{\Lambda}(Z)$ and $\widehat D_{\rho,R}(Z)$ are identical. Since $E_{\rho}^\Lambda(Z)=\ell_\rho(\widehat D_\rho^{\Lambda}(Z))$ and $E_{\rho}^B(Z)=\ell_\rho({\widehat D_{\rho,R}(Z)})$, we have $E_{\rho}^\Lambda(Z)=E_{\rho}^B(Z)$.
\end{proof}

\subsection{Connected $p$-words in the ADHH recursion}
\label{subsec:connected-words}

Later analysis of our learning algorithm relies on estimates on the first- and second-order derivatives of the effective half transition energy. These estimates are obtained through the Cauchy integral estimates in~\Cref{lem:Cauchy-general}. For this we need to replace $H(\mu)$ by $zH(\mu)$, and consider complex $z$. We will fix \(\mu\in\cU\) ($\cU$ is defined in \Cref{eq:cU-defn}), abbreviate \(H=H(\mu)\) and \(H_B=H_B(\mu)\), and introduce the scalar holomorphic functions
\begin{equation}
    g_\rho^{\Lambda}(z):=E_\rho^{\Lambda}(zH), \quad g_\rho^B(z):=E_\rho^B(zH_B).
    \label{eq:g-functions}
\end{equation}
The order \(s\) is independent of \(z\) and satisfies \(1\le s\le n_*(\nu)\), with \(n_*(\nu)\) defined in \Cref{eq:new-n*}. The same \(s\) is used in both functions.

To analyze the functions $g_\rho^{\Lambda}(z)$ and $g_\rho^B(z)$, we will account for locality through the definition of \emph{connected $p$-words}.

\begin{definition}[Connected \(p\)-word]
\label{def:p-word}
Suppose that $H$ is represented as $H(\lambda)=\sum_{a\in\cA}\lambda_aP_a$. For each $a\in \mathcal{A}$, we define $\mathcal{Q}_a$ to be the set of Pauli operators with the same support as $P_a$.
Let $\mathbf a=(a_1,\ldots,a_p)\in\cA^p$ be an ordered tuple, with repetitions allowed, and $\mathbf{Q}=(Q_1,Q_2,\cdots,Q_p)$ where $Q_i\in \mathcal{Q}_{a_i}$. 
\begin{enumerate}[label=(\roman*)]
    \item We define the \emph{overlap graph} of $Q_1,Q_2,\cdots,Q_p$ to be the graph formed by $\{1,\ldots,p\}$ with an edge $i\sim j$ whenever $\supp(P_{a_i})\cap\supp(P_{a_j})\ne\varnothing$.
    
    \item We call $(\mathbf{a},\mathbf{Q})$ a \emph{connected $p$-word} if the overlap graph of $Q_1,Q_2,\cdots,Q_p$ is connected. Note that because $\mathrm{supp}(Q_i)=\mathrm{supp}(P_{a_i})$, this definition is in fact completely determined by $\mathbf{a}$.
    
    \item We call $W_{\mathbf{a},\mathbf{Q}}=Q_1Q_{2}\cdots Q_p$ the \emph{evaluation} of $(\mathbf{a},\mathbf{Q})$, although with a slight abuse of notation we also sometimes refer to it as a connected $p$-word if $(\mathbf{a},\mathbf{Q})$ is a connected $p$-word.
    
    \item We define the \emph{connected $p$-word space} to be  $\mathcal{C}_{p,H}=\mathrm{span}\{W_{\mathbf{a},\mathbf{Q}}:(\mathbf{a},\mathbf{Q}) \text{ is a connected }p\text{-word}\}$.
\end{enumerate}
\end{definition}

We note that despite the name, $W_{\mathbf{a},\mathbf{Q}}$ may have a disconnected support even when $(\mathbf{a},\mathbf{Q})$ is a connected $p$-word due to cancellation of Pauli components. Here connectedness refers to the connectedness of the underlying overlap graph. Connected $p$-words have the following properties that will be useful:
\begin{lemma}[Closure properties of connected $p$-words]
\label{lem:support-calculus}
Let $W_p\in \mathcal C_{p,H},W_q\in \mathcal C_{q,H}$. Then:
\begin{enumerate}[label=(\roman*)]
    \item The connected $p$-word space is closed under $\mathcal P_{N_\rho,\omega}$ and $\mathcal S_{N_\rho}$:
    \begin{equation}
        \mathcal P_{N_\rho,\omega}(W_p), \mathcal S_{N_\rho}(W_p) \in \mathcal C_{p,H}.    
    \end{equation}

    \item The commutator of connected $p$-word and $q$-word generate connected $p+q$-word:
    \begin{equation}
        [W_p,W_q] \in \mathcal C_{p+q,H}.
    \end{equation}

    \item If $A(z)=\sum_{p\ge 1} A_pz^p,B(z)={B_0+}\sum_{p\ge 1} B_pz^p$, satisfies $A_p,B_p\in \mathcal C_{p,H}$ for every $p{\geq 1}$ {and $B_0\in \mathrm{span}(N_\rho)$}, we have:
    \begin{equation}
        e^{A(z)}B(z)e^{-A(z)} ={B_0+} \sum_{p\ge 1} C_pz^p,
    \end{equation}
    where $C_p\in \mathcal C_{p,H}\;(\forall p \ge 1)$.
\end{enumerate}
\end{lemma}

The definition of connected $p$-words, together with the above properties, helps us to greatly reduce the number of terms present in the Taylor expansion of $D_{\rho,j}^{\Lambda}(zH)$ (and consequently $g_{\rho}^{\Lambda}(z)$) in $z$, as shown in the following lemma:

\begin{lemma}[Connected-word expansion]
\label{lem:connected-word-expansion}
We consider ADHH recursion \Cref{eq:ADHH-recursion} for \(H=H(\mu)=\sum_{a\in\cA}\mu_aP_a\) and static field $\nu N_\rho$. For every ADHH step $0\le j\le n_*$ and every $p\ge1$, the degree-$p$ coefficient of $D_{\rho,j}^{\Lambda}(zH)$ is an element of $\mathcal C_{p,H}$:
\begin{equation}
    D_{\rho,j}^{\Lambda}(zH) = \sum_{p\ge 1}D_{p,j}z^p,
\end{equation}
where $D_{p,j} \in \mathcal C_{p,H}$. More precisely, each \(D_{p,j}\) admits a finite expansion into terms $c_{\mathbf{a},\mathbf{Q}}\mu_{a_1}\cdots\mu_{a_p}W_{\mathbf a,\mathbf Q},$ where \((\mathbf a,\mathbf Q)\) is a connected \(p\)-word, and for fixed \(j\), \(\nu\), and $N$, the scalar $c_{\mathbf{a},\mathbf{Q}}$ is independent of \(\mu\) and depends only on $\mathbf{a},\mathbf{Q}$.
\end{lemma}
Essentially, this lemma ensures that in analyzing the $p$-th order Taylor term for $D_{\rho,j}^{\Lambda}(zH)$, we only need to focus on the connected $p$-words, of which there are only $\bigO(n)$, as opposed to $\bigO(n^p)$ arbitrary products of $p$ Hamiltonian terms. The two lemmas above are proved in~Appendix~\ref{app:patch-proofs}.

\subsection{Local approximation error}
\label{subsec:patch-theorem}

We will then analyze the Taylor expansion of $g_\rho^{\Lambda}(z)$ and $g_\rho^B(z)$ in $z$:
\begin{equation}
    g_\rho^{\Lambda}(z)=\sum_{p=0}^{\infty}c_p^{\Lambda}z^p, \quad
    g_\rho^B(z)=\sum_{p=0}^{\infty}c_p^Bz^p.
    \label{eq:g-Taylor}
\end{equation}
Importantly, the Taylor expansion coefficients up to  $K$th order only depend on the Hamiltonian coefficients in a ball of radius $(K+1)r_0$ around the defect $k$. This we prove in the lemma below:
\begin{lemma}
\label{lem:coefficient-matching}
If \(R\ge(K+1)r_0\), then $c_p^{\Lambda}=c_p^B$, where $0\le p\le K$.
\end{lemma}

Moreover, when $|z|<2$, $zH\in \mathfrak{B}_{\rm an}$, and consequently we can bound the final charge-preserving interactions produced by the ADHH recursion using~\Cref{thm:ADHH-normal-form}, which leads to the following lemma:

\begin{lemma}[Uniform analytic transition-energy bound]
\label{lem:analytic-gap-bound}
For every \(\mu\in\cU\), the functions in \Cref{eq:g-functions} are analytic for \(|z|<2\) and satisfy
\begin{equation}
    \sup_{|z|\le2}\max\left\{\abs{g_\rho^{\Lambda}(z)},\abs{g_\rho^B(z)}\right\}\le C_g\Jan.
    \label{eq:analytic-gap-bound}
\end{equation}
Consequently,
\begin{equation}
    |c_p^{\Lambda}|\le C_g\Jan2^{-p},\quad |c_p^B|\le C_g\Jan2^{-p}.
    \label{eq:Taylor-coeff-bound}
\end{equation}
\end{lemma}

The two lemmas above help us prove that the full-lattice effective half transition energy $E_\rho^{\Lambda}(\mu)$ is close to the patch effective half transition energy $F_{\rho,R}(\mu)$. This is done by studying their Taylor expansions. The low-order terms are identical by~\Cref{lem:coefficient-matching}, and the high-order terms are small by~\Cref{lem:analytic-gap-bound}.

\begin{theorem}[Local approximation of the effective half transition energy]
\label{thm:patch-locality}
Let \(R\ge(K+1)r_0\).  Uniformly for every row \(\rho=(k,\ell)\) and every \(\mu\in\cU\),
\begin{equation}
    \abs{{E_\rho^{\Lambda}(\mu)}-F_{\rho,R}(\mu)}\le C_{\rm P}\Jan2^{-K}.
    \label{eq:patch-locality}
\end{equation}
One may take \(C_{\rm P}=2C_g\).
\end{theorem}
Detailed proofs of the lemmas and theorems above are included in~Appendix~\ref{app:patch-proofs}.

\begin{remark}\label{rem:applicability_to_parallel}
Although \Cref{lem:scalar-correction,lem:exact-embedding,lem:support-calculus,lem:connected-word-expansion,lem:coefficient-matching,lem:analytic-gap-bound} and \Cref{thm:patch-locality} are stated for the single-defect charge \(N_\rho\), they also hold for the multiple-defect charge \(N_{\mathcal D,\ell}\) defined in \Cref{eq:parallel-charge}. For each \(k\in\mathcal D\), the corresponding scalar quantities are defined using \(\ket{\Phi_{\varnothing}^{\mathcal D,\ell}}\) and \(\ket{\Phi_{\{k\}}^{\mathcal D,\ell}}\) from \Cref{eq:parallel-code-states} in place of \(\ket{\Phi_0^\rho}\) and \(\ket{\Phi_1^\rho}\), and patch charges ($N_{\rho,B}$ in \Cref{eq:patch-charge} for patch $B$ in the single-defect setting) are obtained by restricting \(N_{\mathcal D,\ell}\) to the patch. The same proofs apply because the charge remains a sum of on-site operators with integer spectra, while the two reference states are product eigenstates of the charge that differ only at \(k\). Importantly, these results do not require the relevant charge sectors to be non-degenerate.
\end{remark}

As introduced before, we will update our guess $\mu$ for the Hamiltonian coefficients to match the observed effective half-transition-energy data through experiments. However, instead of directly considering the map from $\mu$ to the full-lattice data, we consider the corresponding patch data. More precisely, the map we consider is defined as
\begin{equation}
    F(\mu):=(F_\rho(\mu))_{\rho\in\Lambda\times[m]}.
    \label{eq:global-F}
\end{equation}
Here, as before, we use $\rho=(k,\ell)$ to index each data point corresponding to defect site $k$ and repetition index $\ell$. For each $k$ there are $m$ independent estimates, each associated with a basis $\{(s_j^\rho,\beta_j^\rho):j\in \Lambda\}$, and $F_\rho(\mu)$ is the patch effective half-transition-energy function defined in \Cref{eq:patch-gap}. Replacing the full-lattice map with this patch forward map introduces sparsity in the Jacobian and Hessian matrices, thus simplifying later analysis, and it only introduces a small error controlled by~\Cref{thm:patch-locality}.

Combining \Cref{eq:full-gap-data-error} and \Cref{thm:patch-locality}, the observed data satisfy
\begin{equation}
    y=F(\lambda)+e, \quad
    \norm e_\infty\le\eta_{\rm freq}+\eta_{\rm patch}, \quad
    \eta_{\rm patch}:=C_{\rm P}\Jan \cdot 2^{-K},
    \label{eq:patch-data-model}
\end{equation}
where $\eta_{\rm patch}$ accounts for how far the patch effective half transition energies deviate from the full-lattice values. In the analysis below, we will treat this $F$ as the forward map, and locally invert it to obtain the coefficients $\lambda$ from estimates $y$. Importantly, this forward map can be evaluated on a classical computer in time polynomial in $n$ and inverse precision due to its local nature. We will discuss this in detail in Appendix~\ref{app:classical-evaluation}.

\section{Recovering coefficients by optimization}
\label{sec:optimization}

In this section we will discuss how to recover the Hamiltonian through efficiently solving an optimization problem on a classical computer. Proofs of lemmas and theorems in this section can be found in Appendix~\ref{app:optimization-proofs}.

\subsection{Least-squares objective function}
The optimization minimizes the least-squares objective function:
\begin{equation}
    \mathcal L(y,\mu):=\frac{1}{2m}\norm{F(\mu)-y}_2^2,\quad \mu\in\cU.
    \label{eq:loss}
\end{equation}

Define the gradient and residual by
\begin{equation}
    f(y,\mu):=\nabla_\mu\mathcal L(y,\mu), \quad r(y,\mu):=F(\mu)-y.
    \label{eq:gradient-residual-active}
\end{equation}
To guaranty optimization, we prove that Hessian $\mathsf H(y,\mu)$ is well conditioned. The Hessian of the objective function is decomposed into the Jacobian Gram matrix $\mathsf G(\mu):=\frac{1}{m}\mathsf J(\mu)^{\mathsf T}\mathsf J(\mu)$ and a residual term involving the second derivative of $F$:
\begin{equation}
    \mathsf H(y,\mu):=D_\mu f(y,\mu)={\mathsf G(\mu)}+\frac1m\sum_\rho r_\rho(y,\mu)\nabla^2F_\rho(\mu).
    \label{eq:Hessian-split}
\end{equation}

\Cref{subsec:loss-Walsh} derives the Jacobian of the zeroth-order forward map and the number of experimental repetitions required for its conditioning. We then prove conditioning of the effective Jacobian and Hessian in \Cref{subsec:jacobian-hessian}, followed by convergence in \Cref{subsec:convergence}.

\subsection{Zeroth-order forward map and its Jacobian}
\label{subsec:loss-Walsh}

In this section, we will show that in the strong-field limit of $\nu\to\infty$, the zeroth-order half transition energies are sufficient for recovering Hamiltonian coefficients by solving the optimization problem in \Cref{eq:loss}. We recall that $q$ is the maximal support size of Hamiltonian terms $P_a$, $D_{\rm ov}$ is the maximal number of terms a Hamiltonian term can overlap with (including itself), and $\zeta$ is the maximal number of terms that cover any given qubit $k$. $r_0$ upper bounds the diameter of the support of all Hamiltonian terms, and $d$ is the dimension of the physical lattice. These parameters are all constants that are independent of system size $n$.

For a Pauli term $P_a=\bigotimes_{j\in\Lambda}P_{a,j}$, define $p_a:=|\supp(P_a)|$ and $\chi(a)\in\{0,1\}^{\Lambda}$ by $\chi_j(a):=\ind[P_{a,j}\ne I]$. We write $P_a\sim\beta^\rho$ when $P_{a,j}=\beta_j^\rho$ for every $j\in\supp(P_a)$.

In the limit $\nu\to\infty$,~\Cref{lem:scalar-correction} shows that the patch effective half transition energy $F_{\rho,R}(\mu)$ converges to its zeroth-order counterpart. By \Cref{eq:bare-half-gap}, this zeroth-order quantity is a linear function of $\mu$.
Therefore we define
\begin{equation}
    F_{\rho,0}(\mu):=E_{\rho,0}^{B_R(k)}(H_{B_R(k)}(\mu)), \quad X_{\rho} := D F_{\rho,0}(\mu),
    \label{eq:bare-and-correction}
\end{equation}
Define the zeroth-order forward map $F_0(\mu):=(F_{\rho,0}(\mu))_\rho$ and its Jacobian $X:=DF_0(\mu)$. Note that $X_{\rho}$ does not depend on $\mu$ because of the linearity of $F_{\rho,0}(\mu)$. We then consider $X$ as the matrix whose rows are $X_{\rho}$. The matrix $X$ is therefore of size $(nm)\times |\mathcal{A}|$.  This section establishes a lower bound on the sample size $m$ required for the zeroth-order Jacobian $X$ to be well conditioned. This result is subsequently used to show that the effective Jacobian $\mathsf J(\mu)=DF(\mu)$ is also well conditioned.

For the function $F_{\rho,0}$, we will prove the following:
\begin{lemma}
\label{lem:local-Walsh}
The function $F_{\rho,0}$ is linear:
\begin{equation}
    F_{\rho,0}(\mu)=\sum_{a\in\cA}X_{\rho,a}\mu_a,
    \label{eq:bare-linear}
\end{equation}
where $X_{\rho,a}=-\ind[k\in\supp(P_a)]\ind[P_a\sim\beta^\rho](-1)^{s^\rho\cdot\chi(a)}$ is the element of $X$ on the row indexed by $\rho$ and column indexed by $a$. If $R\ge r_0$, each row of $X\in\mathbb R^{nm\times M}$ has at most $\zeta$ nonzero entries, and each column has at most $qm$ nonzero entries. $\beta^\rho$ and $s^\rho$ specify the local Pauli frame and are defined in \Cref{eq:random-pattern}.
\end{lemma}
The entries of the zeroth-order Jacobian $X$ have a Walsh-character sign structure restricted by locality. We will refer to $X$ as the Jacobian of the zeroth-order forward map.

In the proof of the~\Cref{thm:Walsh-conditioning}, we will show that, as the number of samples tends to infinity, the Gram matrix of the zeroth-order Jacobian $X$,
\begin{equation}
\label{eq:Gamma}
    \Gamma := \frac{1}{m}X^{\mathsf T}X
\end{equation}
converges to $\Sigma := \operatorname{diag}\bigl(p_a3^{-p_a}\bigr)_{a\in\mathcal A}$, which is a well-conditioned diagonal matrix. Therefore, we expect $\Gamma$ to be well-conditioned when the number of samples is sufficiently large.
\begin{theorem}
\label{thm:Walsh-conditioning}
There is a constant \(C_m\), such that if
\begin{equation}
    m\ge C_m{\log(2n/\delta)},
    \label{eq:m-condition}
\end{equation}
with probability at least \(1-\delta\), the following holds:
\begin{align}
    \norm{\Gamma-\Sigma}_{\infty\to\infty} & \le 3^{-q}/4, \label{eq:Gamma-close}\\
    \frac{3}{4}\cdot3^{-q}I \preceq\Gamma &\preceq\left(q+3^{-q}/4\right)I, \label{eq:Gamma-spectral}\\
    \norm{\Gamma^{-1}}_{\infty\to\infty} & \le C_\Gamma:=(4/3)\cdot3^{q}. \label{eq:Gamma-inverse}
\end{align}
\end{theorem}

In the next section, we bound the perturbation $D\mathcal R(\mu)$ and derive a sufficient lower bound on the field strength $\nu$ that guarantees that the Jacobian $\mathsf J$ is well conditioned. Specifically, we decompose the Jacobian as
\begin{equation}
    \mathsf J(\mu):= {D F(\mu)} = X + {D\mathcal R(\mu)}.
\label{eq:Jacobian-def}
\end{equation}

\subsection{Well conditioned Jacobian and Hessian}
\label{subsec:jacobian-hessian}

The fact that the Jacobian of the map $F(\mu)$ in the limit of $\nu\to\infty$ is well-conditioned provides a way for us to analyze the finite-$\nu$ situation. We note that, because we will later analyze convergence in the $\ell^{\infty}$-norm rather than $\ell^2$-norm, it is not enough to simply analyze the usual condition number defined through singular values. Rather, we will use the bound \Cref{eq:Gamma-inverse} to analyze the convergence through the Banach fixed-point theorem. We will treat $F_{\rho,0}(\mu)$ in \Cref{eq:bare-and-correction} as a known baseline and analyze the Jacobian of $F(\mu)$ by carefully bounding the deviation
\begin{equation}
    \mathcal R_\rho(\mu):=F_\rho(\mu)-F_{\rho,0}(\mu).
    \label{eq:row-coeff-set}
\end{equation}
For the deviation {$\mathcal R_\rho(\mu)$}, we have the following bounds for its first- and second-order derivatives.
\begin{lemma}[Uniform first- and second-derivative bounds]
\label{lem:row-derivatives}
Uniformly for \(\mu\in\cU\) and all rows \(\rho\),
\begin{align}
    \sum_a\abs{\partial_a\mathcal R_\rho(\mu)} & \le {C_1\Jan}/{\nu},
    \label{eq:first-derivative-bound}\\
    \sup_a\sum_b\abs{\partial_a\partial_bF_\rho(\mu)} & \le{C_2}/{\nu}.
    \label{eq:second-derivative-bound}
\end{align}
All derivatives with respect to coefficients outside \(\cA_\rho:=\{a:\supp(P_a)\subseteq B_R(k)\}\) vanish.
\end{lemma}
This lemma is proved by combining~\Cref{lem:scalar-correction} with the Cauchy estimates \Cref{eq:Cauchy-polydisc}. More precisely,~\Cref{lem:scalar-correction} provides a bound for $\mathcal{R}_{\rho}(\mu)$ that holds uniformaly for all $\|\mu\|_\infty\leq 4$ (by~\Cref{eq:coeff-to-local}) vanishes like $\bigO(1/\nu)$, while~\Cref{eq:Cauchy-polydisc} allows us to use this bound to bound the derivatives of $\mathcal{R}_{\rho}(\mu)$ for $\mu\in\cU$.

\begin{lemma}[Jacobian Gram perturbation]
\label{lem:Jacobian-perturbation}
Uniformly for \(\mu\in\cU\), the Jacobian Gram matrix {$\mathsf G(\mu)=\mathsf{J}(\mu)^{\mathsf T} \mathsf{J}(\mu)/m$} satisfies
\begin{equation}
    \norm{{\mathsf G(\mu)}-\Gamma}_{\infty\to\infty} \le C_JV_R\left( \frac{\Jan}{\nu}+\frac{\Jan^2}{\nu^2} \right).
 \label{eq:Jacobian-Gram-close}
\end{equation} 
On the event that \Cref{eq:Gamma-close} holds, there exists an optimization field threshold $\nu_{\rm opt}=\mathcal{O}(JV_R)$ such that for all $\nu\geq \nu_{\rm opt}$, we have
\begin{equation}
    {\mathsf G(\mu)}\succeq\frac{I}{2\cdot3^q}
    \label{eq:Jacobian-conditioned}
\end{equation}
uniformly on \(\cU\).
\end{lemma}

Let $y_*=F(\lambda)$ be the actual patch effective half-transition-energy values, i.e., $(y_{*})_{\rho} = F_{\rho}(\lambda)$. We consider their estimates within the following domain $\cY$:
\begin{equation}
    \cY:=\{y:\norm{y-y_*}_\infty\le r_y \},
    \label{eq:data-neighborhood}
\end{equation}
where $r_y \le 1$ is chosen to satisfy \Cref{eq:ry-condition}. By \Cref{lem:local-Walsh,lem:row-derivatives}, because $X$ has at most $\zeta$ non-zero entries in each row, each of which has absolute value bounded by $1$, and $\|D\mathcal R_\rho(\mu)\|_{\ell^1}\leq C_1 J_{\rm an}/\nu$, every row of \(\mathsf J(\mu)=X + D\mathcal R(\mu)\) has \(\ell_1\) norm at most
\begin{equation}
    L_F:=\zeta+\frac{C_1\Jan}{\nu}.
    \label{eq:LF}
\end{equation}
Recall that $r(y,\mu)=F(\mu)-y$, then for $y\in\cY$, using $F(\lambda)=y_*$,
\begin{equation}
    \norm{r(y,\mu)}_\infty
    \leq \norm{F(\mu)-F(\lambda)}_\infty + \norm{y-y_*}_\infty
    \le L_F+r_y=:r_E,
    \label{eq:residual-bound}
\end{equation}
where we have used
\begin{equation}
    \norm{F(\mu)-F(\lambda)}_{\infty} \leq  \int_0^1  \norm{{DF((1-s)\lambda+s\mu)}}_\infty \norm{\mu-\lambda}_{\infty} ds\leq  L_F.    
\end{equation}

The above results enable us to show that the Hessian matrix {$\mathsf H(y,\mu)$} is close to $\Gamma$ for large $\nu$. Specifically, given the decomposition \Cref{eq:Hessian-split}, we will show that the Gram matrix part {$\mathsf G(\mu)$} is close to $\Gamma$, and the remainder is small. These observations allow us to obtain the following quantitative bounds and ensure that {$\mathsf H(y,\mu)$} is well-conditioned.
\begin{lemma}[Uniform Hessian conditioning]
\label{lem:Hessian-conditioning}
Uniformly for \((y,\mu)\in\cY\times\cU\),
\begin{equation}
    \norm{\mathsf H(y,\mu)-\Gamma}_{\infty\to\infty}
    \le C_JV_R\left(\frac{\Jan}{\nu}+\frac{\Jan^2}{\nu^2}\right)+\frac{C_2V_Rr_E}{\nu}
    \label{eq:Hessian-close}
\end{equation}
On the event that \Cref{eq:Gamma-close} holds, there exists an optimization field threshold $\nu_{\rm opt}=\mathcal{O}(JV_R)$ such that when $\nu\geq \nu_{\rm opt}$, we have
\begin{equation}
    \mathsf H(y,\mu)\succeq\frac{I}{2\cdot 3^{q}},
    \label{eq:strong-convexity}
\end{equation}
and
\begin{equation}
    \norm{I-\Gamma^{-1}\mathsf H(y,\mu)}_{\infty\to\infty}\le 1/3.
    \label{eq:preconditioned-Hessian}
\end{equation}
In particular, \(\mathcal L(y,\cdot)\) is uniformly strongly convex on \(\cU\).
\end{lemma}
Strong convexity is often used to analyze convergence of optimization algorithms in $\ell^2$-norm. In our later analysis, we will analyze convergence in the $\ell^\infty$-norm, for which strong convexity, measured by the smallest eigenvalue of the Hessian matrix, is insufficient. Instead we will rely on the estimate \Cref{eq:preconditioned-Hessian}.

\subsection{Solving the optimization problem}
\label{subsec:convergence}

We use a fixed-preconditioner gradient iteration to solve the optimization problem in \Cref{eq:loss}. The preconditioner is $\Gamma^{-1}$, with $\Gamma$ defined in \Cref{eq:Gamma}; no Hessian evaluation is needed during the iteration. The update map is
\begin{equation}
    {\mathcal T_y(\mu)}:=\mu-\Gamma^{-1}f(y,\mu)=\mu-\Gamma^{-1}\frac1m\mathsf J(\mu)^{{\mathsf T}}(F(\mu)-y).
    \label{eq:iteration-map}
\end{equation}
We will then prove that this iteration converges to the unique minimizer of $\mathcal{L}(y,\cdot)$ \Cref{eq:loss}, and that the minimizer is a good approximation of the exact coefficient vector $\lambda$.

We will first provide an upper bound of $\norm{\mathsf J(\mu)^{{\mathsf T}}z/m}_\infty$ that will be useful in the proof. We note that, using the decomposition $\mathsf J(\mu)=X + D\mathcal{R}(\mu)$, we have
\begin{equation}
    \norm{\frac{1}{m}\mathsf J(\mu)^{{\mathsf T}}}_{\infty\to\infty}\leq \frac{1}{m}\max_a\sum_{\rho} |X_{\rho a}|+\frac{1}{m}\max_a\sum_\rho |\partial_a \mathcal{R}_{\rho}(\mu)|.
\end{equation}
We let  $\rho=(k,\ell)\in\Lambda\times [m]$. Note that for each term index $a\in \mathcal{A}$, $X_{\rho a}\neq 0$ only if $k\in\supp(P_a)$, and therefore the total number of $\rho$ resulting in non-zero contribution is at most $m\max_a |\supp(P_a)|=mq$. We also have $|X_{\rho a}|\leq 1$ following \Cref{lem:local-Walsh}. Therefore $(1/m)\cdot\max_a\sum_{\rho} |X_{\rho a}|\leq q$, or equivalently $\|X^{\mathsf T}\|_{\infty\to\infty}/m\leq q$. For each $a\in \mathcal{A}$, $\partial_a \mathcal{R}_{\rho}(\mu)\neq 0$ only if $\supp(P_a)\subset B_R(k)$, which restrict the number of $\rho$ with non-zero $\partial_a \mathcal{R}_{\rho}(\mu)$ to be at most $m|B_R(k)|\leq  mV_R$.  Also by \Cref{eq:first-derivative-bound} we have $|\partial_a \mathcal{R}_{\rho}(\mu)|\leq C_1 J_{\rm an}/\nu$. Therefore $(1/m)\cdot\max_a\sum_\rho |\partial_a \mathcal{R}_{\rho}(\mu)|\leq C_1 V_R J_{\rm an}/\nu$.

Combining the above, 
we define the constant $C_y$ that bounds the norm of $\mathsf J(\mu)^{\mathsf T}/m$, and we have
\begin{equation}
    \norm{\frac{1}{m}\mathsf J(\mu)^{\mathsf T}z}_\infty
    \le C_y\norm z_\infty,
    \quad C_y:=q+\frac{C_1V_R\Jan}{\nu}.
    \label{eq:Cy}
\end{equation}
With the above bound, we are ready to prove the convergence of the fixed-preconditioner gradient iteration, \Cref{eq:iteration-map} provided the estimated data vector $y$ is close to the actual $y_*$, and that the field strength $\nu$ is sufficiently large. We recall that $r_y$ is how far we allow $y$ to deviate from $y_*$ in $\ell^\infty$-norm, i.e., $y$ needs to be contained in $\mathcal Y=\left\{y:\|y-y_*\|_\infty\le r_y\right\},$ as defined in \Cref{eq:data-neighborhood}.
\begin{theorem}[Stability and convergence of the fixed-preconditioner gradient iteration]
\label{thm:local-inverse}
We condition on the event that \Cref{eq:Gamma-close} in \Cref{thm:Walsh-conditioning} is correct, and let $r_y$ in \Cref{eq:data-neighborhood} satisfy (with $C_\Gamma$ defined in \Cref{eq:Gamma-inverse} and $C_y$ in \Cref{eq:Cy})

\begin{equation}\label{eq:ry-condition}
 r_y\le\frac{1}{2C_\Gamma C_y}.
\end{equation}
There exists an optimization field threshold $\nu_{\rm opt}$:
\begin{equation}\label{eq:field-strength-optimization}
    {\nu_{\rm opt}}=\mathcal{O}(JV_R)=\mathcal{O}(JR^d)
\end{equation}
such that for every $\nu\geq \nu_{\rm opt}$, \(y\in\cY\), the loss \(\mathcal L(y,\cdot)\) has a unique stationary point \(\mu^\star(y)\in\cU\).  It is the unique minimizer in \(\cU\), and
\begin{equation}
    \norm{{\mu^\star(y)}-\lambda}_\infty \le\frac{3C_\Gamma C_y\norm{y-y_*}_\infty}{2}.
    \label{eq:stability}
\end{equation}
For every initialization \(\mu^{(0)}\in\cU\), the iteration
$\mu^{(t+1)}={\mathcal T_y}(\mu^{(t)}),$ with $\mathcal T_y$ defined in \Cref{eq:iteration-map},
converges geometrically to \(\mu^\star(y)\), with contraction factor at most \(1/3\).
\end{theorem}

We now explain the entire optimization process in detail. We assume that the field strength $\nu$ is above the threshold in \Cref{eq:field-strength-requirement} to ensure the correct estimation of $y$, and is above the threshold in \Cref{eq:field-strength-optimization} to ensure optimization convergence. Therefore it needs to satisfy $\nu\geq \nu_{\min}$, for a combined threshold
\begin{equation}\label{eq:nu-condition}
    \nu_{\min}:=\max\{\nu_{\rm dyn},\nu_{\rm opt}\} = \bigO\left(\max\left\{R^d,\log(1/\eps)[1+\log\log(1/\eps)]^3\right\}\right).
\end{equation}
We choose \(m=\ceil{C_m{\log(4n/\delta)}}\). By \Cref{thm:Walsh-conditioning}, the event \Cref{eq:Gamma-close} then holds with probability at least \(1-\delta/2\). We will condition on this event throughout the following initialization and convergence analysis.

\paragraph{Initialization and convergence.}

We can choose the initialization \(\mu^{(0)}\) as
\begin{equation}
    \mu^{(0)}:=\Gamma^{-1}\frac1mX^{{\mathsf T}}y.
    \label{eq:linear-initializer}
\end{equation}
If \(y=y_*+e\), then because
\begin{equation}
    \mu^{(0)}-\lambda=\Gamma^{-1}\frac{1}{m}X^{{\mathsf T}}\bigl(\mathcal R(\lambda)+e\bigr),    
\end{equation}
we use \Cref{eq:Gamma-inverse}, $\|X^{{\mathsf T}}\|_{\infty\to\infty}/m\leq q$,  and~\Cref{lem:scalar-correction} to obtain
\begin{equation}
    \norm{\mu^{(0)}-\lambda}_\infty \le C_\Gamma q\left( \frac{C_Q\Jan^2}{\nu}+\norm e_\infty\right).
    \label{eq:initializer-bound}
\end{equation}
Hence \(\mu^{(0)}\in\cU\) whenever the right-hand side is at most \(1\), which is satisfied by our field strength condition \Cref{eq:nu-condition}.

Starting from the initialization \(\mu^{(0)}\in\cU\) applying the map {$\mathcal T_y$} converges to {$\mu^\star(y)$} which satisfies \Cref{eq:stability}. Hence, in order to obtain an estimate within {$\eps$} error, we should satisfy
\begin{equation}
    \|e\|_\infty \le \eta_{\rm freq}+\eta_{\rm patch} \le \frac{{\eps}}{2C_\Gamma C_y}.
\end{equation}
Due to \Cref{eq:nu-condition}, $C_\Gamma C_y$ becomes a constant. We aim to bound the frequency and patch error as {$\eta_{\rm freq} \le \eps/(4C_\Gamma C_y)$ and $\eta_{\rm patch} \le \eps/(4C_\Gamma C_y)$}. We choose \(r_y=1/(2C_\Gamma C_y)\), so these error bounds also ensure \(y\in\cY\) for \(0<\eps\le1\).

\paragraph{Patch error and radius.}
The patch error are defined in \Cref{eq:patch-locality} and \Cref{eq:patch-data-model} such as $\eta_{\rm patch} = C_P\Jan2^{-K}$, and the radius becomes $R=(K+1)r_0$. To achieve {$\eta_{\rm patch} \le \eps/(4C_\Gamma C_y)$}, we set 
\begin{equation}
    K=\mathcal{O}(\log(J/\eps))=\mathcal{O}(\log(1/\eps)).
    \label{eq:K_choice}
\end{equation}
Then, the radius becomes $R=\mathcal{O}(\log(1/\eps))$.

\paragraph{Transition energy estimation error and total evolution time.}
Due to \Cref{eq:nu-condition}, we can estimate {each full-lattice effective half transition energy \(E_\rho^\Lambda(\lambda)\)} with single-frequency approximation from~\Cref{thm:Ramsey-approximation}. For this we use the robust frequency estimation protocol in~\Cref{cor:estimate-half-gap}. {Set \(\delta_{\rm freq}=\delta/2\) and \(\delta_{\rm est}=\delta/(2nm)\). The event that \Cref{eq:full-gap-data-error} holds for all sampled $\rho$ then holds with probability at least \(1-\delta/2\). Together with \Cref{eq:Gamma-close} and the deterministic patch bound \Cref{eq:patch-data-model}, this gives the required conditioning and data accuracy jointly with probability at least \(1-\delta\), by a union bound.} To achieve {\(\eta_{\rm freq}=\eps/(4C_\Gamma C_y)\)}, the required evolution time for each {$\rho$} becomes {\(\mathcal{O}(\log(2nm/\delta)/\eta_{\rm freq})\)}. Hence, the total evolution time becomes {\(\mathcal{O}(mn\log(2nm/\delta)/\eta_{\rm freq}) = \mathcal{O}(n\log^2(n/\delta)/\eps)\)}. Therefore, we have the final complexity as stated in \Cref{thm:main}.

\begin{theorem}[Learning a geometrically local Hamiltonian]
\label{thm:main}
Suppose that the Hamiltonian $H$ satisfies the conditions on \Cref{sec:intro} and we can apply static single-qubit fields in \Cref{sec:experiment}. We can obtain an estimate $\hat{\lambda}$ such that $\norm{\widehat\lambda-\lambda}_\infty\le\eps$ {with probability at least $1-\delta$}, with total evolution time $\mathcal{O}({n{\log^2(n/\delta)}}/{\eps})$ using field strength {$\nu\ge \nu_{\min}$} defined in \Cref{eq:nu-condition}.
\end{theorem}

The estimated half-transition energies are converted into Hamiltonian coefficients by evaluating the local forward model and applying a stable iterative fitting procedure. The forward model is computed through a connected-cluster expansion, avoiding a dense simulation of the local patches. This yields a total classical runtime polynomial in \(n\), \(1/\eps\), and \(\log(1/\delta)\). The algorithm and its complexity analysis are given in Appendix~\ref{app:classical-evaluation}.

\section{Parallelization}
\label{sec:parallelization}

We now run the experiment of \Cref{sec:experiment} at a set \(\mathcal D\) of separated defects. The physical field is supported on \(\mathcal B:=\bigcup_{k\in\mathcal D}B_R(k)\). We will define the corresponding charge \(N_{\mathcal D,\ell}^{\mathcal B}\) and a reference charge \(N_{\mathcal D,\ell}\) with nonzero field coefficients throughout \(\Lambda\). The latter defines the transition energies used in reconstruction. To estimate the physical signal at \(k\), we apply the ADHH recursion with a charge supported only on \(B_R(k)\), and separately bound the effect of the other controlled balls. This gives \Cref{lem:parallel-stopped-dynamics,thm:parallel-Ramsey} below. The comparison with the patch function will use \Cref{prop:parallel-patch}, stated later in this section and proved later in Appendix~\ref{app:parallelization} using only the recursion and locality estimates.  In particular, the proof of \Cref{prop:parallel-patch} does not use \Cref{thm:parallel-Ramsey}.

\subsection{Parallel experiment with separated defects}
\label{subsec:parallel-experiment}

In this subsection we will define a parallel Ramsey experiment with a batch of geometrically separated defect sites, each of which will be measured separately at the end of the experiment. This construction provides the experimental basis for the parallel speedup by producing Ramsey signals for all defects in the batch simultaneously.

Let $\mathcal D\subseteq\Lambda$ be a set of defect sites, and we will call it a \emph{batch}.  Define its minimum separation and the corresponding connected-word threshold by

\begin{equation}
    L_{\mathcal D}:= \min_{\substack{k,k'\in\mathcal D\\k\ne k'}}\dist(k,k'), \quad
    P_{\mathcal D}:=\ceil{{L_{\mathcal D}}/{r_0}}.
    \label{eq:parallel-separation}
\end{equation}
For a singleton set, we use the convention $L_{\mathcal D}=P_{\mathcal D}=+\infty$.  A batch is called $L$-separated if $L_{\mathcal D}\ge L$.

For each batch $\mathcal D$ and repetition $\ell=1,\ldots,m$, we independently draw a Pauli operator and an associated sign for each site, similar to the single-defect case. 
\begin{equation}
    s_j^{\mathcal D,\ell}\in\{0,1\}, \quad
    \beta_j^{\mathcal D,\ell}\in\{X,Y,Z\},
    \label{eq:parallel-random-pattern}
\end{equation}
We then define $\tau_{j,z}^{\mathcal D,\ell}$, $\ket{\phi_{j,b}^{\mathcal D,\ell}}$, and $Q_j^{\mathcal D,\ell}$ as in \Cref{eq:tau-z}--\Cref{eq:Q-row}. Let \(B_k=B_R(k)\) and \(\mathcal B=\bigcup_{k\in\mathcal D}B_k\), with \(L_{\mathcal D}>2R\). The physical field vanishes outside \(\mathcal B\), has half strength at the defects, and full strength at the remaining sites of \(\mathcal B\):
\begin{equation}
    H_{\rm ctrl}^{\mathcal D,\ell} ={\sum_{j\in\mathcal B\setminus\mathcal D}}\tau_{j,z}^{\mathcal D,\ell} +\frac12\sum_{k\in\mathcal D}\tau_{k,z}^{\mathcal D,\ell},
    \label{eq:parallel-control}
\end{equation}
For the reference model with a field on every qubit, define
\begin{equation}
    N_{\mathcal D,\ell} :=\sum_{k\in\mathcal D}Q_k^{\mathcal D,\ell} +2\sum_{j\notin\mathcal D}Q_j^{\mathcal D,\ell}.
    \label{eq:parallel-charge}
\end{equation}
Define the charge supported in \(\mathcal B\) and the corresponding Hamiltonian by
\begin{align}
    N_{\mathcal D,\ell}^{\mathcal B} &:=\sum_{k\in\mathcal D}Q_k^{\mathcal D,\ell} +2\sum_{j\in\mathcal B\setminus\mathcal D}Q_j^{\mathcal D,\ell},\\
    G_{\mathcal D,\ell}^{\mathcal B}(\mu) &:=H(\mu)+\nu N_{\mathcal D,\ell}^{\mathcal B} =H(\mu)-\nu H_{\rm ctrl}^{\mathcal D,\ell}
    +\nu\left(|\mathcal B|-|\mathcal D|/2\right)I.
 \label{eq:parallel-G}
\end{align}
The Hamiltonian  \(G_{\mathcal D,\ell}(\mu):=H(\mu)+\nu N_{\mathcal D,\ell}\) is used as an auxiliary reference Hamiltonian in the analysis. In the subsection below, we construct the effective Hamiltonian using the ADHH recursion with charge \(N_{\mathcal D,\ell}\) and input \(zH(\lambda)\), and expand it in powers of \(z\) around zero. Fix an integer \(1\le n_{\rm par}\le n_*(\nu)\) for all recursions in the parallel protocol. Its value in terms of the target accuracy will be specified below in \Cref{eq:parallel-stopping-order}.

For $A\subseteq\mathcal D$, we define
\begin{equation}
    \ket{\Phi_A^{\mathcal D,\ell}} :=\left(\bigotimes_{k\in A}\ket{\phi_{k,1}^{\mathcal D,\ell}}\right)\left(\bigotimes_{j\in\Lambda\setminus A}\ket{\phi_{j,0}^{\mathcal D,\ell}}\right).
    \label{eq:parallel-code-states}
\end{equation}
These states satisfy $N_{\mathcal D,\ell}\ket{\Phi_A^{\mathcal D,\ell}} =|A|\ket{\Phi_A^{\mathcal D,\ell}}$. In particular, the charge-one sector contains the $|\mathcal D|$ states $\ket{\Phi_{\{k\}}^{\mathcal D,\ell}}$, so commutation with \(N_{\mathcal D,\ell}\) alone does not imply that each \(\ket{\Phi_{\{k\}}^{\mathcal D,\ell}}\) is an eigenstate of the effective Hamiltonian.

We will still choose our initial state to be a product state:
\begin{equation}
    \ket{\Psi_+^{\mathcal D,\ell}} :=\left[ \bigotimes_{k\in\mathcal D} \frac{\ket{\phi_{k,0}^{\mathcal D,\ell}}+ \ket{\phi_{k,1}^{\mathcal D,\ell}}}{\sqrt2} \right] \left[ \bigotimes_{j\notin\mathcal D} \ket{\phi_{j,0}^{\mathcal D,\ell}} \right] =2^{-|\mathcal D|/2} \sum_{A\subseteq\mathcal D}\ket{\Phi_A^{\mathcal D,\ell}}.
    \label{eq:parallel-preparation}
\end{equation}
For each $k\in\mathcal D$, we define $\tau_{k,x}^{\mathcal D,\ell}$, $\tau_{k,y}^{\mathcal D,\ell}$, and
\begin{equation}
    \tau_{k,+}^{\mathcal D,\ell} :=\tau_{k,x}^{\mathcal D,\ell}+i\tau_{k,y}^{\mathcal D,\ell} =2\ket{\phi_{k,1}^{\mathcal D,\ell}} \!\bra{\phi_{k,0}^{\mathcal D,\ell}}
    \label{eq:parallel-O-plus}
\end{equation}
as in \Cref{eq:X-Y-quadrature} and \Cref{eq:O-plus}. The expectation value of each $\tau_{k,+}^{\mathcal D,\ell}$ at the end of the experiment is
\begin{equation}
    S_{\mathcal D,\ell;k}(t) :=\bra{\Psi_+^{\mathcal D,\ell}} {e^{iG_{\mathcal D,\ell}^{\mathcal B}(\lambda)t}} \tau_{k,+}^{\mathcal D,\ell} {e^{-iG_{\mathcal D,\ell}^{\mathcal B}(\lambda)t}} \ket{\Psi_+^{\mathcal D,\ell}}.
    \label{eq:parallel-signal}
\end{equation}
All $\tau_{k,x}^{\mathcal D,\ell}$ with $k\in\mathcal D$ have non-overlapping supports and can be measured simultaneously, which is also true for all $\tau_{k,y}^{\mathcal D,\ell}$.  Hence one family of experiments yields the $\Re S_{\mathcal D,\ell;k}(t)$ at every defect, and a second family yields all $\Im S_{\mathcal D,\ell;k}(t)$.

\subsection{Taylor truncation of the effective Hamiltonian}
\label{subsec:parallel-normal-form}

This subsection uses the connected-word expansion of the effective Hamiltonian introduced in \Cref{subsec:connected-words} to obtain a truncated effective Hamiltonian whose every term involves at most one defect, with truncation error exponentially vanishing in $L_{\mathcal{D}}$ through our choice of truncation order. We will see that under this truncated effective Hamiltonian, the dynamics factorize into independent single-defect oscillations on the defect subspace.

More specifically, we use the connected-word expansion of \Cref{lem:connected-word-expansion}.  Recall that a connected \(p\)-word (\Cref{def:p-word}) consists of \(p\) Pauli operators, each having the same support as a Hamiltonian term, whose support-overlap graph is connected (the individual Pauli supports need not themselves be connected subsets of the lattice).  Since each original support has diameter at most \(r_0\), the union of the supports in any connected \(p\)-word has diameter at most \(p r_0\).

Below we will fix $(\mathcal{D},\ell)$ and omit the $\ell$ dependence for simplicity. Run ADHH recursion \Cref{eq:ADHH-recursion} with charge \(N_{\mathcal D,\ell}\) for the fixed number of steps \(n_{\rm par}\) introduced above. As in the sequential construction, this number is held fixed when the input interaction varies. For an integer \(n_{\rm par}\) that satisfies \(1\le n_{\rm par}\le n_*(\nu)\), we let \(\widehat D_{\mathcal D,\ell}(z):=D_{n_{\rm par}}(zH(\lambda))\) be the charge-preserving effective Hamiltonian obtained by stopping the ADHH recursion after \(n_{\rm par}\) steps. The order \(n_{\rm par}\) is held fixed even if the actual field permits more recursion steps. Since $\norm{2H(\lambda)}_{\kappa_0}\le2J<\Jan$, \Cref{thm:uniform-normal-form} implies that this is an operator-valued analytic function on a neighborhood of the closed disc $|z|\le2$, and consequently we have the following convergent Taylor expansion
\begin{equation}
    \widehat D_{\mathcal D,\ell}(z)=\sum_{p\ge1}z^p\widehat D_{\mathcal D,\ell}^{(p)}.
    \label{eq:parallel-taylor}
\end{equation}
Since $\widehat D_{\mathcal D,\ell}(0)=0$ the above expansion starts from $p=1$. The operator-valued version of the Cauchy estimate in \Cref{lem:Cauchy-general}, together with \Cref{eq:uniform-D-bound}, gives
\begin{equation}
    \norm{\widehat D_{\mathcal D,\ell}^{(p)}}_{\kappa_{n_{\rm par}}}\le C_{\rm T}\Jan \cdot 2^{-p},
    \label{eq:parallel-taylor-coeff}
\end{equation}
for some constant $C_{\rm T}>0$. Consequently, for
\begin{equation}
    \widehat D_{\mathcal D,\ell}^{(<P)}:=\sum_{p=1}^{P-1}\widehat D_{\mathcal D,\ell}^{(p)}, \quad
    \widehat D_{\mathcal D,\ell}^{(\ge P)} :=\widehat D_{\mathcal D,\ell}(1)-\widehat D_{\mathcal D,\ell}^{(<P)},
    \label{eq:parallel-taylor-split}
\end{equation}
we have
\begin{equation}
    \norm{\widehat D_{\mathcal D,\ell}^{(\ge P)}}_{\kappa_{n_{\rm par}}} \le 2C_{\rm T}\Jan \cdot 2^{-P}.
    \label{eq:parallel-tail}
\end{equation}
Every coefficient in \Cref{eq:parallel-taylor} is Hermitian and commutes with $N_{\mathcal D,\ell}$ because \(D_j\) is charge preserving at every recursion step for all $|z|\leq 2$.

The proof of \Cref{lem:connected-word-expansion} applies without change to the multi-defect $\widehat D_{\mathcal D,\ell}(z)$. Hence, we can write the Taylor coefficient $\widehat D_{\mathcal D,\ell}^{(p)}$ as
\begin{equation}
    \widehat D_{\mathcal D,\ell}^{(p)} = \sum_{X} W_{\mathcal{D},X}^{(p)},
\end{equation}
where each $W_{\mathcal{D},X}^{(p)}$ is a linear combination of connected $p$-words, and has support diameter at most $pr_0$. Since $\widehat D_{\mathcal D,\ell}^{(p)}$ is Hermitian, we can choose each $W_{\mathcal{D},X}^{(p)}$ to be Hermitian $[\widehat D_{\mathcal D,\ell}^{(p)},N_{\mathcal{D},\ell}]=0$, we have $\langle\widehat D_{\mathcal D,\ell}^{(p)}\rangle_{N_{\mathcal D,\ell}}=\widehat D_{\mathcal D,\ell}^{(p)}$. Consequently
\begin{equation}
\label{eq:D-taylor-expansion-commuting-support-sum}
    \widehat D_{\mathcal D,\ell}^{(p)} = \sum_{X} \widetilde{W}_{\mathcal{D},X}^{(p)}:= \sum_{X} \langle W_{\mathcal{D},X}^{(p)}\rangle_{N_{\mathcal D,\ell}}.
\end{equation}
Note that {the charge-preserving average does not enlarge the support}, and therefore each $\widetilde{W}_{\mathcal{D},X}^{(p)}$ has support diameter at most $pr_0$. Importantly, $[\widetilde{W}_{\mathcal{D},X}^{(p)},N_{\mathcal{D},\ell}]=0$, so we have the following lemma:
\begin{lemma}
    \label{lem:D-taylor-expansion-commuting-support}
    $\widehat D_{\mathcal D,\ell}^{(p)}$ defined in \Cref{eq:parallel-taylor} can be decomposed according to \Cref{eq:D-taylor-expansion-commuting-support-sum} into a sum of Hermitian terms, each commuting with $N_{\mathcal{D},\ell}$ and has support diameter at most $pr_0$.
\end{lemma}
Therefore, if a term $\widetilde{W}_{\mathcal{D},X}^{(p)}$ in this decomposition of $\widehat D_{\mathcal D,\ell}^{(p)}$ touched two distinct defects $k,k'\in\mathcal D$, $p$ needs to satisfy
\begin{equation}
 \dist(k,k')\le pr_0.
 \label{eq:parallel-bridge-distance}
\end{equation}
It follows from \Cref{eq:parallel-separation} that every term $\widetilde{W}_{\mathcal{D},X}^{(p)}$ of degree $p<P_{\mathcal D}$ touches at most one defect.

We will consider the defect subspace spanned by $\ket{\Phi_A^{\mathcal D,\ell}}$ together with its orthogonal projection operator:
\begin{equation}
    \mathcal H_{\rm def}(\mathcal D,\ell):=\operatorname{span}\left\{\ket{\Phi_A^{\mathcal D,\ell}}:A\subseteq\mathcal D\right\}, \quad \Pi_{\mathcal D} :=\sum_{A\subseteq\mathcal D} \ket{\Phi_A^{\mathcal D,\ell}} \!\bra{\Phi_A^{\mathcal D,\ell}}.
    \label{eq:parallel-code}
\end{equation}
In the lemma below, we will show that the Taylor truncation of $\widehat D_{\mathcal D,\ell}(1)$ below degree $P_{\mathcal D}$ leaves this subspace invariant, and within this subspace the dynamics take a especially simple form.

\begin{lemma}[Factorization on the defect subspace]
\label{lem:parallel-code-factorization}
$\mathcal{H}_{\rm def}(\mathcal{D},\ell)$ is an invariant subspace of $\widehat D_{\mathcal D,\ell}^{(<P_{\mathcal D})}$:
\begin{equation}
    (I-\Pi_{\mathcal D})\widehat D_{\mathcal D,\ell}^{(<P_{\mathcal D})}\Pi_{\mathcal D}=0.
    \label{eq:parallel-no-leakage}
\end{equation}
Moreover, for each $k\in\mathcal D$, define
\begin{equation}
    E_{\mathcal D,\ell;k}^{(<P_{\mathcal D})} := \frac12\left( \bra{\Phi_{\{k\}}^{\mathcal D,\ell}} \widehat D_{\mathcal D,\ell}^{(<P_{\mathcal D})} \ket{\Phi_{\{k\}}^{\mathcal D,\ell}} - \bra{\Phi_{\varnothing}^{\mathcal D,\ell}} \widehat D_{\mathcal D,\ell}^{(<P_{\mathcal D})} \ket{\Phi_{\varnothing}^{\mathcal D,\ell}} \right),
\end{equation}
and then is a real number $c_{\mathcal D}$ such that
\begin{align}
    \Pi_{\mathcal D} \widehat D_{\mathcal D,\ell}^{(<P_{\mathcal D})} \Pi_{\mathcal D} =\Pi_{\mathcal D}\left( c_{\mathcal D}I +2\sum_{k\in\mathcal D} E_{\mathcal D,\ell;k}^{(<P_{\mathcal D})}Q_k^{\mathcal D,\ell} \right)\Pi_{\mathcal D}.
    \label{eq:parallel-factorization}
\end{align}
Therefore, for the evolution generated by $G_{\mathcal D,\ell}^{(<P_{\mathcal D})}:=\nu N_{\mathcal D,\ell}+\widehat D_{\mathcal D,\ell}^{(<P_{\mathcal D})}$, the expectation value of $\tau_{k,+}^{\mathcal D,\ell}$ at every defect site $k$ at time $t$ is
\begin{equation}
    \bra{\Psi_+^{\mathcal D,\ell}} e^{iG_{\mathcal D,\ell}^{(<P_{\mathcal D})}t} \tau_{k,+}^{\mathcal D,\ell} e^{-iG_{\mathcal D,\ell}^{(<P_{\mathcal D})}t} \ket{\Psi_+^{\mathcal D,\ell}} =e^{i[\nu+2E_{\mathcal D,\ell;k}^{(<P_{\mathcal D})}]t}.
    \label{eq:parallel-ideal-tone}
\end{equation}
In the above $Q_k^{\mathcal{D},\ell}$ is defined as in \Cref{eq:Q-row} and $\tau_{k,+}^{\mathcal{D},\ell}$ in \Cref{eq:parallel-O-plus}.
\end{lemma}
We therefore have the desired single-frequency signal $e^{i[\nu+2E_{\mathcal D,\ell;k}^{(<P_{\mathcal D})}]t}$ for each defect site, similar to the single-defect setting described in \Cref{thm:Ramsey-approximation}. This enables us to estimate $E_{\mathcal D,\ell;k}^{(<P_{\mathcal D})}$ from experiments and reconstruct coefficient estimates accordingly. We will discuss the estimation procedure in the next subsection. \Cref{eq:parallel-ideal-tone} concerns evolution under \(G_{\mathcal D,\ell}^{(<P_{\mathcal D})}\). In the next subsection we will compare the signal generated by \(G_{\mathcal D,\ell}^{\mathcal B}\) directly with the scalar \(e^{i[\nu+2E_{\mathcal D,\ell;k}^{(<P_{\mathcal D})}]t}\).

\subsection{Parallel estimation of  transition energies}
\label{subsec:parallel-stability}

We now compare the physical signal with the oscillation in \Cref{eq:parallel-ideal-tone}. The difference between the physical and auxiliary Hamiltonians includes fields on every site outside \(\mathcal B\), so its global operator norm is not the small quantity used in this argument.

The two charges act identically on an operator \(W\) supported inside \(B_R(k)\). Indeed, \(N_{\mathcal D,\ell}-N_{\mathcal D,\ell}^{\mathcal B}=2\sum_{j\notin\mathcal B}Q_j^{\mathcal D,\ell}\), which commutes with \(W\), and the onsite charge operators commute with
one another. It follows that
\begin{equation}
    e^{i\theta N_{\mathcal D,\ell}}We^{-i\theta N_{\mathcal D,\ell}} =e^{i\theta N_{\mathcal D,\ell}^{\mathcal B}} We^{-i\theta N_{\mathcal D,\ell}^{\mathcal B}}, \qquad \supp(W)\subseteq B_R(k).
    \label{eq:parallel-local-charge-agreement}
\end{equation}
Hence the maps \(\mathcal P_{N,\omega}\) and \(\mathcal S_N\) agree for these two charges on such operators. Contributions to the recursion built entirely from input terms in \(B_R(k)\) therefore agree, even though the fields differ outside the ball. Contributions reaching beyond it must still be bounded.

For the dynamical estimate, we fix the measured defect \(k\) and use the field on \(B_R(k)\) to construct an effective Hamiltonian, while retaining the fields on the other controlled balls. The separation between balls bounds the correction induced by this construction. This gives \Cref{lem:parallel-stopped-dynamics} below, which is then used to prove \Cref{thm:parallel-Ramsey}, also stated below. The following general perturbation lemma concerns adding an interaction with a small local norm; the proof of \Cref{thm:parallel-Ramsey} uses the subsequent \Cref{lem:parallel-stopped-dynamics} and does not invoke that general lemma.

\begin{lemma}[Local dynamics under a small interaction]
\label{lem:parallel-local-perturbation}
Let \(A\) and \(B\) be Hermitian interactions satisfying \(\norm A_{2\kappa}\le J_A,\norm B_{2\kappa}\le\delta_B,\) for some \(\kappa>0\). For an operator \(O_x\) supported on one site,
\begin{equation}
    \norm{ e^{i(A+B)t}O_xe^{-i(A+B)t} -e^{iAt}O_xe^{-iAt}} \le C_{d,\kappa}\norm{O_x}\, \delta_B|t|(1+v_A|t|)^d,
    \label{eq:parallel-local-perturbation}
\end{equation}
where \(v_A=C(d)\kappa^{-(d+2)}e^\kappa J_A\) is the Lieb-Robinson velocity. The prefactor \(C_{d,\kappa}\) can be chosen such that
\begin{equation}
    C_{d,\kappa}\le C_d(1+\kappa^{-1})^d.
    \label{eq:parallel-local-perturbation-prefactor}
\end{equation}
Here \(C_d\) and \(C(d)\) depend only on \(d\), and are independent of the system size.
\end{lemma}

With the stopping order $n_{\rm par}$ introduced in \Cref{subsec:parallel-normal-form}, we set \(\kappa_{\rm par}:=\kappa_{n_{\rm par}}/2\). \Cref{eq:parallel-taylor-coeff,eq:parallel-tail} imply
\begin{equation}
    \norm{\widehat D_{\mathcal D,\ell}^{(<P_{\mathcal D})}}_{2\kappa_{\rm par}} \le C\Jan, \quad
    \norm{\widehat D_{\mathcal D,\ell}^{(\ge P_{\mathcal D})}}_{2\kappa_{\rm par}}\le C\Jan2^{-P_{\mathcal D}}.
    \label{eq:parallel-normal-form-norms}
\end{equation}
We then define the corresponding Lieb--Robinson velocity as in \Cref{lem:parallel-local-perturbation}.
\begin{equation}
    v_*:=C(d)\kappa_{\rm par}^{-(d+2)}e^{\kappa_{\rm par}}\Jan.
    \label{eq:parallel-vstar}
\end{equation}

For each \(k\in\mathcal D\), put \(B_k=B_R(k)\) and define
\begin{equation}
    N_k:=Q_k^{\mathcal D,\ell}+2\sum_{j\in B_k\setminus\{k\}}Q_j^{\mathcal D,\ell}, \qquad
    N_{{\rm rem},k}:=N_{\mathcal D,\ell}^{\mathcal B}-N_k.
    \label{eq:parallel-single-ball-charge}
\end{equation}
Here \(N_k\) is supported on \(B_k\), and \(N_{{\rm rem},k}\) is supported on the other controlled balls. In particular, \(G_{\mathcal D,\ell}^{\mathcal B}=H+\nu N_k+\nu N_{{\rm rem},k}\), where \(H=H(\lambda)\) and \(\norm H_{\kappa_0}\le J\). Run \Cref{eq:ADHH-recursion} with input \(H\), charge \(N_k\), and \(s=n_{\rm par}\) steps, and denote the resulting operators by \(Y_{k,s},D_{k,s},V_{k,s}\). Iterating the defining conjugations gives \(Y_{k,s}(H+\nu N_k)Y_{k,s}^\dagger=\nu N_k+D_{k,s}+V_{k,s}\). When we apply this same unitary to the physical Hamiltonian, we must also transform \(\nu N_{{\rm rem},k}\). The extra remainder is \(\nu(Y_{k,s}N_{{\rm rem},k}Y_{k,s}^\dagger-N_{{\rm rem},k})\). It is bounded in the proof below by using the separation of the controlled balls. Keeping this term separate allows the recursion itself to use the input \(H\), whose local norm is independent of \(\nu\).

\begin{lemma}[Single-ball comparison in a parallel experiment]
\label{lem:parallel-stopped-dynamics}
Assume \(\nu\ge\max\{\bar\nu_0,54\pi\Jan/\kappa_0\}\), \(1\le s=n_{\rm par}\le n_*(\nu)\), and \(L_{\mathcal D}\ge2R+(P+1)r_0\), with integer \(P\ge1\). Set \(\mathcal K_k:=\nu N_k+D_{k,s}+\nu N_{{\rm rem},k}\). For every bounded operator \(O\) and every real \(t\),
\begin{equation}
    \norm{e^{iG_{\mathcal D,\ell}^{\mathcal B}t}Oe^{-iG_{\mathcal D,\ell}^{\mathcal B}t}-e^{i\mathcal K_kt}Oe^{-i\mathcal K_kt}}\le C\norm O\Jan|B_k|\left\{\nu^{-1}+|t|\left[(2/3)^s+(P+1)2^{-P}\right]\right\}.
    \label{eq:parallel-stopped-dynamics}
\end{equation}
The constant is independent of \(n,|\mathcal D|,s,P,\nu\).
\end{lemma}

The three terms in this bound have separate origins. The term \(\nu^{-1}\) comes from \(Y_{k,s}-I\), the factor \((2/3)^s\) bounds the remainder \(V_{k,s}\), and \((P+1)2^{-P}\) bounds the change in the fields on the other balls under conjugation by \(Y_{k,s}\). In the last estimate, terms of Taylor degree at most \(P\) cannot reach from \(B_k\) to another controlled ball and therefore commute with \(N_{{\rm rem},k}\). The proof, given below in Appendix~\ref{app:parallelization}, sums the remaining coefficients using their exponential decay. Each estimate involves \(|B_k|\), because the recursion is performed separately for the charge \(N_k\) associated with the measured defect.

To deduce a single-frequency signal, we next restrict \(\mathcal K_k\) to the eigenspaces of \(N_k\) with eigenvalues zero and one. These eigenspaces still contain the entire exterior Hilbert space. Their two exterior Hamiltonians differ by approximately the scalar \(2F_{\rho,R}^{\rm par}(\lambda)\), where the patch function is defined below in \Cref{eq:parallel-patch-map}. This comparison controls the relative phase of the two components of the Ramsey state, even though the exterior continues to evolve.

\begin{theorem}[Single-frequency approximation for multiple defects]
\label{thm:parallel-Ramsey}
Fix integers \(P\ge1\) and \(K\ge0\). Assume $\nu\ge\max\{\bar\nu_0,54\pi\Jan/\kappa_0\}$, $1\le n_{\rm par}\le n_*(\nu)$, $L_{\mathcal D}\ge2R+(P+1)r_0$, and $R\ge(K+1)r_0$. For every \(k\in\mathcal D\)
and every real \(t\), the signal in \Cref{eq:parallel-signal} satisfies
\begin{equation}
    \abs{S_{\mathcal D,\ell;k}(t) -e^{i[\nu+2E_{\mathcal D,\ell;k}^{(<P_{\mathcal D})}]t}} \le \frac{C_1\Jan V_R}{\nu} +C_2\Jan V_R|t|\left[(2/3)^{n_{\rm par}}+(P+1)2^{-P}\right] +C_3\Jan|t|2^{-K}.
    \label{eq:parallel-Ramsey-error}
\end{equation}
The same right-hand side bounds
\(\abs{S_{\mathcal D,\ell;k}(t)-e^{i[\nu+2F_{\rho,R}^{\rm par}(\lambda)]t}}\), where \(F_{\rho,R}^{\rm par}\) is the patch function defined below in \Cref{eq:parallel-patch-map}, with \(\rho=(k,\ell)\) and matching Pauli axes on \(B_R(k)\). All constants are independent of \(n,|\mathcal D|,n_{\rm par},L_{\mathcal D},\nu\).
\end{theorem}

The proof is given below in Appendix~\ref{app:parallelization}. It first bounds error of approximating the physical signal using \(F_{\rho,R}^{\rm par}(\lambda)\). It then uses \Cref{prop:parallel-patch}, stated later in this section, to replace \(F_{\rho,R}^{\rm par}(\lambda)\) by \(E_{\mathcal D,\ell;k}^{(<P_{\mathcal D})}\). The proposition bounds their difference by \(C\Jan2^{-K}\), so the resulting change in the oscillation is at most \(C\Jan|t|2^{-K}\), which is absorbed into the last term of \Cref{eq:parallel-Ramsey-error}. The proof of \Cref{prop:parallel-patch} uses coefficient matching and Taylor coefficient bounds and does not use the present theorem; this forward reference is therefore not circular.

This theorem implies that we can estimate $E_{\mathcal D,\ell;k}^{(<P_{\mathcal D})}$ from experiments for multiple $k\in \mathcal{D}$ in parallel using the robust frequency estimation procedure in \Cref{lem:robust-frequency}. We then summarize the cost of extracting these frequencies in the theorem below.

\begin{theorem}[Parallel estimation of  transition energies]\label{thm:parallel-estimate-half-gaps}
Fix a nonempty batch \(\mathcal D\), a repetition \(\ell\), a sufficiently small accuracy \(\eta>0\), and a failure probability \(\delta_{\rm par}\in(0,1)\). Assume that the actual field strength $\nu$ obeys  the lower bound
\begin{equation}
    \nu\ge\nu_{\rm par}=O\!\left(J\max\left\{[\log(eJ/\eta)]^d,(d+1)\log(eJ/\eta)[1+\log((d+1)\log(eJ/\eta))]^3\right\}\right).
    \label{eq:parallel-field-scaling}
\end{equation}
Then, with \(K=P=n_{\rm par}\) and \(R=(K+1)r_0\) specified below in \Cref{eq:parallel-stopping-order}, there exists a sufficient separation threshold \(P_{\rm sep}=\bigO((d+1)\log(1/\eta))\), corresponding to a minimum distance of order \(\bigO(r_0(d+1)\log(1/\eta))\), such that \(L_{\mathcal D}\ge r_0P_{\rm sep}\) allows simultaneous estimation of all {effective half transition energies} \(E_{\mathcal D,\ell;k}^{(<P_{\mathcal D})}\), \(k\in\mathcal D\), to precision \(\eta\) with probability at least {\(1-\delta_{\rm par}\)}, using \(\bigO(\log^2(1/\eta)\log(2|\mathcal D|/\delta_{\rm par}))\) independent non-adaptive global experiments, \(\bigO(\log(2|\mathcal D|/\delta_{\rm par})/\eta)\) total evolution time.
\end{theorem}

We choose the ADHH recursion stopping order to be
\begin{equation}
    n_{\rm par}:=\ceil{C_{\rm par}(d+1)\log(eJ/\eta)},\quad K=P=n_{\rm par},\quad R=(K+1)r_0,
    \label{eq:parallel-stopping-order}
\end{equation}
where \(C_{\rm par}>0\) is a sufficiently large constant. We also assume that the actual field strength $\nu$ obeys the lower bound
\begin{equation}
    \nu\ge\nu_{\rm par}:=\max\left\{{54\pi\Jan}/{\kappa_0},\,{C_{\rm vol}\Jan V_R,\,}C_{\rm fld}\bar\nu_0(n_{\rm par}+2)[1+\log(n_{\rm par}+2)]^3\right\}.
    \label{eq:parallel-field-condition}
\end{equation}
Here \(C_{\rm vol},C_{\rm fld}>0\) are sufficiently large constants. The volume term controls the first term in \Cref{eq:parallel-Ramsey-error}: by \Cref{eq:finite-ball-dressing}, \(\norm{Y_{k,n_{\rm par}}-I}\le C\Jan|B_k|/\nu\), and \(|B_k|\le V_R\). Thus \(\nu\ge C_{\rm vol}\Jan V_R\) makes the error from this unitary transformation smaller than a prescribed constant. The last term in the maximum ensures \(n_{\rm par}\le n_*(\nu)\), with \(n_*(\nu)\) from \Cref{eq:new-n*}, so the recursion estimates are valid through all the chosen steps. The first term supplies the remaining smallness condition \(\nu\ge54\pi\Jan/\kappa_0\).

The estimator uses evolution times at most \(C_T/\eta\), where \(C_T\) is fixed. On this time interval, the other terms in \Cref{eq:parallel-Ramsey-error} are small because the choices \(K=P=n_{\rm par}\) make \((2/3)^{n_{\rm par}}\), \((P+1)2^{-P}\), and \(2^{-K}\) decay faster than the factors \(V_R/\eta\) grow. The quantitative estimates are given below in the proof of \Cref{thm:parallel-estimate-half-gaps} in Appendix~\ref{app:parallelization}. Together, these choices keep the deterministic signal error below the fixed tolerance of \Cref{lem:robust-frequency}; precision \(\eta\) is obtained by
using times of order \(1/\eta\).

The value of \(n_{\rm par}\) in \Cref{eq:parallel-stopping-order} is fixed by \(\eta\). We use that same number of steps for the full-system recursions with charges \(N_{\mathcal D,\ell}\) and \(N_k\), and for the patch recursion on \(B_k\). The two charges agree on \(B_k\), because that ball contains only the defect \(k\), so they define the same charge for the patch recursion. Using a common number of steps makes the coefficients from the reference, physical comparison, and patch calculations comparable at each Taylor degree. When the trial coefficients vary, we keep \(n_{\rm par}\) fixed so that each calculation defines one analytic function of those coefficients. Increasing \(\nu\) enlarges the allowed maximum \(n_*(\nu)\) and lowers the bound \(C\Jan|B_k|/\nu\) on \(\norm{Y_{k,n_{\rm par}}-I}\); it does not change the chosen number of steps used to evaluate the patch function.

\subsection{Final complexity}
\label{subsec:parallel-resources}

In the previous subsection we showed that all effective half transition energies in \(\mathcal D\) can be estimated simultaneously. As in the single-defect analysis, we compare the effective transition energy with one obtained by applying the same ADHH construction on a local patch $B_R(k)$. Both \(E_{\mathcal D,\ell;k}^{(<P_{\mathcal D})}\) and the patch function defined next use exactly \(n_{\rm par}\) recursion steps. The superscript \((<P_{\mathcal D})\) denotes a separate Taylor truncation of the reference interaction: only degrees \(p<P_{\mathcal D}\) are retained. The patch function uses the output of the \(n_{\rm par}\)-step recursion on \(B_R(k)\) without this Taylor truncation. Their low-degree coefficients agree because every contribution to the transition at \(k\) of degree \(p\le K\) is supported within \(B_R(k)\), where the reference charge equals the patch charge. The remaining coefficients are bounded by geometric tails. This is the content of \Cref{prop:parallel-patch}, stated below and proved later in Appendix~\ref{app:parallelization}. We let \(D_{\rho,R}^{\rm par}(\mu)\) be the charge-preserving interaction obtained by running the patch recursion of \Cref{eq:patch-charge} for \(n_{\rm par}\) steps, and define
\begin{equation}
    {F_{\rho,R}^{\rm par}(\mu) :=\frac12\left( \bra{\Phi_1^\rho}D_{\rho,R}^{\rm par}(\mu)\ket{\Phi_1^\rho} -\bra{\Phi_0^\rho}D_{\rho,R}^{\rm par}(\mu)\ket{\Phi_0^\rho} \right).}
    \label{eq:parallel-patch-map}
\end{equation}
Here the patch basis is specified by \(s_j^\rho=s_j^{\mathcal D,\ell}\) and \(\beta_j^\rho=\beta_j^{\mathcal D,\ell}\) for \(j\in B_R(k)\), where \(\rho=(k,\ell)\). No basis choice outside the patch is needed.

\begin{proposition}
\label{prop:parallel-patch}
Let $k\in\mathcal D$, let $\rho=(k,\ell)$, and let \(s_j^\rho=s_j^{\mathcal D,\ell}\), \(\beta_j^\rho=\beta_j^{\mathcal D,\ell}\), for every \(j\in B_R(k)\). We assume $R\ge(K+1)r_0$, and $R<L_{\mathcal D}$. Then $P_{\mathcal D}\ge K+2$ and
\begin{equation}
    \abs{ E_{\mathcal D,\ell;k}^{(<P_{\mathcal D})} {-F_{\rho,R}^{\rm par}(\lambda)}} \le C_{\rm P}^{\rm par}\Jan2^{-K},
    \label{eq:parallel-patch-error}
\end{equation}
for a universal constant $C_{\rm P}^{\rm par}$, uniformly in $n$, $|\mathcal D|$, and the defect locations.
\end{proposition}

Therefore, to control the approximation error in \Cref{eq:parallel-patch-error}, we choose {$K=\mathcal{O}(\log(1/\eps))$} that satisfies  {\(C_{\rm P}^{\rm par}\Jan2^{-K} \le \eps/(4C_\Gamma C_y)\)}, and set  {\(\eta=\eps/(4C_\Gamma C_y)\)} and {\(\delta_{\rm par}=\delta|\mathcal D|/(2nm)\)} in \Cref{thm:parallel-estimate-half-gaps}. For fixed {$\ell$}, we simultaneously obtain the estimate of  {\(F_{\rho,R}^{\rm par}(\lambda)\)} for all $k\in\mathcal D$ within  {\(\eps/(2C_\Gamma C_y)\)-additive} error using total evolution time {\(\mathcal{O}(\log(4nm/\delta)/\eps)\)} with probability {\(1-\delta_{\rm par}\)}. The same separation must keep the controlled balls disjoint and suppress the effect of their fields on one another. Choose \(L_{\mathcal D}\ge2R+(P+1)r_0\), so the distance between any two balls is at least \((P+1)r_0\). With \(P=n_{\rm par}\) and \(R=(n_{\rm par}+1)r_0\) from \Cref{eq:parallel-stopping-order}, the right-hand side equals \(3(n_{\rm par}+1)r_0\). Since \(n_{\rm par}=O((d+1)\log(eJ/\eta))\) and the preceding choice has \(\eta=\Theta(\eps)\), a sufficient separation has scaling
\begin{equation}
    L_{\mathcal D} = \bigO(\log(1/\eps)).
    \label{eq:min-separation}
\end{equation}
We note that $L_{\mathcal D}$ also needs to satisfy another constraint \Cref{eq:parallel-separation-estimation}, but the above scaling is still valid. Since $L_{\mathcal{D}}>2R$, $B_R(k)$ and $B_{R}(k')$ do not overlap for any distinct $k,k'\in \mathcal{D}$,  there is no conflict in the choice of $s_j^\rho$ and $\beta_j^\rho$ for each $j$. Moreover, \(F_{\rho,R}^{\rm par}\) has the same zeroth-order Jacobian \(X\) as the original patch map. The scalar-correction and analytic derivative bounds used in \Cref{sec:optimization} are uniform in the stopping order \(j\le n_*(\nu)\), so the conditioning, stability, and convergence proofs apply to \(F_{\rho,R}^{\rm par}\) with the same thresholds. Measurement outcomes at different defects may be correlated, but because we are analyzing the failure probability using a worst-case union bound, which accounts for the worst effect of correlation, we do not need to change the convergence analysis.

To recover the coefficient by using the optimization process of \Cref{sec:optimization}, the estimate of  {\(F_{\rho,R}^{\rm par}(\lambda)\)} within $\eps$-additive error for all $k=1,2,\dots,n$ and $l=1,2,\dots,m$ are required. Hence, we distribute every {defect site} $k\in \Lambda$ into batches $\mathcal D^{(1)},\dots,\mathcal D^{(\chi)}$ that the minimum separation of each batch satisfies \Cref{eq:min-separation}.

We will design our experiments with the minimum separation distance in \Cref{eq:min-separation} being at least $L$, where $L$ is sufficiently large to ensure $L>2R$ and satisfies \Cref{eq:parallel-separation-estimation}, with scaling $L=\bigO(\log(1/\eps))$. To process all {defect sites} $k$ with as little overhead as possible, we color the $d$-dimensional lattice $\Lambda$ using
\begin{equation}
    \chi\le C_d(1+L)^d=\bigO(L^d) = \bigO(\log^d(1/\eps))
    \label{eq:parallel-color-count}
\end{equation}
color classes, such that defect sites of the same color are at least distance $L$ apart. Defect sites $k$ of the same color are collected into a batch $\mathcal D^{(r)}$, for $r=1,2,\cdots,\chi$. Hence, every batch $\mathcal D^{(r)}$ has minimum separation at least $L_{\mathcal{D}^{(r)}}\geq L$.  

For every color class and every repetition $\ell=1,\ldots,m$, we perform one parallel experimental estimation of the forward map {\(F_{\rho,R}^{\rm par}(\lambda)\)}. Choose \(m=\ceil{C_m{\log(4n/\delta)}}\), reserving failure probability \(\delta/2\) for \Cref{eq:Gamma-close}. Since \(\sum_{r=1}^{\chi}|\mathcal D^{(r)}|=n\), summing the budgets \(\delta|\mathcal D^{(r)}|/(2nm)\) over all batches and repetitions gives total frequency-estimation failure probability at most \(\delta/2\). Thus conditioning and accurate data hold jointly with probability at least \(1-\delta\). Since \(\log(4nm/\delta)=\bigO(\log(n/\delta))\), the total evolution time becomes
\begin{equation}
    T_{\rm tot} =\bigO(\chi \cdot m \cdot (1/\eps)\cdot {\log(4nm/\delta)}) =\bigO(\log^d(1/\eps)\log^2(n/\delta)/\eps)
    \label{eq:parallel-sensing-time}
\end{equation}
For each $(\ell,r)$, we use $\bigO(\log^{2}(1/\eps)\log(n/\delta))$ independent non-adaptive global experiments as in \Cref{thm:parallel-estimate-half-gaps}. Hence, the number of experiments becomes
\begin{equation}
    \bigO(\chi \cdot m \cdot \log^{2}(1/\eps)\cdot \log(n/\delta)) =\bigO(\log^{d+2}(1/\eps)\log^2(n/\delta)).
\end{equation}

The optimization field threshold remains \(\nu_{\rm opt}=\bigO(JR^d)\), as in the single-defect case, because \(F_{\rho,R}^{\rm par}\) has the same zeroth-order Jacobian \(X\), while the error bounds for the forward map and the derivatives (in \Cref{lem:scalar-correction} and \Cref{lem:row-derivatives} respectively) used in the conditioning analysis hold uniformly for every stopping order \(n_{\rm par}\le n_*(\nu)\). Moreover, upon setting \(\eta=\Theta(\eps)\), \Cref{eq:parallel-field-scaling} shows that \(\nu_{\rm par}\) has the same scaling as the single-defect dynamical threshold \(\nu_{\rm dyn}\) (as in \Cref{eq:field-strength-requirement}). Consequently, the combined parallel threshold \(\nu_{\min}^{\rm par}=\max\{\nu_{\rm opt},\nu_{\rm par}\}\) has the same scaling as \(\nu_{\min}\) in \Cref{eq:nu-condition}. Hence, we have the final theorem.

\begin{theorem}[{Learning a geometrically local Hamiltonian with parallel defects}]\label{thm:main-para}
Let \(H(\lambda)\) be the Hamiltonian in \Cref{eq:intro-H} on the box lattice \(\Lambda\) defined in \Cref{sec:intro}, with \(|\Lambda|=n\), satisfying the locality conditions in \Cref{eq:locality-bounds-H}. Assume that \(d,q,r_0\) are fixed and \(\norm{\lambda}_\infty\le1\). Suppose that we can apply the static single-qubit controls \(-\nu H_{\rm ctrl}^{\mathcal D,\ell}\) specified in \Cref{eq:parallel-random-pattern,eq:parallel-control}, prepare the product states in \Cref{eq:parallel-preparation}, and measure the single-qubit Pauli observables \(\tau_{k,x}^{\mathcal D,\ell}\) and \(\tau_{k,y}^{\mathcal D,\ell}\) appearing in \Cref{eq:parallel-O-plus}. For any target accuracy \(\eps\in(0,1)\) and failure probability \(\delta\in(0,1)\), the parallel protocol produces an estimate \(\widehat\lambda\) satisfying \(\norm{\widehat\lambda-\lambda}_\infty\le\eps\) with probability at least \(1-\delta\), using total evolution time $\bigO(\log^d(1/\eps)\log^2(n/\delta)/\eps)$. It suffices to use any field strength $\nu\geq \nu^{\rm par}_{\min}$, where $ {\nu_{\min}^{\rm par}} =\bigO (\max\{\log^d(1/\eps), \log(1/\eps)[1+\log\log(1/\eps)]^3\})$ The protocol uses $\bigO(\log^{d+2}(1/\eps)\log^2(n/\delta))$ independent non-adaptive global experiments. The classical post-processing time is polynomial in \(n\), \(1/\eps\), and \(\log(1/\delta)\), independently of how large the actual field $\nu$ is.
\end{theorem}

The parallel reconstruction fits the measured transition energies using \(F_{\rho,R}^{\rm par}\), evaluated after \(n_{\rm par}\) recursion steps. The difference between these data and the patch values is controlled by the estimation error and \Cref{eq:parallel-patch-error}. The sequential reconstruction uses \(F_{\rho,R}\), evaluated after \(s\) steps. For a fixed  \(\rho\), both patch functions are constructed from the same local Hamiltonian and onsite axes; only the prescribed recursion order changes. The derivative bounds and conditioning estimates used for reconstruction hold uniformly over every allowed order. Because the defects within each batch are sufficiently separated, the forward-model row associated with any one defect is determined locally and has the same conditioning structure as its single-defect counterpart, up to the controlled patch and dynamical errors. Consequently, the same connected-cluster evaluation and fixed-preconditioned iterative fitting procedure applies, with the sequential order \(s\) replaced by \(n_{\rm par}\). Since \(n_{\rm par}=\bigO(\log(1/\eps))\), the resulting classical runtime is polynomial in \(n\), \(1/\eps\), and \(\log(1/\delta)\), independently of the actual field strength \(\nu\). See Appendix~\ref{app:classical-evaluation} for a detailed discussion.

\section*{Statement on the Use of Artificial Intelligence}
The Hamiltonian learning protocol, the idea of applying ADHH prethermalization analysis to the protocol, and the overall proof framework, including the identification and formulation of the principal theorems and lemmas, were developed entirely by the authors. Generative artificial intelligence tools (ChatGPT~5 and 6) assisted the authors in deriving and refining some of the proofs and in drafting and revising portions of the manuscript. The authors reviewed and verified all AI-assisted arguments and text and take full responsibility for the content of this work. 

\section*{Acknowledgments}
We thank John Preskill for helpful discussion. M.S. acknowledges support from the Caltech Summer Undergraduate Research Fellowship and the KAIST Presidential Fellowship. J.L. acknowledges support from the Harvard Quantum Initiative and the Kwanjeong Foundation, and this work benefited from interactions with the NSF AI Institute for Artificial Intelligence and Fundamental Interactions (IAIFI), which is supported by the National Science Foundation under Cooperative Agreement PHY-2525568. Y.T. acknowledges support from the U.S. Department of Energy, Office of Science, Accelerated Research in Quantum Computing Centers, Quantum Utility through Advanced Computational Quantum Algorithms, grant no. DE-SC0025572. I.M. acknowledges support from NSF PHY-2046195 and NSF QLCI grant OMA-2120757.

\bibliographystyle{ieeetr}
\bibliography{note/References}

\newpage

\appendix

\section{Proof of analytic continuation (\Cref{sec:prelim})}
\label{app:prelim-proofs}

\begin{theorem}[\Cref{thm:uniform-normal-form}, Uniform analytic bounds for the ADHH recursion]
We assume $\nu>0$ satisfies {\(\nu\ge\bar\nu_0\)},
\begin{equation}
\label{eq:nu-requirement-app}
    \nu \ge 54\pi J_{\rm an}/\kappa_0,
\end{equation}
and $n_* \ge 1$, with $n_*$ defined in \Cref{eq:new-n*}. The ADHH recursion \Cref{eq:ADHH-recursion} is well defined for every \(Z\in\mathfrak B_{\rm an}\), including non-Hermitian \(Z\), and obeys
\begin{align}
    \norm{D_j}_{\kappa_j}&\le C\Jan,\\
    \norm{V_j}_{\kappa_j}&\le C\Jan(2/3)^j,\\
    \norm{A_j}_{\kappa_j}&\le {C\Jan(2/3)^j}/{\nu},
\end{align}
where $0\le j\le n_*$. The maps \(Z\mapsto D_j(Z),V_j(Z),A_j(Z)\) are analytic on the interior of \(\mathfrak B_{\rm an}\).  
\end{theorem}

\begin{proof}
We will prove that we can extend the ADHH recursion and estimates to non-Hermitian interactions. This is achieved by repeating the finite-order local-norm induction used in Ref.~\cite{abanin2017rigorous} in the complex interaction spaces.

For \(0\le j\leq n_*\), introduce the notation $d_j:=\norm{D_j}_{\kappa_j}$, $v_j:=\norm{V_j}_{\kappa_j}$, and $a_j:=\norm{A_j}_{\kappa_j}$. Set $\delta\kappa_j:=\kappa_j-\kappa_{j+1}$, and $m_{\rm A}(j):={18}/{(\delta\kappa_j\,\kappa_{j+1})}$. Thus \(m_{\rm A}(j)\) is the locality-loss factor denoted $m(j)$ in ADHH and appearing in
\Cref{lem:ADHH-conjugation}. 

Define $\gamma_j(W):=e^{A_j}We^{-A_j}$, and $\alpha_j(W):=\int_0^1 e^{tA_j}We^{-tA_j}\,dt$. Ref.~\cite{abanin2017rigorous} proves that Duhamel's formula gives
\begin{align}
    H_{j+1} &=\gamma_j(D_j)+\gamma_j(V_j)-\alpha_j(V_j)=D_j+W_j,
    \label{eq:uniform-proof-H-step}
\end{align}
where
\begin{equation}
    W_j :=\bigl(\gamma_j(D_j)-D_j\bigr) +\bigl(\gamma_j(V_j)-V_j\bigr) -\bigl(\alpha_j(V_j)-V_j\bigr).
    \label{eq:uniform-proof-Wj}
\end{equation}
Note that the Duhamel's formula does not require the Hermiticity of $A_j$. Thus every change produced in one ADHH step is expressed as a sum of changes  induced by conjugation with \(A_j\). We will apply \Cref{lem:ADHH-conjugation} to bound these changes in the local norm, and these bounds indicate that the changes are local. An important point is that the proof of that lemma in Ref.~\cite{abanin2017rigorous} does not rely on the Hermiticity of $Z$ or unitarity of $Q$. Therefore, the same estimate applies without modification to complex
interactions.

\medskip
\noindent
\textbf{Proving the bounds by induction.} We will then prove the following by induction on $j$:
\begin{equation}
    \label{eq:inductive-proof-goals-analytic-extension}
    3a_j\le\delta\kappa_j,\quad d_j+2v_j\le 5J_{\rm an}, \quad v_j\leq v_0(2/3)^j ,
\end{equation}
for all $0\leq j\leq n_*$. Throughout the proof we will set $T=\frac{2\pi}{\nu}$, and use the following to control $a_j$:
\begin{equation}
\label{eq:uniform-proof-A-bound}
    a_j = \norm{\mathcal S_N(V_j)}_{\kappa_j} \le \frac1T \int_0^Tdt \int_0^t ds\, \norm{e^{is\nu N}V_je^{-is\nu N}}_{\kappa_j} \le \frac{\pi v_j}{\nu}.
\end{equation}
This estimate follows from the definition of $\mathcal{{S}_N}$ in \Cref{eq:homological-inverse}. It does not require \(V_j\) to be Hermitian, since the conjugation by \(e^{is\nu N}\) preserves the support and the local norm because \(N\) is an onsite Hermitian charge operator.

First, for $j=0$, we only need to prove $3a_0\leq \delta\kappa_0$. Because we require $\nu \geq {54\pi J_{\rm an}}/{\kappa_0}\geq {27\pi v_0}/{\kappa_0}$ as stated in \Cref{eq:nu-requirement-app}, we have
\begin{equation}
    3a_0\leq \frac{3\pi v_0}{\nu}\leq \frac{\kappa_0}{9}\leq \left(1-\frac{1}{1+\log(2)}\right)\kappa_0=\delta\kappa_0.    
\end{equation}

Next, we assume \Cref{eq:inductive-proof-goals-analytic-extension} holds for $j$ and prove it for $j+1$, for $j<n_*$. By the inductive hypothesis $3a_j\le\delta\kappa_j$ \Cref{lem:ADHH-conjugation} gives $\norm{\gamma_j(D_j)-D_j}_{\kappa_{j+1}} \le m_{\rm A}(j) a_jd_j$, and $\norm{\gamma_j(V_j)-V_j}_{\kappa_{j+1}} \le m_{\rm A}(j) a_jv_j$. The same estimate applied to \(tA_j\), followed by integration over \(0\le t\le1\), gives
\begin{align}
    \norm{\alpha_j(V_j)-V_j}_{\kappa_{j+1}} \le \int_0^1 \norm{ e^{tA_j}V_je^{-tA_j}-V_j }_{\kappa_{j+1}}\,dt \le \int_0^1 t\,m_{\rm A}(j)a_jv_j\,dt \le \frac{m_{\rm A}(j)a_jv_j}{2}.
    \label{eq:uniform-proof-alpha-V}
\end{align}
Substituting these estimates into
\Cref{eq:uniform-proof-Wj} yields
\begin{equation}
    \norm{W_j}_{\kappa_{j+1}} \le m_{\rm A}(j) a_j \left(d_j+\frac32v_j\right) \le m_{\rm A}(j)a_j(d_j+2v_j).
    \label{eq:uniform-proof-W-bound}
\end{equation}

Since \(D_{j+1}=\langle H_{j+1}\rangle_N\), we have $D_{j+1}-D_j=\langle W_j\rangle_N$, and $V_{j+1}=W_j-\langle W_j\rangle_N$. Using $\norm{\langle W\rangle_N}_{\kappa}\le\norm W_{\kappa}$ and $\norm{W-\langle W\rangle_N}_{\kappa}\le2\norm W_{\kappa}$ (stated as \Cref{eq:projection-norms} in the main text), we obtain
\begin{equation}
\label{eq:uniform-proof-D-V-increment}
    \norm{D_{j+1}-D_j}_{\kappa_{j+1}} \le m_{\rm A}(j) a_j(d_j+2v_j), \quad v_{j+1} \le 2m_{\rm A}(j)a_j(d_j+2v_j).
\end{equation}

\noindent
\textbf{Bound for $V_{j+1}$.} We then prove $v_{j+1}\leq v_0 (2/3)^{j+1}$.
We can relate $v_{j+1}$ to $v_j$ using \Cref{eq:uniform-proof-A-bound}:
\begin{equation}
\label{eq:shrink-vj-intermediate}
    v_{j+1}\leq {2\pi m_{\rm A}(j)(d_j+2v_j)v_j}/{\nu}.
\end{equation}
Note that, using an elementary inequality implicit in the paragraph below Eq.~(4.12) in \cite{abanin2017rigorous}, we have
\begin{equation}
    m_{\rm A}(j)\leq \frac{18(j+2)(1+\log(j+2))^3}{\kappa_0^2}\leq \frac{18 \nu}{\kappa_0^2 \bar\nu_0},    
\end{equation}
for all $0\leq j<n_*$, where the second inequality holds due to our choice of $n_*$ in \Cref{eq:new-n*}. With our current definition of $\bar\nu_0$ in \Cref{eq:uniform-nu0}, we then have
\begin{equation}
    m_{\rm A}(j)\leq \frac{18 \nu}{\kappa_0^2 }\frac{\kappa_0^2}{270\pi J_{\rm an}} = \frac{\nu}{15 \pi J_{\rm an}}.    
\end{equation}

Using the induction hypothesis $d_j+2v_j\leq 5 J_{\rm an}$, we have
\begin{equation}
    {2\pi m_{\rm A}(j)(d_j+2v_j)}/{\nu} \le 2/3.
    \label{eq:uniform-proof-smallness}
\end{equation}
Substituting this into \Cref{eq:shrink-vj-intermediate}, we then obtain $v_{j+1}\leq 2v_j/3$. This proves
\begin{equation}
\label{eq:vj+1-bound-appendix}
    v_{j+1}\leq v_0 (2/3)^{j+1}.
\end{equation}

\medskip
\noindent
\textbf{Bound for $D_{j+1}$.} Next, we derive a bound for $d_{j+1}$ to prove $d_{j+1}+2v_{j+1}\leq 5 J_{\rm an}$.  Note that from \Cref{eq:uniform-proof-D-V-increment} and \Cref{eq:uniform-proof-smallness}, we have
\begin{equation}
    \label{eq:D-increment-bound}
    \|D_{j+1}-D_j\|_{\kappa_{j+1}}\leq \frac{v_j}{3}\leq \frac{v_0}{3}\left(\frac{2}{3}\right)^{j},
\end{equation}
where the second inequality follows from the induction hypothesis. Because $\kappa_{j+1}\leq \kappa_j$, we have $\|D_j\|_{\kappa_{j+1}}\leq \|D_j\|_{\kappa_{j}}$. Therefore, using the triangle inequality
\begin{align}
    d_{j+1}-d_j&=\|D_{j+1}\|_{\kappa_{j+1}}-\|D_{j}\|_{\kappa_{j}}\\
    & \leq |D_{j+1}\|_{\kappa_{j+1}}-\|D_{j}\|_{\kappa_{j+1}} \\
    & \leq \|D_{j+1}-D_j\|_{\kappa_{j+1}}.    
\end{align}
Using \Cref{eq:D-increment-bound} we then have $d_{j+1}-d_j\leq v_j/3$. This ensures
\begin{equation}
    d_{j+1}+2v_{j+1}\leq d_j + \frac{v_j}{3} + \frac{4v_j}{3}< d_j+2v_j\leq 5J_{\rm an}.
\end{equation}

\medskip
\noindent
\textbf{Bound for $A_{j+1}$.} Using \Cref{eq:uniform-proof-A-bound} and the bound for $v_{j+1}$ in \Cref{eq:vj+1-bound-appendix}, we immediately have a bound $a_{j+1}\leq ({\pi v_0}/{\nu})(2/3)^{j+1}$.  We want to ensure $3a_{j+1}\leq \delta\kappa_{j+1}$. Because of the assumption \Cref{eq:nu-requirement-app}, we have $\nu \geq {27\pi v_0}/{\kappa_0}$, which leads to $a_{j+1}\leq ({\kappa_0}/{27})(2/3)^{j+1}$.

Therefore to show $3a_{j+1}\leq \delta\kappa_{j+1}$, using the definition of $\kappa_j$ in \Cref{eq:kappa-j-defn}, we only need to show
\begin{equation}
    \frac{1}{9}\left(\frac{2}{3}\right)^{j+1}\leq \frac{1}{1+\log(j+2)}-\frac{1}{1+\log(j+3)}.    
\end{equation}
This elementary inequality holds for all $j$.

\Cref{eq:inductive-proof-goals-analytic-extension} directly implies the estimates we want to prove in the statement of the theorem. For $D_j$, we have $\|D_j\|_{\kappa_j}< d_j+2v_j\leq 5J_{\rm an}$. For $V_j$, by noting $v_0\leq 2J_{\rm an}$, we have $\|V_j\|_{\kappa_j}\leq 2J_{\rm an}(2/3)^j$. For $A_j$, using \Cref{eq:uniform-proof-A-bound} and the bound for $V_j$, we have
\begin{equation}
    \|A_j\|_{\kappa_j}\leq \frac{2\pi J_{\rm an}}{\nu}\left(\frac{2}{3}\right)^j.    
\end{equation}

\medskip
\noindent \textbf{Analyticity.} Since \(\Lambda\) is finite, the space \(\mathfrak I_\Lambda\) of labelled interactions is finite dimensional. Charge projections and \(\mathcal S_N\) are complex-linear on this space. The interaction commutator in \Cref{eq:interaction-operations} is bilinear, so \(\operatorname{ad}_A\) is an endomorphism of \(\mathfrak I_\Lambda\) whose matrix entries depend linearly on \(A\). Its exponential makes \((A,W)\mapsto e^{\operatorname{ad}_A}W\) entire as an interaction-valued map. Thus each fixed ADHH step is analytic in \(\mathfrak I_\Lambda\), and so are \(Z\mapsto D_j(Z),V_j(Z),A_j(Z)\) for every fixed \(0\le j\le n_*\). Cauchy's integral formula consequently holds for these labelled families in each local norm \(\norm{\cdot}_{\kappa_j}\). Applying the triangle inequality to that integral uses the preceding uniform local-norm bounds directly, without a volume-dependent comparison to the norm of the summed operator.
\end{proof}

\section{Proof of local approximation (\Cref{sec:local-approx})}
\label{app:patch-proofs}
\begin{lemma}[\Cref{lem:scalar-correction}]
For every complex interaction \(Z\in\mathfrak B_{\rm an}\), there exists a constant \(C_Q\) such that
\begin{equation}
 \abs{E_\rho^{\Lambda}(Z)-E_{\rho,0}^{\Lambda}(Z)}, \abs{E_\rho^{B}(Z)-E_{\rho,0}^{B}(Z)}
 \le {C_Q \Jan^2}/{\nu}.
 \label{eq:scalar-correction-full-app}
\end{equation}
The bound holds for every patch $B$ and full domain $\Lambda$.
\end{lemma}

\begin{proof}
Let \(\Omega\) denote either the full domain \(\Lambda\) or a patch \(B\), and let \(Z_\Omega\) be the interaction entering the corresponding ADHH recursion. Thus \(Z_\Lambda=Z\) and \(Z_B=\mathcal E_BZ\), with \(\mathcal E_B\) defined in \Cref{eq:interaction-patch-restriction}. For each label \(X\), the map \(\mathcal E_B\) takes the normalized partial trace over \(X\setminus B\) and inserts the identity there. The resulting component acts only in \(X\cap B\), but keeps the original label \(X\) and the weight \(e^{\kappa_0|X|}\) in the local norm. Since this map decreases the norm of each component, the same local-norm estimate applies to both choices of \(\Omega\). We suppress the labels \(\rho\) and \(\Omega\) below. By \Cref{eq:interaction-patch-restriction}, $\norm{Z_\Omega}_{\kappa_0} \le \norm{Z}_{\kappa_0} \le \Jan$.
 
As in the proof of~\Cref{thm:ADHH-normal-form}, we define $d_j:=\norm{D_j}_{\kappa_j}$, $v_j:=\norm{V_j}_{\kappa_j}$, and $a_j:=\norm{A_j}_{\kappa_j}$. The induction in the proof of \Cref{thm:uniform-normal-form} gives, for \(0\le j\le n_*\),
\begin{equation}
 v_j\le v_0(2/3)^j,
 \quad
 d_j+2v_j\le 5\Jan,
 \quad
 a_j\le\frac{\pi}{\nu}v_j,
 \label{eq:scalar-uniform-bounds}
\end{equation}
and \(v_0\le2\norm{Z_\Omega}_{\kappa_0}\le2\Jan\).

Since \(\ell_\rho\) vanishes on the charge-changing component, i.e., $\ell_\rho((I-\mathcal{P}_0)(Z)) = 0$ for $\mathcal{P}_0$ defined in \Cref{eq:charge-projection}, $E_{\rho,0}^{\Omega}(Z)=\ell_\rho(Z_\Omega)=\ell_\rho(D_0)$ and $E_\rho^\Omega(Z)=\ell_\rho(D_{{s}})$. Consequently, using \Cref{eq:ell-local-bound-general-main},
\begin{align}
    \abs{E_\rho^\Omega(Z)-E_{\rho,0}^\Omega(Z)} & =\abs{\sum_{j=0}^{{s}-1}\ell_\rho(D_{j+1}-D_j)}\\
    & \le \sum_{j=0}^{{s}-1}\abs{\ell_\rho(D_{j+1}-D_j)} \\
    & \le \sum_{j=0}^{{s}-1} \norm{D_{j+1}-D_j}_{\kappa_{j+1}}.
    \label{eq:scalar-telescoping}
\end{align}

By \Cref{eq:uniform-proof-D-V-increment},
\begin{equation}
    \norm{D_{j+1}-D_j}_{\kappa_{j+1}}\le m_{\rm A}(j) a_j(d_j+2v_j),
    \label{eq:scalar-D-increment}
\end{equation}
where $m_{\rm A}(j):={18}/{(\delta\kappa_j\,\kappa_{j+1})}$. As before, we introduce $\kappa_j$ \Cref{eq:kappa-j-defn} to monitor the loss of locality during the ADHH recursion.  Set \(L_j:=1+\log(j+1)\), so that \(\kappa_j=\kappa_0/L_j\).  Then
\begin{align}
    m_{\rm A}(j)  = \frac{18 L_jL_{j+1}^2}{\kappa_0^2 \log((j+2)/(j+1))} \le \frac{18 (j+2)\bigl[1+\log(j+2)\bigr]^3}{\kappa_0^2},
\end{align}
where we used \(\log((j+2)/(j+1))\ge1/(j+2)\) and \(L_j\le L_{j+1}\).  Combining \Cref{eq:scalar-uniform-bounds,eq:scalar-D-increment} therefore gives
\begin{equation}
    \norm{D_{j+1}-D_j}_{\kappa_{j+1}} \le
    \frac{5\pi\Jan}{\nu}\,m_{\rm A}(j)v_j \le
    \frac{10\pi\Jan^2}{\nu}\,
    m_{\rm A}(j)\left(\frac23\right)^j.
    \label{eq:scalar-increment-final}
\end{equation}
The sum of the $j$-dependent part above over $j$ is upper bounded by a constant:
\begin{equation}
    \sum_{j=0}^{s-1}m_{\rm A}(j)\left(\frac23\right)^j
    \le \frac{18}{\kappa_0^2}
    \sum_{j=0}^{\infty} (j+2)\bigl[1+\log(j+2)\bigr]^3 \left(\frac23\right)^j
    =:C_{\kappa_0}<\infty.
    \label{eq:scalar-summability}
\end{equation}
In particular, this bound is independent of \(n_*\).  Substitution into
\Cref{eq:scalar-telescoping} yields
\begin{equation}
    \abs{E_\rho^\Omega(Z)-E_{\rho,0}^\Omega(Z)}\le \frac{10\pi C_{\kappa_0} \Jan^2}{\nu}.
\end{equation}
Taking \(C_Q:=10\pi C_{\kappa_0}\) and then \(\Omega=\Lambda\) or \(\Omega=B\) proves \Cref{eq:scalar-correction-full-app}.  All estimates used above hold for complex interactions and are uniform in the domain $\mathfrak{B}_{\rm an}$, so the same constant applies to the full system $\Lambda$ and every patch $B$.
\end{proof}

\begin{lemma}[\Cref{lem:support-calculus}, Closure properties of connected $p$-words]
Let $W_p\in \mathcal C_{p,H},W_q\in \mathcal C_{q,H}$. Then:
\begin{enumerate}[label=(\roman*)]
\item The connected $p$-word space is closed under $\mathcal P_{N_\rho,\omega}$ and $\mathcal S_{N_\rho}$:
\begin{equation}
    \mathcal P_{N_\rho,\omega}(W_p), \mathcal S_{N_\rho}(W_p) \in \mathcal C_{p,H}.    
\end{equation}
\item The commutator of connected $p$-word and $q$-word generate connected $p+q$-word:
\begin{equation}
    [W_p,W_q] \in \mathcal C_{p+q,H}.    
\end{equation}

\item If $A(z)=\sum_{p\ge 1} A_pz^p,B(z)={B_0+}\sum_{p\ge 1} B_pz^p$, satisfies $A_p,B_p\in \mathcal C_{p,H}$ for every $p{\geq 1}$ {and $B_0\in \mathrm{span}(N_\rho)$}, we have:
\begin{equation}
    e^{A(z)}B(z)e^{-A(z)} ={B_0+} \sum_{p\ge 1} C_pz^p,
\end{equation}
where $C_p\in \mathcal C_{p,H}\;(\forall p \ge 1)$.
\end{enumerate}
\end{lemma}

\begin{proof} We will prove the three statements in the lemma one by one.
\begin{enumerate}[label=(\roman*)]
\item By linearity we only need to consider the case where $W_p$ is the evaluation connected $p$-word $(\mathbf{a},\mathbf{Q})$, where $\mathbf{Q}=(Q_1,\cdots,Q_p)$, as in \Cref{def:p-word}. Because $\rho$ is fixed in this lemma we will omit it in the notation. Recall the definition of $\mathcal P_{N,\omega}$:
\begin{equation}
    \mathcal P_{N,\omega}(W_p)=\frac{1}{2\pi}\int_0^{2\pi}e^{-i\omega\theta}e^{i\theta N}W_pe^{-i\theta N}\,d\theta.   
\end{equation}
Note that $e^{i\theta N}W_p e^{-i\theta N}=e^{i\theta N}Q_1 e^{-i\theta N}\cdots e^{i\theta N}Q_p e^{-i\theta N}$. Because $N$ is a sum of single-qubit operators, each $e^{i\theta N}Q_i e^{-i\theta N}=\sum_{j}\alpha_{ij}Q_{i,j}$, where $Q_{i,j}$ is a Pauli operator with exactly the same support as $Q_i$. Therefore
\begin{equation}
    e^{i\theta N}W_p e^{-i\theta N}=\sum_{j_1,\cdots,j_p}\alpha_{1,j_1}\cdots\alpha_{p,j_p} Q_{1,j_1}\cdots Q_{p,j_p}\in \mathcal{C}_{p,H},
\end{equation}
since each $Q_{1,j_1}\cdots Q_{p,j_p}$ is the evaluation of a $p$-word $(\mathbf{a},(Q_{1,j_1},\cdots, Q_{p,j_p}))$. Consequently $\mathcal P_{N,\omega}(W_p)\in \mathcal{C}_{p,H}$. The same argument also shows $\mathcal S_{N}(W_p)\in \mathcal{C}_{p,H}$ by its definition in \Cref{eq:homological-inverse}.

\item By the bilinearity of the commutator, we only need to consider the case where $W_p$ and $W_q$ are evaluations of the connected $p$-word $(\mathbf{a},(Q_1,\cdots,Q_p))$ and connected $q$-word $(\mathbf{a}',(Q'_1,\cdots,Q'_q))$ respectively. If the overlap graph of $Q_1,\cdots,Q_p,Q'_1,\cdots,Q'_q$ (as defined in \Cref{def:p-word}) is connected then
\begin{equation}
    [W_p,W_q]=[Q_1Q_2\cdots Q_p,Q'_1Q'_2\cdots Q'_q]=Q_1Q_2\cdots Q_pQ'_1Q'_2\cdots Q'_q-Q'_1Q'_2\cdots Q'_qQ_1Q_2\cdots Q_p    
\end{equation}
is in $\mathcal{C}_{p+q,H}$ because both terms above are in $\mathcal{C}_{p+q,H}$. On the other hand, if the overlap graph is disconnected, then because $(\mathbf{a},(Q_1,\cdots,Q_p))$ and $(\mathbf{a}',(Q'_1,\cdots,Q'_q))$ are connected words, we must have $\mathrm{supp}(W_p)\cap \mathrm{supp}(W_q)=\varnothing$, and hence $[W_p,W_q]=0\in \mathcal{C}_{p+q,H}$.

\item Finally, using BCH-formula:
\begin{equation}
    e^{A(z)}B(z)e^{-A(z)}=\sum_{m=0}^{\infty}\frac1{m!}\ad_{A(z)}^m(B(z)).    
\end{equation}
Then, the degree $p$ coefficient of $e^{A(z)}B(z)e^{-A(z)}$, $C_p$ are represented as
\begin{equation}
\label{eq:Cp_expression}
    C_p=\sum_{m=0}^p\frac{1}{m!}\sum_{\substack{q\ge 0,p_1,\dots,p_m\ge 1, \\ q+p_1+\cdots+p_m = p}}\ad_{A_{p_m}}\cdots\ad_{A_{p_1}}(B_{q}).    
\end{equation}
Using (ii), {for $q\geq 1$,} we have $\ad_{A_{p_m}}\cdots\ad_{A_{p_1}}(B_{q}) \in \mathcal C_{p_m+\cdots+p_1+q,H}=\mathcal C_{p,H}$. 
{For $q=0$, because $B_0\in\mathrm{span}(N)$, we can write $B_0=\beta N$ 
\begin{equation}
    \ad_{A_{p_1}}(B_{0})=[A_{p_1},\beta N]=-\beta\sum_\omega\omega\,\mathcal P_{N,\omega}(A_{p_1})\in\mathcal C_{p_1,H}.    
\end{equation}
Therefore $\ad_{A_{p_m}}\cdots\ad_{A_{p_1}}(B_{0}) \in \mathcal C_{p_m+\cdots+p_1,H}=\mathcal C_{p,H}$}. Therefore, {each summand in \Cref{eq:Cp_expression} is in $\mathcal C_{p,H}$, and as a result} $C_p \in \mathcal C_{p,H}$.
\end{enumerate}
\end{proof}

\begin{lemma}[\Cref{lem:connected-word-expansion}, Connected-word expansion]
We consider ADHH recursion \Cref{eq:ADHH-recursion} for \(H=H(\mu)=\sum_{a\in\cA}\mu_aP_a\) and static field $\nu N_\rho$. For every ADHH step $0\le j\le n_*$ and every $p\ge1$, the degree-$p$ coefficient of $D_{\rho,j}^{\Lambda}(zH)$ is an element of $\mathcal C_{p,H}$:
\begin{equation}
    D_{\rho,j}^{\Lambda}(zH) = \sum_{p\ge 1}D_{p,j}z^p,
\end{equation}
where $D_{p,j} \in \mathcal C_{p,H}$.
More precisely, each \(D_{p,j}\) admits a finite expansion into terms
$c_{\mathbf{a},\mathbf{Q}}\mu_{a_1}\cdots\mu_{a_p}W_{\mathbf a,\mathbf Q},$ where \((\mathbf a,\mathbf Q)\) is a connected \(p\)-word, and for fixed \(j\), \(\nu\), and $N$, the scalar $c_{\mathbf{a},\mathbf{Q}}$ is independent of \(\mu\) and depends only on $\mathbf{a},\mathbf{Q}$.
\end{lemma}

\begin{proof}
Suppose that \(H_0=zH\) for a complex \(z\) and \(H=H(\mu)\). We employ the ADHH recursion defined in \Cref{eq:ADHH-recursion}, with \(N=N_\rho\). Since \(H\in\mathcal C_{1,H}\), linearity and \Cref{lem:support-calculus}(i) show that the degree-\(p\) coefficients of \(H_0,D_0,V_0,A_0\) are elements of \(\mathcal C_{p,H}\) for every \(p\ge1\).  Because when $z=0$ we have $H_0=D_0=V_0=A_0=0$, the constant terms in $H_0,D_0,V_0,A_0$ are therefore all zero.

Suppose that the degree-\(p\) coefficients of \(H_j,D_j,V_j,A_j\) are elements of \(\mathcal C_{p,H}\) for every \(p\ge1\), and that their constant terms vanish. By \Cref{lem:support-calculus}(iii), with \(B_0=\nu N\), the positive-degree coefficients of \(e^{A_j}(\nu N+H_j)e^{-A_j}\) belong to the appropriate connected-word spaces. Since subtracting \(\nu N\) changes only the degree-zero coefficient, the degree-\(p\) coefficient of \(H_{j+1}\) is in \(\mathcal C_{p,H}\), and \(H_{j+1}(0)=0\).

\(D_{j+1}\) and \(V_{j+1}\) are calculated from \(H_{j+1}\) using \(\langle\,\cdot\,\rangle_{N_\rho}=\mathcal P_{N_\rho,0}\), and \(A_{j+1}\) is calculated by applying \(\mathcal S_{N_\rho}\) to \(V_{j+1}\). \Cref{lem:support-calculus}(i) shows that \(\mathcal C_{p,H}\) is closed under these maps. Hence, the degree-\(p\) coefficients of \(D_{j+1},V_{j+1},A_{j+1}\) are in \(\mathcal C_{p,H}\), and their constant terms also vanish. Therefore, by induction, for every \(j\), the degree-\(p\) coefficient of \(D_j\) is an element of \(\mathcal C_{p,H}\).

We now track the Hamiltonian coefficients $\{\mu_a\}$ in  \(H_j,D_j,V_j,A_j\) using the equivalent expression \Cref{eq:recursion-no-N} for ADHH recursion.  \Cref{eq:recursion-no-N} allows us to only access $N$ through $\mathcal P_{N,\omega}$ and $\mathcal S_N$.  Here we think of \(H_j,D_j,V_j,A_j\) as formal power series in \(z\) whose coefficients are polynomials in the independent variables \(\{\mu_a\}\). Initially, each summand \(z\mu_aP_a\) pairs one coupling factor \(\mu_a\) with one interaction label \(a\). $\mathcal P_{N,\omega}$ and $\mathcal S_N$ expand \(W_{\mathbf a,\mathbf Q}\) into linear combinations of \(W_{\mathbf a,\mathbf Q'}\), with \(Q_i'\in\mathcal Q_{a_i}\). Thus every resulting term retains the original interaction-label tuple \(\mathbf a=(a_1,\ldots,a_p)\) and coupling monomial \(\mu_{a_1}\cdots\mu_{a_p}\). Commutators between two such terms multiply their coupling monomials and concatenate their label and Pauli tuples. Consequently, the preceding induction argument ensures that the resulting terms can be represented by connected words. Thus the same induction shows that every degree-\(p\) coefficient is a homogeneous polynomial of degree \(p\) in \(\mu\). In particular, \(D_{p,j}\) has a finite expansion into terms \(c_{\mathbf a,\mathbf Q}\mu_{a_1}\cdots\mu_{a_p} W_{\mathbf a,\mathbf Q}\), with \((\mathbf a,\mathbf Q)\) a connected \(p\)-word. For fixed \(N,\nu,j\), the coefficients \(c_{\mathbf a,\mathbf Q}\) are obtained by finite sums and products of scalar weights from the recursion, all independent of \(\mu\).
\end{proof}

\begin{lemma}[\Cref{lem:coefficient-matching}]
If \(R\ge(K+1)r_0\), then $c_p^{\Lambda}=c_p^B$, where $0\le p\le K$.
\end{lemma}

\begin{proof}
The constant coefficients satisfy \(c_0^\Lambda=c_0^B=0\), because the recursion vanishes at zero input. Fix \(1\le p\le K\). By \Cref{lem:connected-word-expansion}, a nonzero contribution to \(c_p^\Lambda\) has the form \(c\,\mu_{a_1}\cdots\mu_{a_p}\ell_\rho(W_{\mathbf a,\mathbf Q})\), with \((\mathbf a,\mathbf Q)\) a connected \(p\)-word. Since \(\ell_\rho\) annihilates operators whose support does not contain \(k\), any $W_{\mathbf{a,Q}}$ having a nonzero contribution must satisfy \(k\in\supp(W_{\mathbf a,\mathbf Q})\subseteq\bigcup_iS_i\), where \(S_i:=\supp(Q_i)=\supp(P_{a_i})\).

Choose a set $S_{i_0}$ containing $k$ and set $x_0=k$. For any $S_{i_m}$ and any $x\in S_{i_m}$,  connectedness gives a simple path $S_{i_0},S_{i_1},\ldots,S_{i_m}$, for $m\le p-1$, where $S_{i_{t-1}}$ overlaps with $S_{i_t}$. Pick \(x_t\in S_{i_{t-1}}\cap S_{i_t}\) for \(t=1,\ldots,m\){, and let $x_{m+1}=x$}.  Because the diameter of each $S_{i_t}$ is at most \(r_0\), we have $\dist(x_{t-1},x_t)\leq r_0$, and consequently $\dist(k,x) \le(m+1)r_0 \le pr_0$. Since \(x\) was arbitrary, the whole union lies in \(B_{pr_0}(k)\). Hence, 
\begin{equation}
    \bigcup_{i=1}^pS_i\subseteq B_{pr_0}(k)\subseteq B_R(k)=B.
    \label{eq:cluster-radius}
\end{equation}
We first note that even though $c_p^B$ is the $p$-th Taylor coefficient of $E_\rho^B(zH_B)$, it is in fact equal to that of  $E_\rho^\Lambda(zH_B)$, because $E_\rho^B(zH_B)=E_\rho^\Lambda(zH_B)$ by \Cref{lem:exact-embedding}. Hence going from $c_p^\Lambda$ to $c_p^B$ is equivalent to replacing $H$ with $H_B$ while keeping $N$ unchanged, which is in turn equivalent to setting all $\mu_a=0$ when $\mathrm{supp}(P_a)\not\subseteq B$. Therefore, $c_p^\Lambda\neq c_p^B$ only if a nonzero contribution in $c_p^\Lambda$ \(c\,\mu_{a_1}\cdots\mu_{a_p}\ell_\rho(W_{\mathbf a,\mathbf Q})\) involves a label $a_i$ for which $S_i=\mathrm{supp}(P_{a_i})\not\subseteq B$, which contradicts \Cref{eq:cluster-radius}.
\end{proof}

\begin{lemma}[\Cref{lem:analytic-gap-bound}, Uniform analytic transition-energy bound]
For every \(\mu\in\cU\), the functions in \Cref{eq:g-functions} are analytic for \(|z|<2\) and satisfy
\begin{equation}
    \sup_{|z|\le2}
    \max\left\{\abs{g_\rho^{\Lambda}(z)},\abs{g_\rho^B(z)}\right\}
    \le C_g\Jan.
    \label{eq:analytic-gap-bound-app}
\end{equation}
Consequently,
\begin{equation}
    |c_p^{\Lambda}|\le C_g\Jan2^{-p}, \quad
    |c_p^B|\le C_g\Jan2^{-p}.
    \label{eq:Taylor-coeff-bound-app}
\end{equation}
\end{lemma}

\begin{proof}
For \(\mu\in\cU\), \Cref{eq:real-basin-local} gives $\norm H_{\kappa_0}\le2J$. The map \(\mathcal E_B\) keeps the Pauli terms supported in \(B\) and annihilates the others, so \(H_B=\mathcal E_BH\). Its componentwise norm bound in \Cref{eq:interaction-patch-restriction} gives \(\norm{H_B}_{\kappa_0}\le\norm H_{\kappa_0}\le2J\). The full-system recursion has input \(zH(\mu)\), and the patch recursion has input \(zH_B(\mu)=z\mathcal E_BH(\mu)\). Both inputs lie in \(\mathfrak B_{\rm an}\) for every \(|z|\le2\), so the uniform recursion bounds apply to both on this disc. Hence, for \(|z|\le2\),
\begin{equation}
    \norm{zH}_{\kappa_0},\ \norm{zH_B}_{\kappa_0}
    \le4J=\Jan.
\end{equation}
The ADHH recursion \Cref{eq:ADHH-recursion} is analytic and \Cref{eq:uniform-D-bound} gives
\begin{equation}
    \norm{\widehat D_\rho^{\Lambda}(zH)}_{\kappa_{{s}}},\
    \norm{\widehat D_{\rho,R}(zH_B)}_{\kappa_{{s}}}
    \le C\Jan.
    \label{eq:D-analytic-bound-proof}
\end{equation}
Using \Cref{eq:ell-local-bound-general-main} proves \Cref{eq:analytic-gap-bound-app}.  Cauchy's coefficient estimate on every circle \(|z|=r<2\) gives \(|c_p|\le C_g\Jan r^{-p}\); letting \(r=2\) yields \Cref{eq:Taylor-coeff-bound-app}.
\end{proof}

\begin{theorem}[\Cref{thm:patch-locality}, Local approximation of the effective half transition energy]
Let \(R\ge(K+1)r_0\).  Uniformly for every row \(\rho=(k,\ell)\) and every \(\mu\in\cU\),
\begin{equation}
    \abs{{E_\rho^{\Lambda}(\mu)}-F_{\rho,R}(\mu)}
    \le C_{\rm P}\Jan2^{-K}.
    \label{eq:patch-locality-app}
\end{equation}
One may take \(C_{\rm P}=2C_g\).
\end{theorem}

\begin{proof}
In \Cref{eq:patch-locality-app}, $E_\rho^\Lambda(\mu)$ and $F_{\rho,R}(\mu)$ are the two functions in \Cref{eq:g-functions} for $z=1$. Their Taylor coefficients agree through degree \(K\) by \Cref{lem:coefficient-matching}.  The bounds in \Cref{lem:analytic-gap-bound} imply absolute convergence at \(z=1\), and therefore
\begin{align}
    \abs{g_\rho^{\Lambda}(1)-g_\rho^B(1)} & \le\sum_{p=K+1}^{\infty}\bigl(|c_p^{\Lambda}|+|c_p^B|\bigr) \\
    & \le2C_g\Jan\sum_{p=K+1}^{\infty}2^{-p}
    =2C_g\Jan2^{-K}.
    \label{eq:patch-tail-proof}
\end{align}
This is the claimed result with \(C_{\rm P}=2C_g\).
\end{proof}

\section{Proof of coefficient recovery (\Cref{sec:optimization})}
\label{app:optimization-proofs}

For \(P_a=\bigotimes_{j\in\Lambda}P_{a,j}\), let \(\chi(a)\in\{0,1\}^{\Lambda}\) is defined as \(\chi_j(a)=\ind[P_{a,j}\neq I]\). We write  \(P_a\sim\beta^\rho\) if \(P_{a,j}=\beta_j^\rho\) for every \(j\in\supp(P_a)\). We recall that $q$ is the maximal support size of Hamiltonian terms $P_a$, $D_{\rm ov}$ is the maximal number of terms a Hamiltonian term can overlap with (including itself), and $\zeta$ is the maximal number of terms that cover any given qubit $k$.

\begin{lemma}[\Cref{lem:local-Walsh}]
The function $F_{\rho,0}$ is linear:
\begin{equation}
    F_{\rho,0}(\mu)=\sum_{a\in\cA}X_{\rho,a}\mu_a,
\end{equation}
where $X_{\rho,a}=-\ind[k\in\supp(P_a)]\ind[P_a\sim\beta^\rho](-1)^{s^\rho\cdot\chi(a)}$ is the element of $X$ on the row indexed by $\rho$ and column indexed by $a$. If $R\ge r_0$, each row of $X\in\mathbb R^{nm\times M}$ has at most $\zeta$ nonzero entries, and each column has at most $qm$ nonzero entries. $\beta^\rho$ and $s^\rho$ specify the local Pauli frame and are defined in \Cref{eq:random-pattern}.
\end{lemma}

\begin{proof}
For $\ket{\Phi_b^\rho}$, the charge-preserving average $\langle\,\cdot\,\rangle_{N_\rho}$ leaves expectation unchanged:
\begin{equation}
    \bra{\Phi_b^\rho}\langle P_a\rangle_{N_\rho}\ket{\Phi_b^\rho}=\bra{\Phi_b^\rho}P_a\ket{\Phi_b^\rho}.
\end{equation}
Hence, $X_{\rho,a}=\frac{1}{2}(\bra{\Phi_1^\rho}P_a\ket{\Phi_1^\rho}-\bra{\Phi_0^\rho}P_a\ket{\Phi_0^\rho})$. If \(P_a\nsim\beta^\rho\), $\bra{\Phi_b^\rho}P_a\ket{\Phi_b^\rho} = 0$ for both $b=0,1$. If \(P_a\sim\beta^\rho\), then $\bra{\Phi_0^\rho}P_a\ket{\Phi_0^\rho}=(-1)^{s^\rho\cdot\chi(a)}$. The state \(\Phi_1^\rho\) differs only at \(k\), hence
\begin{equation}
    \bra{\Phi_1^\rho}P_a\ket{\Phi_1^\rho}
    =(-1)^{\ind[k\in\supp(P_a)]}(-1)^{s^\rho\cdot\chi(a)}.
    \label{eq:Phi1-Pa-expectation}
\end{equation}
Finally, we have $ X_{\rho,a}=-\ind[k\in\supp(P_a)]\ind[P_a\sim\beta^\rho](-1)^{s^\rho\cdot\chi(a)}$. If \(k\in\supp(P_a)\), \(R\ge r_0\) imply \(\supp(P_a)\subseteq B_R(k)\). At fixed \(k\), at most \(\zeta\) (defined in \Cref{eq:zeta-Dov}) different Pauli contains the defect qubit $k$, proving the row sparsity. For a fixed \(a\), nonzero entries occur only at the \(q\) possible defect sites and the \(m\) repetitions, proving the column sparsity.
\end{proof}

\begin{theorem}[\Cref{thm:Walsh-conditioning}]
There is a constant \(C_m\), such that if
\begin{equation}
    m\ge C_m{\log(2n/\delta)},
    \label{eq:m-condition-app}
\end{equation}
with probability at least \(1-\delta\), the following holds:
\begin{align}
    &\norm{\Gamma-\Sigma}_{\infty\to\infty}
    \le 3^{-q} /4 , \label{eq:Gamma-close-app}\\
    &\frac{3}{4}\cdot3^{-q}I \preceq\Gamma
    \preceq\left(q+3^{-q}/4\right)I, \label{eq:Gamma-spectral-app}\\
    &\norm{\Gamma^{-1}}_{\infty\to\infty}\le C_\Gamma:=(4/3)\cdot3^{q}. \label{eq:Gamma-inverse-app}
\end{align}
\end{theorem}

\begin{proof}
We first compute the expectation of \(\Gamma\). For \(a=b\), the product \(X_{(k,\ell),a}^2\) is nonzero only when \(k\in\supp(P_a)\) and the sampled axes agree with \(P_a\) on all \(p_a\) sites in its support. Therefore, $\bbE[\Gamma_{aa}]=p_a3^{-p_a}$.

Now let \(a\neq b\). If \(\supp(P_a)\neq\supp(P_b)\), choose a site in the symmetric difference of the two supports. Averaging over the independent sign at this site gives $\bbE\bigl[X_{(k,\ell),a}X_{(k,\ell),b}\bigr]=0$. If the supports are equal, then \(P_a\neq P_b\) means that their Pauli labels differ at some site. In this case, no sampled axis can agree with both strings, so $X_{(k,\ell),a}X_{(k,\ell),b}=0$. Thus, $\bbE[\Gamma]=\Sigma$, and $\Sigma=\operatorname{diag}\bigl(p_a3^{-p_a}\bigr)_{a\in\cA}$.

Notice also that \(X_{(k,\ell),a}X_{(k,\ell),b}\) can be nonzero only when $k\in\supp(P_a)\cap\supp(P_b)$. Hence \(\Gamma_{ab}=0\) whenever the two supports are disjoint, and each row of \(\Gamma\) has at most \(D_{\rm ov}\) entries that can
be nonzero.

Fix \(a,b\) whose supports overlap. At most \(q\) defect sites contribute to \(\Gamma_{ab}\). Hence \(\Gamma_{ab}-\Sigma_{ab}\) is a sum of at most \(qm\) independent centered bounded random variables, divided by \(m\). Hoeffding's inequality therefore gives
\begin{equation}
    \bbP\left( |\Gamma_{ab}-\Sigma_{ab}|>t \right) \le 2e^{-c_qmt^2},
    \label{eq:Hoeffding-Gram}
\end{equation}
where \(c_q>0\) depends only on \(q\). Set $t={3^{-q}}/{(4D_{\rm ov})}$. There are at most \(MD_{\rm ov}\) pairs whose supports overlap. Using the bounds on \(M\) and \(D_{\rm ov}\), we have $\log(2MD_{\rm ov}/{\delta}) = {\bigO(\log(2n/\delta))}$.
    
Thus, by the union bound and the choice of $m$ in \Cref{eq:m-condition-app}, with probability at least \(1-\delta\),
\begin{equation}
    |\Gamma_{ab}-\Sigma_{ab}| \le \frac{3^{-q}}{4D_{\rm ov}}
\end{equation}
for all such pairs. Consequently, for every \(a\),
\begin{equation}
    \sum_b|\Gamma_{ab}-\Sigma_{ab}| \le 3^{-q}/4.
    \label{eq:Gram-row-error}
\end{equation}
By the definition of the induced infinity norm, $\norm{\Gamma-\Sigma}_{\infty\to\infty}=\max_a\sum_b|\Gamma_{ab}-\Sigma_{ab}|$. Therefore, \Cref{eq:Gram-row-error} proves $\norm{\Gamma-\Sigma}_{\infty\to\infty} \le 3^{-q}/4$, which is \Cref{eq:Gamma-close-app}.

Let \(\lambda\) be an eigenvalue of \(\Gamma\). By Gershgorin's circle theorem, there is an \(a\) such that
\begin{equation}
    \Gamma_{aa}-\sum_{b\neq a}|\Gamma_{ab}| \le \lambda \le \Gamma_{aa}+\sum_{b\neq a}|\Gamma_{ab}|.
\end{equation}
Since \(\Sigma\) is diagonal, \Cref{eq:Gram-row-error} implies
\begin{equation}
    \lambda\ge\Sigma_{aa}-\sum_b|\Gamma_{ab}-\Sigma_{ab}|\ge p_a3^{-p_a}-3^{-q}/4,
\end{equation}
and similarly, $\lambda \le p_a3^{-p_a}+3^{-q}/4$. Finally, \(1\le p_a\le q\) gives $3^{-q}\le p_a3^{-p_a}\le q$. Therefore, every eigenvalue of \(\Gamma\) lies in $[(3/4)\cdot3^{-q}, q+3^{-q}/4]$, which proves \Cref{eq:Gamma-spectral-app}.

It remains to bound the inverse. Since $\Sigma$ is diagonal, \Cref{eq:Gram-row-error} gives, for every $a$,
\begin{align}
    |\Gamma_{aa}|-\sum_{b\neq a}|\Gamma_{ab}| &\ge \Sigma_{aa} -|\Gamma_{aa}-\Sigma_{aa}| -\sum_{b\neq a}|\Gamma_{ab}-\Sigma_{ab}| \\
    &= \Sigma_{aa}-\sum_b|\Gamma_{ab}-\Sigma_{ab}| \\
    &\ge 3^{-q}-\frac{3^{-q}}{4}= (3/4)\cdot3^{-q}.
\end{align}
Thus, $\Gamma$ is strictly row-diagonally dominant. The standard inverse bound for a strictly row-diagonally-dominant matrix gives
\begin{align}
    \norm{\Gamma^{-1}}_{\infty\to\infty} &\le \left[\min_a\left( |\Gamma_{aa}|-\sum_{b\neq a}|\Gamma_{ab}|\right)\right]^{-1} \le (4/3)\cdot3^q =C_\Gamma,
\end{align}
which proves \Cref{eq:Gamma-inverse-app}.
\end{proof}

\begin{lemma}[\Cref{lem:row-derivatives}, Uniform first- and second-derivative bounds]
Uniformly for \(\mu\in\cU\) and all rows \(\rho\),
\begin{align}
    \sum_a\abs{\partial_a\mathcal R_\rho(\mu)} & \le {C_1\Jan}/{\nu},
    \label{eq:first-derivative-bound-app}\\
    \sup_a\sum_b\abs{\partial_a\partial_bF_\rho(\mu)} & \le{C_2}/{\nu}. \label{eq:second-derivative-bound-app}
\end{align}
All derivatives with respect to coefficients outside \(\cA_\rho:=\{a:\supp(P_a)\subseteq B_R(k)\}\) vanish.
\end{lemma}

\begin{proof}
By \Cref{lem:scalar-correction}, we have $|\mathcal R_\rho(Z)|\le C_Q{\Jan^2}/{\nu}$ throughout the complex local-norm ball.  Fix \(\mu\in\cU\) and a direction \(K_\xi\) as 
\begin{equation}
    K_\xi:=\sum_{a\in\cA_\rho}\xi_aP_a,
    \label{eq:coefficient-direction}
\end{equation}
where $\norm\xi_\infty\le1$. By \Cref{eq:coeff-to-local}, \(\norm{K_\xi}_{\kappa_0}\le J\). 
For \(|z|\le 1\),
\begin{equation}
    \norm{H_B(\mu)+zK_\xi}_{\kappa_0}
    \le2J+J<\Jan.
    \label{eq:direction-inside-ball}
\end{equation}
Denoting the directional derivative of $\mathcal R_\rho(\mu)$ in the $K_\xi$ direction by $D\mathcal R_\rho(\mu)[K_\xi]$, Cauchy's estimate in \Cref{eq:Cauchy-polydisc} gives
\begin{equation}
    |D\mathcal R_\rho(\mu)[K_\xi]|
    \le\frac{C_Q\Jan^2/\nu}{r_1}
    \le \frac{C_1\Jan}{\nu}.
    \label{eq:first-directional-Cauchy}
\end{equation}
Taking the supremum over \(|\xi_a|\le1\) and using $\sum_a|v_a|=\sup_{|\xi_a|\le1}\abs{\sum_a\xi_av_a}$ proves \Cref{eq:first-derivative-bound-app}.

For the second derivative, fix \(a\) and again take \(\norm\xi_\infty\le1\). Since 
{$\norm{P_a}_{\kappa_0}\le e^{\kappa_0(2r_0+1)^d}$,}
choose $r_a:={J}/{(4e^{\kappa_0(2r_0+1)^d})}$, 
and $r_2:=1/4$. The polydisc \(|z|\le r_a\), \(|w|\le r_2\) for
\begin{equation}
    (z,w)\mapsto\mathcal R_\rho(H_B(\mu)+zP_a+wK_\xi)
    \label{eq:two-variable-map}
\end{equation}
remains in \(\mathfrak B_{\rm an}\).  The two-variable Cauchy estimate yields
\begin{equation}
    |D^2\mathcal R_\rho(\mu)[P_a,K_\xi]|
    \le\frac{C_Q\Jan^2/\nu}{r_ar_2}
    \le\frac{C_2}{\nu}.
    \label{eq:second-directional-Cauchy}
\end{equation}
$F_{\rho,0}$ is linear, so \(D^2F_\rho=D^2\mathcal R_\rho\).  Duality in the second index proves \Cref{eq:second-derivative-bound-app}.  Dependence only on \(H_B\) makes every derivative outside \(\cA_\rho\) vanish.
\end{proof}

\begin{lemma}[\Cref{lem:Jacobian-perturbation}, Jacobian Gram perturbation]
Uniformly for \(\mu\in\cU\), the Jacobian Gram matrix {$\mathsf G(\mu)=\mathsf{J}(\mu)^{\mathsf T} \mathsf{J}(\mu)/m$} satisfies
\begin{equation}
    \norm{{\mathsf G(\mu)}-\Gamma}_{\infty\to\infty}
    \le C_JV_R\left(
    \frac{\Jan}{\nu}+\frac{\Jan^2}{\nu^2}
    \right).
    \label{eq:Jacobian-Gram-close-app}
\end{equation}
On the event that \Cref{eq:Gamma-close} holds, there exists an optimization field threshold $\nu_{\rm opt}=\mathcal{O}(JV_R)$ such that for all $\nu\geq \nu_{\rm opt}$, we have
\begin{equation}
    {\mathsf G(\mu)}\succeq\frac{I}{2\cdot3^q}
    \label{eq:Jacobian-conditioned-app}
\end{equation}
uniformly on \(\cU\).
\end{lemma}

\begin{proof}
Set $\Delta(\mu):=D\mathcal R(\mu)$, so that $\mathsf J(\mu)=X+\Delta(\mu)$. Let $d_\Delta:=C_1{\Jan}/{\nu}$. By \Cref{eq:first-derivative-bound}, every row of \(\Delta\) has \(\ell_1\) norm at most \(d_\Delta\).  A fixed column of \(\Delta\) can be nonzero only in rows whose patch contains the corresponding support, hence in at most \(mV_R\) rows.  By \Cref{lem:local-Walsh}, every row of \(X\) has at most \(\zeta\) nonzero entries and every column has at most \(qm\) nonzero entries, and all entries are bounded by $1$ in absolute value.

Expand
\begin{equation}
    \mathsf G(\mu)-\Gamma=\frac1m\bigl(X^{\mathsf T}\Delta+\Delta^{\mathsf T}X+\Delta^{\mathsf T}\Delta\bigr).
    \label{eq:Jacobian-Gram-expansion}
\end{equation}
For a fixed coefficient row \(a\), we have $\frac{1}{m}\sum_{\rho}|\Delta_{\rho a}|\leq V_R d_\Delta$, and
\begin{align}
    \frac1m\sum_b|(X^{\mathsf T}\Delta)_{ab}|
    &\le\frac1m\sum_\rho|X_{\rho a}|
    \sum_b|\Delta_{\rho b}|
    \le qd_\Delta,\label{eq:XTDelta-bound}\\
    \frac1m\sum_b|(\Delta^{\mathsf T}X)_{ab}|
    &\le\frac1m\sum_\rho|\Delta_{\rho a}|
    \sum_b|X_{\rho b}|
    \le\zeta V_Rd_\Delta,\label{eq:DeltaTX-bound}\\
    \frac1m\sum_b|(\Delta^{\mathsf T}\Delta)_{ab}|
    &\le\frac1m\sum_\rho|\Delta_{\rho a}|
    \sum_b|\Delta_{\rho b}|
    \le V_Rd_\Delta^2.
    \label{eq:DeltaTDelta-bound}
\end{align}
These inequalities imply \Cref{eq:Jacobian-Gram-close-app} by the definition of $\|\cdot\|_{\infty\to\infty}$ norm. The perturbation is symmetric, so its spectral norm is bounded by its infinity norm. Choosing $\nu$ larger than $3^q\cdot8C_J\Jan V_R=\mathcal{O}(J\log^d(1/\eps))$ yields the perturbation to be less than $3^{-q}/4$. Combining this with \Cref{eq:Gamma-spectral} on the assumed event \Cref{eq:Gamma-close} gives us \Cref{eq:Jacobian-conditioned-app}.
\end{proof}

\begin{lemma}[\Cref{lem:Hessian-conditioning}, Uniform Hessian conditioning]
Uniformly for \((y,\mu)\in\cY\times\cU\),
\begin{equation}
    \norm{\mathsf H(y,\mu)-\Gamma}_{\infty\to\infty}
    \le C_JV_R\left(\frac{\Jan}{\nu}+\frac{\Jan^2}{\nu^2}\right)
    +C_2V_R\frac{r_E}{\nu}
    \label{eq:Hessian-close-app}
\end{equation}
On the event that \Cref{eq:Gamma-close} holds, there exists an optimization field threshold $\nu_{\rm opt}=\mathcal{O}(JV_R)$ such that when $\nu\geq \nu_{\rm opt}$, we have
\begin{equation}
    \mathsf H(y,\mu)\succeq\frac{I}{2\cdot 3^q},
\end{equation}
and
\begin{equation}
    \norm{I-\Gamma^{-1}\mathsf H(y,\mu)}_{\infty\to\infty}\le 1/3.
\end{equation}
In particular, \(\mathcal L(y,\cdot)\) is uniformly strongly convex on \(\cU\).
\end{lemma}

\begin{proof}
By triangle inequality and \Cref{eq:Hessian-split},
\begin{equation}
    \norm{\mathsf H(y,\mu)-\Gamma}_{\infty\to\infty} \le \norm{\mathsf G(\mu)-\Gamma}_{\infty\to\infty}+\frac{1}{m}\norm{\sum_{\rho}r_\rho(y,u)\nabla^2F_\rho(\mu)}_{\infty\to\infty}.
\end{equation}
The $\norm{\mathsf G(\mu)-\Gamma}_{\infty\to\infty}$ term is bounded by \Cref{lem:Jacobian-perturbation}.  Consider the residual term in \Cref{eq:Hessian-split}, which is expressed as
\begin{equation}
    \frac{1}{m}\norm{\sum_{\rho}r_\rho(y,u)\nabla^2F_\rho(\mu)}_{\infty\to\infty} = \max_a\sum_b\left|
    \frac1m\sum_\rho r_\rho(y,\mu)
    \partial_a\partial_bF_\rho(\mu)
    \right|.
\end{equation}
Fix a coefficient index \(a\).  Only rows $\rho$ whose patch $B_{R}(k)$ contains \(\supp(P_a)\) can contribute to row \(a\), so there are at most \(mV_R\) such rows.  Denoting the set of such rows $\rho$ by $\cA_\rho$, $|\cA_\rho|\leq mV_R$, by \Cref{eq:second-derivative-bound} and \Cref{eq:residual-bound},
\begin{align}
    &\sum_b\left|
    \frac1m\sum_\rho r_\rho(y,\mu)
    \partial_a\partial_bF_\rho(\mu)
    \right|
    \le\frac1m\sum_{\rho:\,a\in\cA_\rho}
    |r_\rho(y,\mu)|\sum_b|\partial_a\partial_bF_\rho(\mu)|
    \le \frac{C_2V_Rr_E}{\nu}.
    \label{eq:residual-Hessian-row-bound}
\end{align}
Combining this with \Cref{eq:Jacobian-Gram-close} proves \Cref{eq:Hessian-close-app}.

On the assumed event \Cref{eq:Gamma-close}, \Cref{eq:Gamma-spectral} gives $\lambda_{\min}(\Gamma) \ge (3/4)\cdot3^{-q}$. We can choose $\nu=\mathcal{O}(V_R\Jan)=\mathcal{O}(J\log^d(1/\eps))$ such that 
\begin{equation}
    \eta_{\mathsf{H}} = \|\Gamma-\mathsf H(y,\mu)\|_{\infty\to\infty} \leq  C_JV_R\left(\frac{\Jan}{\nu}+\frac{\Jan^2}{\nu^2}\right)
    +C_2V_R\frac{r_E}{\nu} \le \frac{1}{4}\cdot3^{-q}.
\end{equation}
We have
\begin{equation}
    \lambda_{\min}(\mathsf H)
    \ge\lambda_{\min}(\Gamma)-\eta_{\mathsf{H}}
    \ge\frac{3}{4}\cdot3^{-q}-\frac{1}{4}\cdot3^{-q}
    =\frac{1}{2}\cdot3^{-q},
    \label{eq:Hessian-Weyl-proof}
\end{equation}
which proves strong convexity.  Finally,
\begin{equation}
    \norm{I-\Gamma^{-1}\mathsf H}_{\infty\to\infty}
    =\norm{\Gamma^{-1}(\Gamma-\mathsf H)}_{\infty\to\infty}
    \le C_\Gamma\eta_{\mathsf{H}}
    \le1/3,
    \label{eq:preconditioned-Hessian-proof}
\end{equation}
where the final inequality uses \Cref{eq:Gamma-inverse} where the constant $C_\Gamma$ is also defined.
\end{proof}

\begin{theorem}[\Cref{thm:local-inverse}, Stability and convergence of the fixed-preconditioner gradient iteration]
We condition on the event that \Cref{eq:Gamma-close} in \Cref{thm:Walsh-conditioning} is correct, and let $r_y$ in \Cref{eq:data-neighborhood} satisfy (with $C_\Gamma$ defined in \Cref{eq:Gamma-inverse} and $C_y$ in \Cref{eq:Cy})
\begin{equation}
 r_y\le\frac{1}{2C_\Gamma C_y}.
\end{equation}
There exists an optimization field threshold $\nu_{\rm opt}$:
\begin{equation}
    {\nu_{\rm opt}}=\mathcal{O}(JV_R)=\mathcal{O}(JR^d)
\end{equation}
such that for every $\nu\geq \nu_{\rm opt}$, \(y\in\cY\), the loss \(\mathcal L(y,\cdot)\) has a unique stationary point \(\mu^\star(y)\in\cU\).  It is the unique minimizer in \(\cU\), and
\begin{equation}
 \norm{{\mu^\star(y)}-\lambda}_\infty
 \le\frac{3C_\Gamma C_y\norm{y-y_*}_\infty}{2}.
\end{equation}
For every initialization \(\mu^{(0)}\in\cU\), the iteration
$\mu^{(t+1)}={\mathcal T_y}(\mu^{(t)}),$ with $\mathcal T_y$ defined in \Cref{eq:iteration-map}, converges geometrically to \(\mu^\star(y)\), with contraction factor at most \(1/3\).
\end{theorem}

\begin{proof}
Recall that $f(y,\mu):=\nabla_\mu \mathcal L(y,\mu)$, and $\mathcal T_y(\mu):=\mu-\Gamma^{-1}f(y,\mu)$, and that $\cU = \left\{\mu:\norm{\mu-\lambda}_\infty\le 1\right\}$.  We work on the assumed event \Cref{eq:Gamma-close}. By \Cref{lem:Hessian-conditioning}, there exists the optimization field threshold $\nu_{\rm opt}=\mathcal{O}(J\log^d(1/\eps))$ such that when $\nu\geq \nu_{\rm opt}$,
\begin{equation}
\label{eq:preconditioned-Hessian-used}
    \norm{I-\Gamma^{-1}\mathsf H(y,\mu)}_{\infty\to\infty} \le 1/3.
\end{equation}
We divide the proof into four steps.

\medskip
\noindent
\textbf{Step 1: \(\mathcal T_y\) is a contraction on \(\cU\).}
Differentiating the fixed-point map gives
\begin{equation}
\label{eq:D-Phi-detailed}
    D_\mu\mathcal T_y(\mu) = I-\Gamma^{-1}\mathsf H(y,\mu).
\end{equation}
Fix arbitrary \(\mu,\mu'\in\cU\). Since \(\cU\) is convex, the line segment $\mu_t:=\mu'+t(\mu-\mu')$, where $0\le t\le1$, is contained in \(\cU\). The fundamental theorem of calculus therefore gives
\begin{equation}
    \mathcal T_y(\mu)-\mathcal T_y(\mu') = \int_0^1 D_\mu\mathcal T_y(\mu_t) (\mu-\mu') \,\mathrm dt.
\end{equation}
Using \Cref{eq:preconditioned-Hessian-used}, we obtain
\begin{equation}
\label{eq:Phi-contraction-detailed}
    \norm{ \mathcal T_y(\mu)-\mathcal T_y(\mu') }_\infty \le \int_0^1 \norm{ D_\mu\mathcal T_y(\mu_t) }_{\infty\to\infty} \norm{\mu-\mu'}_\infty \,\mathrm dt \le \frac{1}{3} \norm{\mu-\mu'}_\infty.
\end{equation}
Thus \(\mathcal T_y\) is a contraction on \(\cU\) with contraction factor at most $1/3$.

\medskip
\noindent
\textbf{Step 2: \(\mathcal T_y\) maps \(\cU\) into itself.}
At the true parameter \(\lambda\), the definition of the score map gives
\begin{equation}
\label{eq:Phi-at-lambda-detailed}
    \mathcal T_y(\lambda)-\lambda = \Gamma^{-1} \frac1m \mathsf J(\lambda)^{\mathsf T} (y-y_*).
\end{equation}
Hence, by \Cref{eq:Cy},
\begin{equation}
\label{eq:Phi-lambda-data-bound}
\begin{aligned}
    \norm{ \mathcal T_y(\lambda)-\lambda }_\infty \le C_\Gamma C_y \norm{y-y_*}_\infty \le C_\Gamma C_y r_y.
\end{aligned}
\end{equation}
The assumed radius condition $r_y\le{1}/{(2C_\Gamma C_y)}$ therefore implies
\begin{equation}
\label{eq:Phi-lambda-half-basin-detailed}
    \norm{ \mathcal T_y(\lambda)-\lambda }_\infty \le 1/2.
\end{equation}
Now let \(\mu\in\cU\). Applying the contraction estimate
\Cref{eq:Phi-contraction-detailed} with \(\mu'=\lambda\), we obtain
\begin{align}
    \norm{ \mathcal T_y(\mu)-\lambda }_\infty \le \norm{ \mathcal T_y(\mu)-\mathcal T_y(\lambda) }_\infty + \norm{ \mathcal T_y(\lambda)-\lambda }_\infty \le \frac{1}{3}\norm{\mu-\lambda}_\infty + \frac{1}2 < 1.
\label{eq:self-map-proof-detailed}
\end{align}
Consequently, $\mathcal T_y(\cU)\subseteq\cU$. Because \(\cU\) is closed in a finite-dimensional normed space, it is complete under the metric induced by \(\norm{\cdot}_\infty\). Banach's fixed-point theorem therefore gives a unique point and minimizer $\mu^\star(y)\in\cU$ satisfying $\mathcal T_y(\mu^\star(y))=\mu^\star(y).$ Moreover, for every initialization \(\mu^{(0)}\in\cU\), all iterates
$\mu^{(t+1)}=\mathcal T_y(\mu^{(t)})$ remain in \(\cU\), and
\begin{equation}
\label{eq:geometric-convergence-detailed}
\begin{aligned}
    \norm{ \mu^{(t)}-\mu^\star(y) }_\infty = \norm{ \mathcal T_y^{\,t}(\mu^{(0)}) - \mathcal T_y^{\,t}(\mu^\star(y)) }_\infty \le \frac{1}{3^t} \norm{ \mu^{(0)}-\mu^\star(y) }_\infty.
\end{aligned}
\end{equation}
Thus the iteration converges geometrically with contraction factor at most $1/3$.

\medskip
\noindent
\textbf{Step 3: stability with respect to the data.}
Using the fixed-point identities for \(\mu^\star(y)\) and \(\lambda\), we obtain
\begin{align}
    \norm{ \mu^\star(y)-\lambda }_\infty &= \norm{ \mathcal T_y(\mu^\star(y)) - \mathcal T_{y_*}(\lambda) }_\infty \le \norm{ \mathcal T_y(\mu^\star(y)) - \mathcal T_y(\lambda) }_\infty + \norm{ \mathcal T_y(\lambda) - \mathcal T_{y_*}(\lambda) }_\infty.
\label{eq:stability-split-detailed}
\end{align}
The first term is bounded using the contraction property:
\begin{equation}
    \norm{ \mathcal T_y(\mu^\star(y)) - \mathcal T_y(\lambda) }_\infty \le {\norm{ \mu^\star(y)-\lambda }_\infty}/{3}.
\end{equation}
For the second term, \Cref{eq:Phi-at-lambda-detailed} gives
\begin{equation}
    \norm{ \mathcal T_y(\lambda) - \mathcal T_{y_*}(\lambda) }_\infty = \norm{ \Gamma^{-1} \frac1m \mathsf J(\lambda)^{\mathsf T} (y-y_*) }_\infty \le C_\Gamma C_y \norm{y-y_*}_\infty.
\end{equation}
Substitution into \Cref{eq:stability-split-detailed} gives
\begin{equation}
\label{eq:stability-before-rearranging}
    \norm{ \mu^\star(y)-\lambda }_\infty \le \frac{1}{3} \norm{ \mu^\star(y)-\lambda }_\infty + C_\Gamma C_y \norm{y-y_*}_\infty.
\end{equation}
Rearranging yields
\begin{align}
    \norm{ \mu^\star(y)-\lambda }_\infty \le \frac{3C_\Gamma C_y \norm{y-y_*}_\infty}{2}
\label{eq:stability-first-detailed}
\end{align}
This proves the existence and uniqueness of the stationary point, its minimizing property, the stability bound, and the geometric convergence of the fixed-point iteration.
\end{proof}

\section{Proof of parallelization (\Cref{sec:parallelization})}
\label{app:parallelization}
\begin{lemma}[\Cref{lem:parallel-code-factorization}, Factorization on the defect subspace]
$\mathcal{H}_{\rm def}(\mathcal{D},\ell)$ is an invariant subspace of $\widehat D_{\mathcal D,\ell}^{(<P_{\mathcal D})}$:
\begin{equation}
    (I-\Pi_{\mathcal D}) \widehat D_{\mathcal D,\ell}^{(<P_{\mathcal D})} \Pi_{\mathcal D}=0.
\end{equation}
Moreover, for each $k\in\mathcal D$,
define
\begin{equation}
    E_{\mathcal D,\ell;k}^{(<P_{\mathcal D})} := \frac12\left( \bra{\Phi_{\{k\}}^{\mathcal D,\ell}} \widehat D_{\mathcal D,\ell}^{(<P_{\mathcal D})} \ket{\Phi_{\{k\}}^{\mathcal D,\ell}} - \bra{\Phi_{\varnothing}^{\mathcal D,\ell}} \widehat D_{\mathcal D,\ell}^{(<P_{\mathcal D})} \ket{\Phi_{\varnothing}^{\mathcal D,\ell}} \right),
\end{equation}
and then is a real number $c_{\mathcal D}$ such that
\begin{align}
    \Pi_{\mathcal D} \widehat D_{\mathcal D,\ell}^{(<P_{\mathcal D})} \Pi_{\mathcal D} =\Pi_{\mathcal D}\left( c_{\mathcal D}I +2\sum_{k\in\mathcal D} E_{\mathcal D,\ell;k}^{(<P_{\mathcal D})}Q_k^{\mathcal D,\ell} \right)\Pi_{\mathcal D}.
\end{align}
Therefore, for the evolution generated by $G_{\mathcal D,\ell}^{(<P_{\mathcal D})}:=\nu N_{\mathcal D,\ell}+\widehat D_{\mathcal D,\ell}^{(<P_{\mathcal D})}$, the expectation value of $\tau_{k,+}^{\mathcal D,\ell}$ at every defect site $k$ at time $t$ is
\begin{equation}
    \bra{\Psi_+^{\mathcal D,\ell}} e^{iG_{\mathcal D,\ell}^{(<P_{\mathcal D})}t} \tau_{k,+}^{\mathcal D,\ell} e^{-iG_{\mathcal D,\ell}^{(<P_{\mathcal D})}t} \ket{\Psi_+^{\mathcal D,\ell}} =e^{i[\nu+2E_{\mathcal D,\ell;k}^{(<P_{\mathcal D})}]t}.
\end{equation}
In the above $Q_k^{\mathcal{D},\ell}$ is defined as in \Cref{eq:Q-row} and $\tau_{k,+}^{\mathcal{D},\ell}$ in \Cref{eq:parallel-O-plus}.
\end{lemma}

\begin{proof}
For simplicity, we write $\Pi:=\Pi_{\mathcal D}$, $N:=N_{\mathcal D,\ell}$, $Q_k:=Q_k^{\mathcal D,\ell}$, $G:=G_{\mathcal D,\ell}^{(<P_{\mathcal D})}$, and $O_k:=\tau_{k,+}^{\mathcal D,\ell}$. We then write $\widehat D_{\mathcal D,\ell}^{(<P_{\mathcal D})}$ as a sum of interaction terms,  $\widehat D_{\mathcal D,\ell}^{(<P_{\mathcal D})}=\sum_X W_X$

By \Cref{lem:D-taylor-expansion-commuting-support,eq:parallel-separation,eq:parallel-bridge-distance}, every term $W_X$ is Hermitian and satisfies $|X\cap\mathcal D|\le 1$ and $[W_X,N]=0$. For support $X$, define
\begin{equation}
    N_X:=\sum_{k\in X\cap\mathcal D}Q_k+
2\sum_{j\in X\setminus\mathcal D}Q_j .  
\end{equation}
Since $W_X$ is supported on $X$, it commutes with $N-N_X$. Therefore,
$[W_X,N_X]=[W_X,N]-[W_X,N-N_X]=0.$

We consider two cases. First suppose that $X\cap\mathcal D=\varnothing$. In this case, every defect-subspace basis state $\ket{\Phi_A^{\mathcal D,\ell}}$ has the same restriction on $X$:
\begin{equation}
\ket{\Phi_{X,0}}=\bigotimes_{j\in X}\ket{\phi_{j,0}} .    
\end{equation}
This is the unique eigenstate of $N_X$ with eigenvalue zero. Since $W_X$ commutes with $N_X$, it must preserve this one-dimensional eigenspace, and therefore $W_X\ket{\Phi_A^{\mathcal D,\ell}}=a_X\ket{\Phi_A^{\mathcal D,\ell}}$ where $a_X=\braket{\Phi_A^{\mathcal D,\ell}|W_X|\Phi_A^{\mathcal D,\ell}}=\braket{\Phi_{X,0}|W_X|\Phi_{X,0}}$, and is therefore independent of $A$.

Now suppose that $X\cap\mathcal D=\{k\}$. The restriction of a defect-subspace basis state $\ket{\Phi_A^{\mathcal D,\ell}}$ to $X$ is either $\ket{\Phi_{X,0}}=\bigotimes_{j\in X}\ket{\phi_{j,0}}$, or $\ket{\Phi_{X,1}^{k}}=\ket{\phi_{k,1}}\bigotimes_{j\in X\setminus\{k\}}\ket{\phi_{j,0}}$ . These states satisfy $N_X\ket{\Phi_{X,0}}=0$, and $N_X\ket{\Phi_{X,1}^{k}}=\ket{\Phi_{X,1}^{k}}$.

The eigenvalue-zero and eigenvalue-one eigenspaces of $N_X$ are both one dimensional. Indeed, the contribution of every non-defect site to $N_X$ is either zero or two, whereas the defect site $k$ contributes either zero or one. Hence, the only state with eigenvalue one is the state where the defect site $k$ is excited and all other sites remain in their zero states.  Since $[W_X,N_X]=0$ and both eigenspaces are one dimensional, $\ket{\Phi_{X,0}}$ and $\ket{\Phi_{X,1}^{k}}$ are eigenstates of $W_X$. Because $W_X$ is supported on $X$, every $\ket{\Phi_A^{\mathcal D,\ell}}$ is therefore an eigenstate of $W_X$, with its eigenvalue depending only on whether $k\in A$. Thus,

\begin{equation}
W_X\ket{\Phi_{A}^{\mathcal D,\ell}} =
\begin{cases}
    a_X\ket{\Phi_{A}^{\mathcal D,\ell}}, & k\notin A,\\
    d_X\ket{\Phi_{A}^{\mathcal D,\ell}}, & k\in A,
\end{cases}
\label{eq:parallel-W-action}
\end{equation}
where $a_X=\braket{\Phi_{X,0}|W_X|\Phi_{X,0}}$, $d_X=\braket{\Phi_{X,1}^{k}|W_X|\Phi_{X,1}^{k}}$, and both are real numbers that are independent of $A$. Therefore, every term $W_X$ maps the defect subspace into itself, and $(I-\Pi)W_X\Pi=0$. Summing over all interaction terms proves
\begin{equation}\label{eq:parallel-no-leakage-proof}
    (I-\Pi)\widehat D_{\mathcal D,\ell}^{(<P_{\mathcal D})}\Pi=0.
\end{equation}

Next we determine the restriction of the Hamiltonian to the defect
subspace. On the defect subspace,
\begin{equation}
Q_k\ket{\Phi_A^{\mathcal D,\ell}}
=\ind[k\in A]\ket{\Phi_A^{\mathcal D,\ell}} .
\label{eq:parallel-Q-action}
\end{equation}
Therefore, the action in \Cref{eq:parallel-W-action} can be written as $\Pi W_X\Pi=\Pi(a_X I+(d_X-a_X)Q_k)\Pi$ when $X$ contains the defect $k$, and as $\Pi W_X\Pi=a_X\Pi$ when $X$ contains no defect.

Set \(c_{\mathcal D}:=\bra{\Phi_{\varnothing}^{\mathcal D,\ell}}\widehat D_{\mathcal D,\ell}^{(<P_{\mathcal D})}\ket{\Phi_{\varnothing}^{\mathcal D,\ell}}=\sum_X a_X.\) For each $k\in\mathcal D$, the definition of $E_{\mathcal D,\ell;k}^{(<P_{\mathcal D})}$ in the lemma and \Cref{eq:parallel-W-action} give

\begin{align}
    2E_{\mathcal D,\ell;k}^{(<P_{\mathcal D})}
    &= \bra{\Phi_{\{k\}}^{\mathcal D,\ell}} \widehat D_{\mathcal D,\ell}^{(<P_{\mathcal D})} \ket{\Phi_{\{k\}}^{\mathcal D,\ell}}
    -\bra{\Phi_{\varnothing}^{\mathcal D,\ell}} \widehat D_{\mathcal D,\ell}^{(<P_{\mathcal D})}\ket{\Phi_{\varnothing}^{\mathcal D,\ell}}
    =\sum_{X:X\cap\mathcal D=\{k\}}(d_X-a_X).
\label{eq:parallel-half-gap-coefficients}
\end{align}

Then summing over all $X$ gives
\begin{equation}
\Pi \widehat D_{\mathcal D,\ell}^{(<P_{\mathcal D})} \Pi = \Pi \left( c_{\mathcal D}I + 2\sum_{k\in\mathcal D} E_{\mathcal D,\ell;k}^{(<P_{\mathcal D})} Q_k \right) \Pi .
\label{eq:parallel-factorization-proof}
\end{equation}

Since all non-defect qubits are in the $Q_j=0$ state on the defect
subspace,
\begin{equation}
\Pi N\Pi = \Pi \left( \sum_{k\in\mathcal D}Q_k \right) \Pi .
\label{eq:parallel-charge-restriction}
\end{equation}
Combining
\Cref{eq:parallel-factorization-proof} and
\Cref{eq:parallel-charge-restriction}, we obtain
\begin{equation}
\Pi G \Pi = \Pi \left( c_{\mathcal D}I + \sum_{k\in\mathcal D} \left( \nu+ 2E_{\mathcal D,\ell;k}^{(<P_{\mathcal D})} \right) Q_k \right) \Pi .
\label{eq:parallel-projected-G}
\end{equation}

We can see that $(I-\Pi) N \Pi = 0$ holds. Combining with \Cref{eq:parallel-no-leakage-proof} we have $(I-\Pi) G_{\mathcal D,\ell}^{(<P_{\mathcal D})}\Pi = 0$. Therefore, we have
\begin{equation}
G_{\mathcal D,\ell}^{(<P_{\mathcal D})} \Pi = \Pi \left( c_{\mathcal D}I + \sum_{k\in\mathcal D} \left( \nu+ 2E_{\mathcal D,\ell;k}^{(<P_{\mathcal D})} \right) Q_k \right) \Pi .
\end{equation}
And this yields 
\begin{align}
    &G\ket{\Phi_A^{\mathcal D,\ell}} = G\Pi\ket{\Phi_A^{\mathcal D,\ell}} \nonumber \\ 
    &= \Pi \left( c_{\mathcal D}I + \sum_{k\in\mathcal D} \left( \nu+ 2E_{\mathcal D,\ell;k}^{(<P_{\mathcal D})} \right) Q_k \right) \Pi\ket{\Phi_A^{\mathcal D,\ell}} 
    = \left( c_{\mathcal D} + \sum_{j\in A} \left( \nu+ 2E_{\mathcal D,\ell;j}^{(<P_{\mathcal D})} \right) \right) \ket{\Phi_A^{\mathcal D,\ell}} .
\label{eq:parallel-code-energy}
\end{align}

It remains to compute the Ramsey signal. From \Cref{eq:parallel-code-energy}, every defect-subspace basis state is an eigenstate of $G$, so applying the time evolution of a defect-subspace basis state is
\begin{equation}
    e^{-iGt}\ket{\Phi_A^{\mathcal D,\ell}} = \exp\left( -i \left[ c_{\mathcal D} + \sum_{j\in A} \left( \nu+ 2E_{\mathcal D,\ell;j}^{(<P_{\mathcal D})} \right) \right]t \right) \ket{\Phi_A^{\mathcal D,\ell}} .
\label{eq:parallel-forward-evolution}
\end{equation}

Now consider $\tau_{k,+}^{\mathcal D,\ell}$. By definition,
it changes only the state of defect $k$:
\begin{equation}
    \tau_{k,+}^{\mathcal D,\ell}\ket{\Phi_A^{\mathcal D,\ell}} =
    \begin{cases}
        2\ket{\Phi_{A\cup\{k\}}^{\mathcal D,\ell}}, & k\notin A,\\
        0, & k\in A .
    \end{cases}
\label{eq:parallel-O-action}
\end{equation}

For $k\notin A$, applying the three operations appearing in the
Ramsey signal gives
\begin{align}
    e^{iGt} \tau_{k,+}^{\mathcal D,\ell} e^{-iGt} \ket{\Phi_A^{\mathcal D,\ell}} = 2 e^{-i\varepsilon_A t} e^{i\varepsilon_{A\cup\{k\}}t} \ket{\Phi_{A\cup\{k\}}^{\mathcal D,\ell}},
    \label{eq:parallel-O-three-step}
\end{align}
where \(\varepsilon_A := c_{\mathcal D} + \sum_{j\in A} \left( \nu+ 2E_{\mathcal D,\ell;j}^{(<P_{\mathcal D})} \right).\) The phase difference is
\begin{align}
    \varepsilon_{A\cup\{k\}}-\varepsilon_A 
    &= \left( c_{\mathcal D} + \sum_{j\in A} \left( \nu+ 2E_{\mathcal D,\ell;j}^{(<P_{\mathcal D})} \right) + \nu+ 2E_{\mathcal D,\ell;k}^{(<P_{\mathcal D})} \right) - \left( c_{\mathcal D} + \sum_{j\in A} \left( \nu+ 2E_{\mathcal D,\ell;j}^{(<P_{\mathcal D})} \right) \right) \\ 
    &= \nu+ 2E_{\mathcal D,\ell;k}^{(<P_{\mathcal D})}.
\end{align}

Hence, using \Cref{eq:parallel-O-action} we have
\begin{equation}
    e^{iGt} \tau_{k,+}^{\mathcal D,\ell} e^{-iGt} \ket{\Phi_A^{\mathcal D,\ell}} = e^{i\left[\nu+ 2E_{\mathcal D,\ell;k}^{(<P_{\mathcal D})}\right]t} \tau_{k,+}^{\mathcal D,\ell} \ket{\Phi_A^{\mathcal D,\ell}} .
    \label{eq:parallel-O-action-evolution}
\end{equation}
The same equation holds for $k\in A$, because both sides vanish to $0$. Finally, using the expansion
\begin{equation}
    \ket{\Psi_+^{\mathcal D,\ell}} = 2^{-|\mathcal D|/2} \sum_{A\subseteq\mathcal D} \ket{\Phi_A^{\mathcal D,\ell}},
\end{equation}
and applying \Cref{eq:parallel-O-action-evolution}, we obtain
\begin{align}
    \bra{\Psi_+^{\mathcal D,\ell}} e^{iGt} \tau_{k,+}^{\mathcal D,\ell} e^{-iGt} \ket{\Psi_+^{\mathcal D,\ell}}= 
    e^{i\left[\nu+ 2E_{\mathcal D,\ell;k}^{(<P_{\mathcal D})}\right]t} \bra{\Psi_+^{\mathcal D,\ell}} \tau_{k,+}^{\mathcal D,\ell} \ket{\Psi_+^{\mathcal D,\ell}} .
\end{align}
Since the $k$-th defect is prepared in $\left({\ket{\phi_{k,0}^{\mathcal D,\ell}}+\ket{\phi_{k,1}^{\mathcal D,\ell}}}\right)/{\sqrt2}$ we have $\bra{\Psi_+^{\mathcal D,\ell}}
\tau_{k,+}^{\mathcal D,\ell}
\ket{\Psi_+^{\mathcal D,\ell}}
=1$. Therefore,
\begin{equation}
    \bra{\Psi_+^{\mathcal D,\ell}} e^{iG_{\mathcal D,\ell}^{(<P_{\mathcal D})}t} \tau_{k,+}^{\mathcal D,\ell} e^{-iG_{\mathcal D,\ell}^{(<P_{\mathcal D})}t} \ket{\Psi_+^{\mathcal D,\ell}} = e^{i\left[\nu+ 2E_{\mathcal D,\ell;k}^{(<P_{\mathcal D})}\right]t}.
\end{equation}
\end{proof}

\begin{lemma}[\Cref{lem:parallel-local-perturbation}, Local dynamics under a small interaction]
Let \(A\) and \(B\) be Hermitian interactions satisfying \(\norm A_{2\kappa}\le J_A,\norm B_{2\kappa}\le\delta_B,\) for some \(\kappa>0\). For an operator \(O_x\) supported on one site,
\begin{equation}
    \norm{ e^{i(A+B)t}O_xe^{-i(A+B)t} -e^{iAt}O_xe^{-iAt}} \le C_{d,\kappa}\norm{O_x}\, \delta_B|t|(1+v_A|t|)^d,
\end{equation}
where \(v_A=C(d)\kappa^{-(d+2)}e^\kappa J_A\) is the Lieb-Robinson velocity. The prefactor \(C_{d,\kappa}\) can be chosen such that $C_{d,\kappa}\le C_d(1+\kappa^{-1})^d$. Here \(C_d\) and \(C(d)\) depend only on \(d\), and are independent of the system size.
\end{lemma}

\begin{proof}
It suffices to consider \(t\ge0\), since replacing \(A,B\) by \(-A,-B\) gives the same bound for \(t<0\). Let $O_x(t):=e^{iAt}O_xe^{-iAt}$.

Using Duhamel's formula, we have
\begin{equation}
    e^{i(A+B)t}O_xe^{-i(A+B)t} - e^{iAt}O_xe^{-iAt} = i\int_0^t e^{i(A+B)(t-s)} [B,O_x(s)] e^{-i(A+B)(t-s)} \,ds .
    \label{eq:parallel-duhamel}
\end{equation}
Therefore,
\begin{equation}
    \norm{ e^{i(A+B)t}O_xe^{-i(A+B)t} - e^{iAt}O_xe^{-iAt} } \le \int_0^{|t|} \norm{[B,O_x(s)]}\,ds .
\label{eq:parallel-duhamel-bound}
\end{equation}

Write $B=\sum_X B_X$, where \(B_X\) is supported on \(X\). Then
\begin{equation}
    \norm{[B,O_x(s)]} \le \sum_X\norm{[B_X,O_x(s)]}.
\label{eq:parallel-commutator-sum}
\end{equation}

We apply the Lieb--Robinson bound in
\Cref{lem:ADHH-LR} to the interaction \(A\). For each \(X\),
the evolved observable \(O_x(s)\) satisfies
\begin{equation}
    \norm{[B_X,O_x(s)]} \le {2}\norm{B_X}\norm{O_x} e^{-\kappa(\dist(X,x)-v_A s)}
\label{eq:parallel-LR-bound}
\end{equation}
when \(\dist(X,x)>v_As\).

For nearby terms, we use the trivial bound $\norm{[B_X,O_x(s)]}\le2\norm{B_X}\norm{O_x}$. Combining the two bounds gives
\begin{equation}
    \norm{[B_X,O_x(s)]}\le2\norm{B_X}\norm{O_x}\min\left\{1,e^{-\kappa(\dist(X,x)-v_As)}\right\}.
    \label{eq:parallel-local-LR-min}
\end{equation}

We now bound the sum over \(X\), keeping track of the dependence on \(\kappa\). Throughout this proof, \(C_d\) denotes a positive constant depending only on \(d\), whose value may change from line to line. By \Cref{eq:local-norm}, \(\sum_{X\ni y}\norm{B_X}\le\delta_B\) for every site \(y\). Split the interaction into terms with \(\dist(X,x)\le v_As\) and the remaining terms. Each support in the first part contains a site within distance \(v_As\) of \(x\), so
\begin{equation}
    \sum_{X:\dist(X,x)\le v_As}\norm{B_X} \le \sum_{\substack{y\in\Lambda:\\\dist(y,x)\le v_As}} \sum_{X\ni y}\norm{B_X} \le C_d\delta_B(1+v_As)^d.
\label{eq:parallel-near-bound}
\end{equation}

For the second part, assign each support to a site nearest to \(x\), and then sum over lattice shells of width one. The shell
\(j<\dist(y,x)-v_As\le j+1\), for \(j\ge0\), contains at most
\begin{equation}
    C_d(1+v_As+j)^{d-1}\leq C_d(1+v_As)^{d-1}(1+j)^{d-1}\leq C_d(1+v_As)^{d}(1+j)^{d-1}    
\end{equation}
sites (this bound is deliberately loose to match \Cref{eq:parallel-near-bound}). Thus
\begin{align}
    \sum_{X:\dist(X,x)>v_As}\norm{B_X}e^{-\kappa(\dist(X,x)-v_As)} &\le \delta_B\sum_{\substack{y\in\Lambda:\\\dist(y,x)>v_As}} e^{-\kappa(\dist(y,x)-v_As)} \\
    &\le C_d\delta_B(1+v_As)^d \sum_{j=0}^{\infty}e^{-\kappa j}(1+j)^{d-1} \\
    & \le C_d(1+\kappa^{-1})^d\delta_B(1+v_As)^d.
    \label{eq:parallel-far-bound}
\end{align}
Summing the geometric series gives
\begin{equation}
    \sum_{j=0}^{\infty}e^{-\kappa j}(1+j)^{d-1} \le (d-1)!\sum_{j=0}^{\infty} e^{-\kappa j}\binom{j+d-1}{d-1} =\frac{(d-1)!}{(1-e^{-\kappa})^d} \le (d-1)!(1+\kappa^{-1})^d,
\end{equation}
where the final inequality follows from \(e^\kappa\ge1+\kappa\). Therefore, combining \Cref{eq:parallel-local-LR-min,eq:parallel-near-bound,eq:parallel-far-bound}, we obtain
\begin{equation}
    \norm{[B,O_x(s)]} \le C_d(1+\kappa^{-1})^d\norm{O_x} \delta_B(1+v_As)^d .
    \label{eq:parallel-commutator-final}
\end{equation}

Finally, substituting \Cref{eq:parallel-commutator-final} into \Cref{eq:parallel-duhamel-bound} gives
\begin{align}
    \norm{ e^{i(A+B)t}O_xe^{-i(A+B)t} - e^{iAt}O_xe^{-iAt} } &\le C_d(1+\kappa^{-1})^d\norm{O_x} \delta_B \int_0^{|t|} (1+v_As)^d ds \\ 
    &\le C_d(1+\kappa^{-1})^d\norm{O_x} \delta_B |t|(1+v_A|t|)^d .
\end{align}
Taking \(C_{d,\kappa}=C_d(1+\kappa^{-1})^d\), with \(C_d\) sufficiently large and depending only on \(d\), proves the lemma.
\end{proof}

\begin{lemma}[\Cref{lem:parallel-stopped-dynamics}, Single-ball comparison in a parallel experiment]
Assume \(\nu\ge\max\{\bar\nu_0,54\pi\Jan/\kappa_0\}\), \(1\le s=n_{\rm par}\le n_*(\nu)\), and \(L_{\mathcal D}\ge2R+(P+1)r_0\), with integer \(P\ge1\). Set \(\mathcal K_k:=\nu N_k+D_{k,s}+\nu N_{{\rm rem},k}\). For every bounded operator \(O\) and every real \(t\),
\begin{equation}
    \norm{e^{iG_{\mathcal D,\ell}^{\mathcal B}t}O e^{-iG_{\mathcal D,\ell}^{\mathcal B}t} -e^{i\mathcal K_kt}Oe^{-i\mathcal K_kt}}\le C\norm O\Jan|B_k|\left\{ \nu^{-1}+|t|\left[(2/3)^s+(P+1)2^{-P}\right]\right\}.
\end{equation}
The constant is independent of \(n,|\mathcal D|,s,P,\nu\).
\end{lemma}

\begin{proof}[Proof of \Cref{lem:parallel-stopped-dynamics}]
Fix \(k\in\mathcal D\), set \(s=n_{\rm par}\), and use the
operators in \Cref{eq:parallel-single-ball-charge}. Throughout this
proof, \(H=H(\lambda)\), so \(\norm H_{\kappa_0}\le J\). We first conjugate the physical Hamiltonian \(G_{\mathcal D,\ell}^{\mathcal B}=H+\nu N_k+\nu N_{{\rm rem},k}\) by the unitary \(Y_{k,s}\) constructed by the recursion for \(H+\nu N_k\). Besides the recursion remainder \(V_{k,s}\), this produces an additional remainder \(W_{k,s}\). More precisely, the recursion identity gives \(Y_{k,s}(H+\nu N_k)Y_{k,s}^\dagger=\nu N_k+D_{k,s}+V_{k,s}\). Adding the conjugated term \(\nu N_{{\rm rem},k}\) and recalling \(\mathcal K_k=\nu N_k+D_{k,s}+\nu N_{{\rm rem},k}\), we obtain
\begin{align}
    Y_{k,s}G_{\mathcal D,\ell}^{\mathcal B}Y_{k,s}^\dagger&=\mathcal K_k+V_{k,s}+W_{k,s},\\
    W_{k,s}&:=\nu\left(Y_{k,s}N_{{\rm rem},k}Y_{k,s}^\dagger-N_{{\rm rem},k}\right).
    \label{eq:parallel-stopped-dynamics-app}
\end{align}
The operator \(W_{k,s}\) is therefore the change in the fields on the
other controlled balls induced by the unitary for the ball at \(k\).

For a label \(X\) disjoint from \(B_k\), every operator supported in \(X\) commutes with \(N_k\). As in the proof of \Cref{thm:ADHH-dynamics}, this implies \((V_{k,j})_X=(A_j)_X=0\). Here \(A_j\) denotes the generator at step \(j\) of the recursion with charge \(N_k\). All its nonzero components, and those of \(V_{k,j}\), thus have labels meeting \(B_k\). \Cref{eq:finite-ball-dressing} gives \(\norm{Y_{k,s}-I}\le C\Jan|B_k|/\nu\) and \(\norm{V_{k,s}}\le C\Jan|B_k|(2/3)^s\). It remains to bound \(W_{k,s}\).

Consider the same recursion with input \(zH\), holding \(N_k\), \(\nu\), and \(s\) fixed, and write \(A_j(z)=\sum_{p\ge1}z^p A_{j,p}\). Since \(\norm{zH}_{\kappa_0}\le2J<\Jan\) for \(|z|\le2\), the uniform analytic estimates apply on a neighborhood of this disc. Cauchy's coefficient estimate on \(|z|=2\), applied to \Cref{eq:uniform-A-bound}, gives
\begin{equation}
    \norm{A_{j,p}}_{\kappa_j}\le\frac{C\Jan}{\nu}(2/3)^j2^{-p}.
    \label{eq:parallel-generator-coefficients}
\end{equation}
The induction in \Cref{lem:connected-word-expansion} applies to \(A_j\) as well as to \(D_j\): the maps \(\mathcal P_{N_k,\omega}\) and \(\mathcal S_{N_k}\) do not enlarge the support of an operator, and nonzero commutators join overlapping input supports. Each degree-\(p\) contribution is therefore built from \(p\) input Pauli terms whose support union has diameter at most \(pr_0\).

A nonzero contribution to \(A_j\) must involve an input support meeting \(B_k\). If all its input supports miss \(B_k\), every intermediate operator commutes with \(N_k\); the projection \(\mathrm{Id}-\mathcal P_{N_k,0}\) then gives zero, so it cannot produce a generator. To contribute to \([A_j,N_{{\rm rem},k}]\), the same union of supports must also meet \(\bigcup_{k'\ne k}B_R(k')\), where \(N_{{\rm rem},k}\) is supported. The distance between these two sets is at least \(L_{\mathcal D}-2R\ge(P+1)r_0\). A union of diameter at most \(pr_0\) cannot meet both when \(p\le P\). Thus \([A_{j,p},N_{{\rm rem},k}]=0\) for \(p\le P\), and only the Taylor coefficients of degree greater than \(P\) remain.

To bound \(\sum_{p>P}\|[A_{j,p},N_{{\rm rem},k}]\|\), let \(q_{\rm lab}:=\max_a|X_a|\le(2r_0+1)^d\), where \(X_a\) are the fixed component labels of the Pauli terms in \(H(\lambda)\). A label at degree \(p\) is a union of at most \(p\) such labels and has at most \(pq_{\rm lab}\) sites. Write \(N_{{\rm rem},k}=\sum_x(N_{{\rm rem},k})_x\), with \(\norm{(N_{{\rm rem},k})_x}\le2\). For a component \((A_{j,p})_X\), only sites \(x\in X\) contribute to its commutator with this charge. The elementary commutator inequality therefore gives
\begin{equation}
    \norm{[(A_{j,p})_X,N_{{\rm rem},k}]} \le 4|X|\norm{(A_{j,p})_X} \le4pq_{\rm lab}\norm{(A_{j,p})_X}.
\end{equation}

This estimate counts only sites of \(X\), independently of the
total number of other controlled balls. Summing over the labels
meeting \(B_k\) and using \Cref{eq:finite-ball-global-norm} yields
\begin{align}
    \norm{[A_j(1),N_{{\rm rem},k}]} &\le4q_{\rm lab}\sum_{p>P}p
    \sum_{X:X\cap B_k\ne\varnothing}\norm{(A_{j,p})_X}\\
    &\le4q_{\rm lab}|B_k|\sum_{p>P}p\norm{A_{j,p}}_{\kappa_j}\\
    &\le\frac{C\Jan|B_k|}{\nu}(2/3)^j(P+1)2^{-P}.
    \label{eq:parallel-remote-commutator}
\end{align}
The last line uses \Cref{eq:parallel-generator-coefficients} and
\(\sum_{p>P}p2^{-p}=(P+2)2^{-P}\le2(P+1)2^{-P}\).

At the physical input \(z=1\), \(H\) is Hermitian and
\(A_j(1)^\dagger=-A_j(1)\); abbreviate \(A_j(1)\) by \(A_j\).
For each step,
\begin{equation}
    e^{A_j}N_{{\rm rem},k}e^{-A_j}-N_{{\rm rem},k}=\int_0^1 e^{uA_j}[A_j,N_{{\rm rem},k}]e^{-uA_j}\,du.
\end{equation}

Unitary invariance bounds the norm of this difference by \(\norm{[A_j,N_{{\rm rem},k}]}\). Telescoping the product \(Y_{k,s}=e^{A_{s-1}}\cdots e^{A_0}\) expresses \(Y_{k,s}N_{{\rm rem},k}Y_{k,s}^\dagger-N_{{\rm rem},k}\) as a sum of these single-step differences, each conjugated by the subsequent unitary factors. More precisely,
\begin{equation}
    Y_{k,s}N_{{\rm rem},k}Y_{k,s}^{\dagger}-N_{{\rm rem},k}=\sum_{j=0}^{s-1}
    \bigl(e^{A_{s-1}}\cdots e^{A_{j+1}}\bigr)
    \bigl(e^{A_j}N_{{\rm rem},k}e^{-A_j}-N_{{\rm rem},k}\bigr)
    \bigl(e^{-A_{j+1}}\cdots e^{-A_{s-1}}\bigr).    
\end{equation}

Hence
\begin{equation}
    \norm{W_{k,s}}
    \le\nu\sum_{j=0}^{s-1}\norm{[A_j,N_{{\rm rem},k}]}
    \le C\Jan|B_k|(P+1)2^{-P},
    \label{eq:parallel-remote-field-error}
\end{equation}
where \(\sum_{j<s}(2/3)^j\le3\). The factor \(\nu\) in \(W_{k,s}\) cancels the factor \(\nu^{-1}\) in \Cref{eq:parallel-remote-commutator}. If \(|\mathcal D|=1\), then \(N_{{\rm rem},k}=W_{k,s}=0\).

The Hamiltonians and remainders in \Cref{eq:parallel-stopped-dynamics-app} are Hermitian, and \(Y_{k,s}\) is unitary. Apply the propagator argument of \Cref{thm:ADHH-dynamics} to that identity, with remainder \(V_{k,s}+W_{k,s}\). The two conjugating unitaries contribute \(2\norm{Y_{k,s}-I}\), and Duhamel's formula bounds the change of generator by \(|t|(\norm{V_{k,s}}+\norm{W_{k,s}})\). Thus
\begin{equation}
    \norm{e^{-iG_{\mathcal D,\ell}^{\mathcal B}t}-e^{-i\mathcal K_kt}}
    \le2\norm{Y_{k,s}-I}
    +|t|\left(\norm{V_{k,s}}+\norm{W_{k,s}}\right).
\end{equation}
For unitary \(U,U_0\), \(\norm{U^\dagger OU-U_0^\dagger OU_0} \le2\norm O\norm{U-U_0}\). Substitution of the three bounds above into this inequality proves \Cref{eq:parallel-stopped-dynamics}.
\end{proof}

\begin{theorem}[\Cref{thm:parallel-Ramsey}, Single-frequency approximation for multiple defects]
Fix integers \(P\ge1\) and \(K\ge0\). Assume $\nu\ge\max\{\bar\nu_0,54\pi\Jan/\kappa_0\}$, $1\le n_{\rm par}\le n_*(\nu)$,  $L_{\mathcal D}\ge2R+(P+1)r_0$, and $R\ge(K+1)r_0$. For every \(k\in\mathcal D\) and every real \(t\), the signal in \Cref{eq:parallel-signal} satisfies
\begin{equation}
 \begin{aligned}
 \abs{S_{\mathcal D,\ell;k}(t)
 -e^{i[\nu+2E_{\mathcal D,\ell;k}^{(<P_{\mathcal D})}]t}}
 \le \frac{C_1\Jan V_R}{\nu}
 +C_2\Jan V_R|t|\left[(2/3)^{n_{\rm par}}+(P+1)2^{-P}\right]
 +C_3\Jan|t|2^{-K}.
 \end{aligned}
\end{equation}
The same right-hand side bounds \(\abs{S_{\mathcal D,\ell;k}(t) -e^{i[\nu+2F_{\rho,R}^{\rm par}(\lambda)]t}}\), where \(F_{\rho,R}^{\rm par}\) is the patch function defined below in \Cref{eq:parallel-patch-map}, with \(\rho=(k,\ell)\) and matching Pauli axes on \(B_R(k)\). All constants are independent of \(n,|\mathcal D|,n_{\rm par},L_{\mathcal D},\nu\).
\end{theorem}

\begin{proof}[Proof of \Cref{thm:parallel-Ramsey}]
Fix \(k\in\mathcal D\), set \(s=n_{\rm par}\), put \(C=\Lambda\setminus B_k\), and write \(\rho=(k,\ell)\). We first analyze the effective signal \(S_{{\rm eff},k}(t)\) generated by \(\mathcal K_k\) and show that it is close to \(e^{i[\nu+2F_{\rho,R}^{\rm par}(\lambda)]t}\). We then apply \Cref{lem:parallel-stopped-dynamics} to compare the physical signal with \(S_{{\rm eff},k}(t)\), and combine these two estimates. Finally, we bound the additional error from replacing \(F_{\rho,R}^{\rm par}(\lambda)\) by \(E_{\mathcal D,\ell;k}^{(<P_{\mathcal D})}\) in the oscillation.

\medskip
\noindent
\textbf{Approximating \(S_{{\rm eff},k}(t)\) with \(e^{i[\nu+2F_{\rho,R}^{\rm par}(\lambda)]t}\).}
The state in \Cref{eq:parallel-preparation} has density matrix \(\ket{\phi_{B_k,+}}\bra{\phi_{B_k,+}}\otimes\sigma_C\), where \(\ket{\phi_{B_k,+}}=(\ket{\phi_{B_k,0}}+ \ket{\phi_{B_k,1}})/\sqrt2\) is defined as in \Cref{eq:ball-Ramsey-states}, using the axes of experiment \((\mathcal D,\ell)\). The matrix \(\sigma_C\) describes all remaining qubits, including the other defects in their superposition states. It is the same for the two values of the measured defect. The eigenspaces of \(N_k\) with eigenvalues \(b=0,1\) are \(\operatorname{span}\{\ket{\phi_{B_k,b}}\}\otimes\mathcal H_C\). Since \([D_{k,s},N_k]=[N_{{\rm rem},k},N_k]=0\), \(\mathcal K_k=\nu N_k+D_{k,s}+\nu N_{{\rm rem},k}\) (defined in \Cref{lem:parallel-stopped-dynamics}) preserves each of these eigenspaces.

For a full-system operator \(W\), let \(\mathcal M_b(W):=(\bra{\phi_{B_k,b}}\otimes I_C) W(\ket{\phi_{B_k,b}}\otimes I_C)\), as in \Cref{eq:ball-partial-matrix-element}. Define the operators on the exterior by
\begin{equation}
    K_b:=\mathcal M_b(D_{k,s})+\nu N_{{\rm rem},k},
    \qquad b=0,1.
\end{equation}
Here \(N_{{\rm rem},k}\) is viewed as an operator on \(C\). The restriction of \(\mathcal K_k\) to the eigenspace with eigenvalue \(b\) acts as \(\nu bI_C+K_b\). In particular, by the same argument as in \Cref{eq:ball-sector-generator,eq:ball-sector-evolution}, for every \(\ket\psi\in\mathcal H_C\),
\begin{equation}
    e^{-i\mathcal K_kt}(\ket{\phi_{B_k,b}}\otimes\ket\psi)=e^{-i\nu bt}\ket{\phi_{B_k,b}}\otimes e^{-iK_bt}\ket\psi.
\end{equation}
 
The same term \(\nu N_{{\rm rem},k}\) occurs in both $K_0$ and $K_1$, so \(K_1-K_0=\mathcal M_1(D_{k,s})-\mathcal M_0(D_{k,s})\). We next show that \(K_1-K_0\) is close to the scalar operator \(2F_{\rho,R}^{\rm par}(\lambda)I_C\). To do so, we compare two applications of the ADHH recursion, both stopped after \(s=n_{\rm par}\) steps.

The first recursion has input \(H(\lambda)\) and charge \(N_k\), and produces \(D_{k,s}\). This is the recursion used above to compare with the physical dynamics. It acts on the full system and includes interactions across the boundary of \(B_k\) and outside of it. Its output determines \(K_1-K_0\) through \(\mathcal M_1(D_{k,s})-\mathcal M_0(D_{k,s})\).

The second recursion defines the patch quantity \(F_{\rho,R}^{\rm par}(\lambda)\) in \Cref{eq:parallel-patch-map}. Its input is \(H_{B_k}(\lambda)\), which contains only the Pauli terms supported inside \(B_k\), and its charge is \(N_{\rho,B_k}\), which agrees with \(N_k\) on the patch. Its output is \(D_{\rho,R}^{\rm par}(\lambda)\). Viewing this output as an operator on the full system by extending it with \(I_C\), the definition of the patch function gives
\begin{equation}
    \mathcal M_1(D_{\rho,R}^{\rm par}(\lambda))-\mathcal M_0(D_{\rho,R}^{\rm par}(\lambda))=2F_{\rho,R}^{\rm par}(\lambda)I_C.
\end{equation}
 
Thus the desired estimate compares the same partial matrix-element difference for the outputs of these two recursions. We compare their Taylor coefficients as in the proof of \Cref{thm:Ramsey-approximation}: multiply both input Hamiltonians by a scalar \(z\), expand the outputs in \(z\), and keep \(s=n_{\rm par}\) fixed. The states \(\ket{\phi_{B_k,0}}\) and \(\ket{\phi_{B_k,1}}\) differ only at \(k\). Consequently, any contribution acting trivially on \(k\) vanishes under \(\mathcal M_1-\mathcal M_0\). A surviving degree-\(p\) contribution must therefore have an input support union containing \(k\).

By \Cref{lem:connected-word-expansion}, the overlap graph of these input supports is connected. Since each input support has diameter at most \(r_0\), their union has diameter at most \(pr_0\). As this union contains \(k\), it lies inside \(B_{pr_0}(k)\).

For \(p\le K\), the assumption \(R\ge(K+1)r_0\) places this entire union inside \(B_k\). All input terms involved in such a contribution are therefore present in both recursions. Moreover, the charges agree on \(B_k\), so the maps \(\mathcal P_{N,\omega}\) and \(\mathcal S_N\) agree on every intermediate operator in the contribution. The two recursions thus give identical contributions through degree \(K\) after applying \(\mathcal M_1-\mathcal M_0\).

Only Taylor degrees greater than \(K\) remain in the difference. Applying the same Cauchy estimates and geometric-tail summation used to obtain \Cref{eq:ball-conditional-difference} gives
\begin{equation}
    \norm{K_1-K_0-2F_{\rho,R}^{\rm par}(\lambda)I_C}
    \le C\Jan2^{-K}.
    \label{eq:parallel-exterior-difference}
\end{equation}
We note that this bound holds in operator norm on \(\mathcal H_C\). 

Let \(S_{{\rm eff},k}(t)\) denote the expectation of \(\tau_{k,+}^{\mathcal D,\ell}\) in this initial state under \(\mathcal K_k\). On the two eigenspaces under consideration, \(\tau_{k,+}^{\mathcal D,\ell}\) maps \(\ket{\phi_{B_k,0}}\) to \(2\ket{\phi_{B_k,1}}\) and annihilates \(\ket{\phi_{B_k,1}}\). Inserting the preceding evolution formula for the two components of \(\ket{\phi_{B_k,+}}\), and then averaging over \(\sigma_C\), gives
\begin{equation}
    S_{{\rm eff},k}(t)=e^{i\nu t}\Tr_C\!\left(\sigma_C e^{iK_1t}e^{-iK_0t}\right).
    \label{eq:parallel-effective-signal}
\end{equation}
Set \(\delta:=2F_{\rho,R}^{\rm par}(\lambda)\) and \(\Delta:=K_1-K_0-\delta I_C\). Since \(K_1=K_0+\delta I_C+\Delta\), Duhamel's formula gives
\begin{equation}
    \norm{e^{iK_1t}e^{-iK_0t}-e^{i\delta t}I_C}
    =\norm{e^{iK_1t}-e^{i(K_0+\delta I_C)t}}\le |t|\norm\Delta.
\end{equation}

Using \(\abs{\Tr_C(\sigma_C A)}\le\norm A\) and \Cref{eq:parallel-exterior-difference}, we conclude that
\begin{equation}
    \label{eq:from_effective_to_single_frequency}
    \abs{S_{{\rm eff},k}(t)-e^{i[\nu+2F_{\rho,R}^{\rm par}(\lambda)]t}} \le C\Jan|t|2^{-K}.
\end{equation}
The common exterior field cancels in \(K_1-K_0\), but it remains in both propagators in \Cref{eq:parallel-effective-signal}. The Duhamel estimate uses $K_1-K_0-\delta I_C,$ and does not require \(N_{{\rm rem},k}\) to commute with \(\mathcal M_b(D_{k,s})\).

\medskip
\noindent \textbf{Approximating \(S_{\mathcal D,\ell;k}(t)\) with \(S_{{\rm eff},k}(t)\).}
Apply \Cref{lem:parallel-stopped-dynamics} with \(O=\tau_{k,+}^{\mathcal D,\ell}\), whose norm is two, and take the expectation in the initial density matrix. This bounds \(\abs{S_{\mathcal D,\ell;k}(t)-S_{{\rm eff},k}(t)}\) as
\begin{equation}
    \label{eq:from_physical_to_effective_app}
    \abs{S_{\mathcal D,\ell;k}(t)-S_{{\rm eff},k}(t)}\leq C\Jan|B_k|\left\{\nu^{-1}+|t|\left[(2/3)^s+(P+1)2^{-P}\right]\right\}.
\end{equation}
Combining the two bounds \Cref{eq:from_effective_to_single_frequency,eq:from_physical_to_effective_app} and using \(|B_k|\le V_R\)
proves the stated approximation 
\begin{equation}
    \abs{S_{\mathcal D,\ell;k}(t) -e^{i[\nu+2F_{\rho,R}^{\rm par}(\lambda)]t}} \le \frac{C_1\Jan V_R}{\nu} +C_2\Jan V_R|t|\left[(2/3)^{n_{\rm par}}+(P+1)2^{-P}\right] +C_3\Jan|t|2^{-K}.
\end{equation}

\medskip
\noindent
\textbf{Replacing \(F_{\rho,R}^{\rm par}(\lambda)\) with \(E_{\mathcal D,\ell;k}^{(<P_{\mathcal D})}\).}
Finally, \Cref{prop:parallel-patch}, whose proof appears later in this appendix, gives \(\abs{E_{\mathcal D,\ell;k}^{(<P_{\mathcal D})} -F_{\rho,R}^{\rm par}(\lambda)}\le C\Jan2^{-K}\). That proposition compares the recursion with charge \(N_{\mathcal D,\ell}\) to the patch recursion. Their charges act identically on operators supported in \(B_k\), by the same reason as \Cref{eq:parallel-local-charge-agreement}. Consequently, the coefficients contributing to the transition at \(k\) agree through degree \(K\), and their tails obey the analytic bounds. This argument does not use the present signal estimate. Since both energies are real,
\begin{equation}
    \abs{e^{i[\nu+2F_{\rho,R}^{\rm par}(\lambda)]t} -e^{i[\nu+2E_{\mathcal D,\ell;k}^{(<P_{\mathcal D})}]t}}\le2|t|\abs{F_{\rho,R}^{\rm par}(\lambda) -E_{\mathcal D,\ell;k}^{(<P_{\mathcal D})}} \le C\Jan|t|2^{-K}.
\end{equation}
Adding this error to the approximation using the patch function proves \Cref{eq:parallel-Ramsey-error}.
\end{proof}

\begin{theorem}[\Cref{thm:parallel-estimate-half-gaps}, Parallel estimation of  transition energies]
Fix a nonempty batch \(\mathcal D\), a repetition \(\ell\), a sufficiently small accuracy \(\eta>0\), and a failure probability \(\delta_{\rm par}\in(0,1)\). Assume that the actual field strength $\nu$ obeys  the lower bound
\begin{equation}
    \nu\ge\nu_{\rm par} =O\!\left(J\max\left\{[\log(eJ/\eta)]^d, (d+1)\log(eJ/\eta) [1+\log((d+1)\log(eJ/\eta))]^3\right\}\right).
\end{equation}
Then, with \(K=P=n_{\rm par}\) and \(R=(K+1)r_0\) chosen as in \Cref{eq:parallel-stopping-order}, there exists a sufficient separation threshold \(P_{\rm sep}=\bigO((d+1)\log(1/\eta))\), corresponding to a minimum distance of order \(\bigO(r_0(d+1)\log(1/\eta))\), such that \(L_{\mathcal D}\ge r_0P_{\rm sep}\) allows simultaneous estimation of all {effective half transition energies} \(E_{\mathcal D,\ell;k}^{(<P_{\mathcal D})}\), \(k\in\mathcal D\), to precision \(\eta\) with probability at least {\(1-\delta_{\rm par}\)}, using \(\bigO(\log^2(1/\eta)\log(2|\mathcal D|/\delta_{\rm par}))\) independent non-adaptive global experiments, \(\bigO(\log(2|\mathcal D|/\delta_{\rm par})/\eta)\) total evolution time.
\end{theorem}

\begin{proof}
Fix \(\mathcal D\) and \(\ell\). We recover each effective half transition energy by estimating the oscillation frequency of the corresponding parallel Ramsey signal. To express this relation, we write the frequency as $\omega_{k,\ell} := \nu+2E_{\mathcal D,\ell;k}^{(<P_{\mathcal D})}$, where $k\in\mathcal D$. Since \(\nu\) is known, frequency accuracy \(2\eta\) will give half-transition-energy accuracy \(\eta\).

We first identify a common prior interval for these frequencies. The two product states defining \(E_{\mathcal D,\ell;k}^{(<P_{\mathcal D})}\) differ only at site \(k\). Hence, \Cref{eq:ell-local-bound-general-main,eq:parallel-normal-form-norms} give $\abs{E_{\mathcal D,\ell;k}^{(<P_{\mathcal D})}} \le C_E\Jan$, where \(C_E\) is a uniform constant. Thus every frequency belongs to \([\nu-\Phi,\nu+\Phi]\), with the known half-width \(\Phi:=2C_E\Jan=\bigO(1)\). This supplies the prior interval required by \Cref{lem:robust-frequency}.

We now verify that the parameter choices in \Cref{eq:parallel-stopping-order} keep the deterministic signal error below the tolerance required by \Cref{lem:robust-frequency} for all evolution times used by the estimator. The estimator uses times at most \(T_{\max}=C_T/\eta\). We bound the deterministic signal error in \Cref{eq:parallel-Ramsey-error} uniformly over this entire interval. Choose \(n_{\rm par},K,P,R\) as in \Cref{eq:parallel-stopping-order}. For \(\ell_\eta=\log(eJ/\eta)\), these choices give \(n_{\rm par}=\ceil{C_{\rm par}(d+1)\ell_\eta}\) and \(V_R=O(\ell_\eta^d)\), with fixed locality parameters. The field condition \Cref{eq:parallel-field-condition} ensures \(n_{\rm par}\le n_*(\nu)\). Its volume term also controls the time-independent contribution to the signal error: since \(\nu\ge C_{\rm vol}\Jan V_R\), taking \(C_{\rm vol}\) sufficiently large gives
\begin{equation}
    \frac{C_1\Jan V_R}{\nu}\le\frac1{12\sqrt8}.
    \label{eq:parallel-v-bound}
\end{equation}

Having bounded the time-independent contribution, we next control the term \(C_2\Jan V_R|t|(2/3)^{n_{\rm par}}\) in \Cref{eq:parallel-Ramsey-error}, which arises from \(V_{k,n_{\rm par}}\). Its maximum over the estimator's time interval occurs at \(t=T_{\max}\). The resulting upper bound is a constant times \(\ell_\eta^d \exp\{[1-C_{\rm par}(d+1)\log(3/2)]\ell_\eta\}\). The positive term in the exponent comes from \(J/\eta=e^{\ell_\eta-1}\), and the negative term comes from \((2/3)^{n_{\rm par}}\). A sufficiently large fixed \(C_{\rm par}\) makes the exponential decay dominate the polynomial factor, giving
\begin{equation}
    C_2\Jan V_RT_{\max}(2/3)^{n_{\rm par}}\le\frac1{12\sqrt8}.
    \label{eq:parallel-stopped-distortion-small}
\end{equation}

It remains to bound the terms involving \((P+1)2^{-P}\) and \(2^{-K}\) in \Cref{eq:parallel-Ramsey-error}. The former arises from \(W_{k,n_{\rm par}}\), the change in the fields on the other controlled balls under conjugation by \(Y_{k,n_{\rm par}}\). To ensure the separation required to bound this term, set
\begin{equation}
    P_{\rm sep}:=\ceil{2R/r_0}+P+1
    \label{eq:parallel-Psep-estimation}
\end{equation}
and require
\begin{equation}
    L_{\mathcal D}\ge r_0P_{\rm sep}.
    \label{eq:parallel-separation-estimation}
\end{equation}
Then the distance between any two controlled balls is at least \(L_{\mathcal D}-2R\ge(P+1)r_0\), which is the separation hypothesis of \Cref{thm:parallel-Ramsey}. The remaining two error terms at \(T_{\max}\) are bounded by a constant times \(\ell_\eta^{d+1} \exp\{[1-C_{\rm par}(d+1)\log2]\ell_\eta\}\): here the additional polynomial factor allows for \(P+1\), and \(P=K=n_{\rm par}\). Increasing the fixed \(C_{\rm par}\) if necessary therefore gives
\begin{equation}
    C_2\Jan V_RT_{\max}(P+1)2^{-P}+C_3\Jan T_{\max}2^{-K}\le\frac1{12\sqrt8}.
    \label{eq:parallel-tail-distortion-small}
\end{equation}

The bounds in \Cref{eq:parallel-v-bound,eq:parallel-stopped-distortion-small,eq:parallel-tail-distortion-small} hold simultaneously for sufficiently small \(\eta\). Their sum bounds the difference between the exact expected signal and the ideal single-frequency signal.  The required signal-error tolerance is a fixed constant, while the transition-energy accuracy \(\eta\) is achieved by using evolution times up to \(T_{\max}=C_T/\eta\) in the frequency estimator. Sampling error from finitely many measurements is accounted for separately below.

Combining the three error budgets with \Cref{thm:parallel-Ramsey}, we obtain, for every \(k\in\mathcal D\) and \(0\le t\le T_{\max}\),
\begin{equation}
    \abs{ S_{\mathcal D,\ell;k}(t)-e^{i\omega_{k,\ell}t}} \le\frac{1}{4\sqrt8} <\frac{1}{2\sqrt8}.
    \label{eq:parallel-signal-frequency-condition}
\end{equation}
This deterministic bound leaves enough room for the sampling error in the inputs to \Cref{lem:robust-frequency}.

To construct those inputs, let \(\widehat c_{k,\ell}(t)\) and \(\widehat s_{k,\ell}(t)\) be sample means obtained by measuring \(\tau_{k,x}^{\mathcal D,\ell}\) and \(\tau_{k,y}^{\mathcal D,\ell}\), respectively. Their expectations are the real and imaginary parts of \(S_{\mathcal D,\ell;k}(t)\). Each measurement outcome lies in \(\{-1,1\}\), so a fixed number of independent repetitions makes either sample mean accurate to \(1/(2\sqrt8)\) with probability at least \(2/3\). Disjoint experimental runs are used for the two means and for distinct input samples, giving the independence required by the frequency-estimation lemma.

If \(\widehat c_{k,\ell}(t)\) differs from \(\operatorname{Re}S_{\mathcal D,\ell;k}(t)\) by at most \(1/(2\sqrt8)\), the triangle inequality and \Cref{eq:parallel-signal-frequency-condition} give the cosine error bound
\begin{equation}
    \abs{\widehat c_{k,\ell}(t)-\cos(\omega_{k,\ell}t)}\le\frac{1}{2\sqrt8}+\frac{1}{4\sqrt8}=\frac{3}{4\sqrt8}<\frac{1}{\sqrt8}.
\end{equation}
A similar sine error bound follows whenever \(\widehat s_{k,\ell}(t)\) differs from \(\operatorname{Im}S_{\mathcal D,\ell;k}(t)\) by at most \(1/(2\sqrt8)\). Each inequality holds with probability at least \(2/3\). These are the input-accuracy conditions of \Cref{lem:robust-frequency}.

A single invocation of that lemma with root-mean-square frequency error \(\eta\) returns an estimate \(\widetilde\omega_{k,\ell}\) satisfying \(\mathbb E|\widetilde\omega_{k,\ell}-\omega_{k,\ell}|^2\le\eta^2\). To convert this into a constant-success accuracy guarantee, Markov's inequality applied to the squared error gives $\Pr\!\left( \abs{\widetilde\omega_{k,\ell}-\omega_{k,\ell}}>2\eta \right)\le 1/4$.

To make the estimates simultaneous over the batch, we require failure probability at most \(\delta_{\rm par}/|\mathcal D|\) at each defect. For this purpose, repeat the frequency estimator independently an odd number $ N_{\rm rep} =\bigO( \log({2|\mathcal D|}/{\delta_{\rm par}})) $ of times, with a sufficiently large constant, and let \(\widehat\omega_{k,\ell}\) be the median of the resulting estimates. Each repetition succeeds with probability at least \(3/4\). A concentration bound on the number of unsuccessful repetitions therefore gives, with probability at least \(1-\delta_{\rm par}/|\mathcal D|\),
\begin{equation}
    \abs{\widehat\omega_{k,\ell}-\omega_{k,\ell}}\le2\eta.
    \label{eq:parallel-frequency-accuracy}
\end{equation}

To recover the desired energy, subtract the known field contribution and divide by two: $\widehat E_{\mathcal D,\ell;k}:= {(\widehat\omega_{k,\ell}-\nu)}/{2}$. Therefore, for each $k$, with probability at least $1-\delta_{\rm par}/|\mathcal{D}|$, we have
\begin{equation}
    \abs{ \widehat E_{\mathcal D,\ell;k} - E_{\mathcal D,\ell;k}^{(<P_{\mathcal D})}} = \frac12\abs{\widehat\omega_{k,\ell}-\omega_{k,\ell}} \le\eta.
    \label{eq:parallel-gap-accuracy}
\end{equation}
A union bound makes this guarantee simultaneous for all \(k\in\mathcal D\), with probability at least \(1-\delta_{\rm par}\). Independence between outcomes at different defects is not required since we are using the union bound.

We now count the number experiments needed. All defects have the same prior interval and target accuracy, so they can share a common non-adaptive sampling schedule. The observables \(\tau_{k,x}^{\mathcal D,\ell}\), \(k\in\mathcal D\), have disjoint supports and can be measured simultaneously; the same holds for the \(\tau_{k,y}^{\mathcal D,\ell}\). Thus each global run in either measurement setting supplies an outcome at every defect.

One invocation of \Cref{lem:robust-frequency} requires \(\bigO(\log^2(\Phi/\eta))=\bigO(\log^2(1/\eta))\) input samples and total evolution time \(\bigO(1/\eta)\). The fixed repetitions used to form each sample mean and the two measurement settings contribute only constant factors. Multiplying by \(N_{\rm rep}\), the number of global experiments is $\bigO(\log^2(1/\eta)\log({2|\mathcal D|}/{\delta_{\rm par}}))$, and their total evolution time is $\bigO(\log({2|\mathcal D|}/{\delta_{\rm par}})/\eta).$ All runs can be scheduled in advance and performed on independently prepared states. There is no multiplicative factor \(|\mathcal D|\), because the same global runs provide data for every defect.

Having established the estimation guarantees and resource bounds, we now verify that the sufficient separation in \Cref{eq:parallel-separation-estimation}, with \(P_{\rm sep}\) defined in \Cref{eq:parallel-Psep-estimation}, has the logarithmic scaling asserted in \Cref{thm:parallel-estimate-half-gaps}. Substituting \(R=(n_{\rm par}+1)r_0\) and \(P=n_{\rm par}\) into \Cref{eq:parallel-Psep-estimation} gives \(P_{\rm sep}=3(n_{\rm par}+1)\). Since \(n_{\rm par}=O((d+1)\log(eJ/\eta))\), the sufficient distance \(r_0P_{\rm sep}\) has the claimed logarithmic dependence on \(1/\eta\). These choices depend on the target accuracy and fixed locality parameters. A larger actual field strength does not change \(n_{\rm par}\), and a larger batch does not change the required distance between its defects. This completes the proof.
\end{proof}

\begin{proposition}[\Cref{prop:parallel-patch}]

Let $k\in\mathcal D$, let $\rho=(k,\ell)$, and let \(s_j^\rho=s_j^{\mathcal D,\ell}\), \(\beta_j^\rho=\beta_j^{\mathcal D,\ell}\), for every \(j\in B_R(k)\). We assume $R\ge(K+1)r_0$, and $R<L_{\mathcal D}$. Then $P_{\mathcal D}\ge K+2$ and
\begin{equation}
    \abs{ E_{\mathcal D,\ell;k}^{(<P_{\mathcal D})} {-F_{\rho,R}^{\rm par}(\lambda)}} \le C_{\rm P}^{\rm par}\Jan2^{-K},
    \label{eq:parallel-patch-error-appendix}
\end{equation}
for a universal constant $C_{\rm P}^{\rm par}$, uniformly in $n$, $|\mathcal D|$, and the defect locations.
\end{proposition}

\begin{proof}
Fix $k\in\mathcal D$ and write $B:=B_R(k)$. The assumptions imply $\frac{L_{\mathcal D}}{r_0}>\frac{R}{r_0}\ge K+1$, and therefore $P_{\mathcal D}\ge K+2$. Moreover, $B$ contains no defect other than $k$.

For $1\le p<P_{\mathcal D}$, define $e_{\mathcal D,\ell;k}^{(p)} := \left( \bra{\Phi_{\{k\}}^{\mathcal D,\ell}} \widehat D_{\mathcal D,\ell}^{(p)} \ket{\Phi_{\{k\}}^{\mathcal D,\ell}} - \bra{\Phi_{\varnothing}^{\mathcal D,\ell}} \widehat D_{\mathcal D,\ell}^{(p)} \ket{\Phi_{\varnothing}^{\mathcal D,\ell}} \right)/2.$ Then
\begin{equation}
    E_{\mathcal D,\ell;k}^{(<P_{\mathcal D})} = \sum_{p=1}^{P_{\mathcal D}-1} e_{\mathcal D,\ell;k}^{(p)}.
    \label{eq:parallel-gap-expansion}
\end{equation}
By \Cref{eq:parallel-taylor}, the quantities \(e_{\mathcal D,\ell;k}^{(p)}\), \(1\le p<P_{\mathcal D}\), are the degree-\(p\) Taylor coefficients at \(z=0\) of the scalar function
\begin{equation}
\label{eq:parallel-full-scalar}
    g_{\mathcal D,\ell;k}^{\Lambda}(z) :=\frac12\left( \bra{\Phi_{\{k\}}^{\mathcal D,\ell}} \widehat D_{\mathcal D,\ell}(z) \ket{\Phi_{\{k\}}^{\mathcal D,\ell}} - \bra{\Phi_{\varnothing}^{\mathcal D,\ell}} \widehat D_{\mathcal D,\ell}(z) \ket{\Phi_{\varnothing}^{\mathcal D,\ell}} \right).
\end{equation}

For the patch map \(F_{\rho,R}^{\rm par}\) generated by early-stopping the ADHH recursion, as defined in \Cref{eq:parallel-patch-map}, introduce the scalar parameter \(z\) as in \Cref{eq:g-functions} and write ${g_{\rho,{\rm par}}^B(z) :=F_{\rho,R}^{\rm par}(z\lambda)} =\sum_{p\ge0}c_p^Bz^p.$ Since {\(F_{\rho,R}^{\rm par}(0)=0\)}, we have \(c_0^B=0\). Evaluating this expansion at \(z=1\) gives ${F_{\rho,R}^{\rm par}(\lambda)} =\sum_{p\ge1}c_p^B.$

We claim that
\begin{equation}
    e_{\mathcal D,\ell;k}^{(p)}=c_p^B, \quad 1\le p\le K.
    \label{eq:parallel-coefficient-matching}
\end{equation}
Since \(B\cap\mathcal D=\{k\}\), the restriction of \(N_{\mathcal D,\ell}\) to \(B\) is exactly \(N_{\rho,B}\) from \Cref{eq:patch-charge}. Moreover, the restrictions of \(\ket{\Phi_{\varnothing}^{\mathcal D,\ell}}\) and \(\ket{\Phi_{\{k\}}^{\mathcal D,\ell}}\) to \(B\) coincide with those of the single-defect states \(\ket{\Phi_0^\rho}\) and \(\ket{\Phi_1^\rho}\) defined in \Cref{eq:Phi0_Phi1}, respectively.

Thus \(g_{\rho,{\rm par}}^B(z)=F_{\rho,R}^{\rm par}(z\lambda)\) from \Cref{eq:parallel-patch-map} and \(g_{\mathcal D,\ell;k}^{\Lambda}(z)\) defined in \Cref{eq:parallel-full-scalar} play the roles of \(g_\rho^B(z)\) and \(g_\rho^\Lambda(z)\), respectively, from \Cref{eq:g-functions} in the single-defect case. Since \(c_p^B\) is the Taylor coefficient for $g_{\rho,{\rm par}}^B(z)$ and $e_{\mathcal D,\ell;k}^{(p)}$ is the Taylor coefficient for $g_{\mathcal D,\ell;k}^{\Lambda}(z)$, applying \Cref{lem:coefficient-matching} to the multiple-defect charge $N_{\mathcal{D},\ell}$, as justified by the \Cref{rem:applicability_to_parallel}, yields \Cref{eq:parallel-coefficient-matching}. Here the ADHH recursion stops at the same fixed order \(n_{\rm par}\le n_*(\nu)\) for both \(g_{\mathcal D,\ell;k}^{\Lambda}(z)\) and $g_{\rho,{\rm par}}^B(z)$. The proof of \Cref{lem:coefficient-matching} applies unchanged at any such common stopping order.

By \Cref{eq:parallel-taylor-coeff} and the same local estimate as in
\Cref{eq:ell-local-bound-general-main}, $\abs{e_{\mathcal D,\ell;k}^{(p)}}\le C_{\rm T}\Jan2^{-p}$. For the patch coefficients, {the same Cauchy estimate at step \(n_{\rm par}\) gives} $\abs{c_p^B}\le C_g\Jan2^{-p}$. Using \Cref{eq:parallel-coefficient-matching}, we conclude that
\begin{align}
    \abs{ E_{\mathcal D,\ell;k}^{(<P_{\mathcal D})} {-F_{\rho,R}^{\rm par}(\lambda)}} & 
    \le \sum_{p=K+1}^{P_{\mathcal D}-1} \abs{e_{\mathcal D,\ell;k}^{(p)}} + \sum_{p=K+1}^{\infty} \abs{c_p^B} \\ 
    & \le (C_{\rm T}+C_g)\Jan \sum_{p=K+1}^{\infty}2^{-p} = C_{\rm P}^{\rm par}\Jan2^{-K}.
\end{align}
This proves \Cref{eq:parallel-patch-error-appendix}.
\end{proof}

\section{Efficient classical evaluation and optimization}
\label{app:classical-evaluation}

In this section we will show that the fixed-preconditioned least-squares recovery described in \Cref{eq:iteration-map} is classically efficient.  For fixed locality data, its total arithmetic cost is
\begin{equation}
    T_{\rm cl} \le nm\,\operatorname{poly}({s},\ell_\eps) (J/\eps)^{C_{\rm alg}} +\bigO(M\ell_\eps^2), \quad \ell_\eps=\log({eJ}/{\eps}),
    \label{eq:classical-cost-conclusion}
\end{equation}
where \(C_{\rm alg}\) depends only on the dimension $d$ fixed locality parameters \(q,r_0\)  in \Cref{eq:locality-bounds-H} (it depends also on \(\zeta,D_{\rm ov}\) in \Cref{eq:zeta-Dov} but these are completely determined by $q,r_0,d$).  In particular, the post-processing is polynomial in \(n\), \(J/\eps\), and \(\log(1/\delta)\). We note that $J$ is always chosen to be a constant.  The algorithm never forms or evaluates the Hessian of the objective function. The second-derivative estimate is used only in the proof of \Cref{lem:Hessian-conditioning}.

A dense simulation of the entire patch would cost $2^{\bigO(|B_R(k)|)}=2^{\bigO(R^d)}$, which is not polynomial in \(1/\eps\) for general fixed \(d>1\).  Efficient post-processing uses the connected-cluster expansion of \Cref{lem:connected-word-expansion}, which underlies the proof of \Cref{thm:patch-locality}.

For an experiment indexed by \(\rho=(k,\ell)\), with the forward map $F_{\rho,R}(\mu)$ defined in \Cref{eq:patch-gap}, we define 
\begin{equation}
    g_{\rho,\mu}^{B}(z) := F_{\rho,R}(z\mu) =\sum_{p\ge0}c_{\rho,p}(\mu)z^p.
    \label{eq:classical-row-series}
\end{equation}
Here we omit the dependence on $R$ because it is fixed.
We also define the degree-\(L\) approximation
\begin{equation}
    F_\rho^{[L]}(\mu) :=\sum_{p=0}^{L}c_{\rho,p}(\mu), \quad \mathsf J_\rho^{[L]}(\mu) :=D F_\rho^{[L]}(\mu).
    \label{eq:classical-row-truncation}
\end{equation}

\begin{proposition}[Efficient forward map and Jacobian evaluation]
\label{prop:classical-row-evaluation}
There are constants \(C_{\rm ev},C_F,C_J'>0\), depending only on the dimension $d$ and fixed locality parameters $q,r_0$ and the ADHH recursion \Cref{eq:ADHH-recursion}, such that $F_\rho^{[L]}(\mu)$ and its partial derivatives $\partial_a F_\rho^{[L]}(\mu)$ for all $a$ such that $\supp(P_a)\subset B_R(k)$ (these are all the non-zero partial derivatives of $F_\rho^{[L]}(\mu)$) can be evaluated in
\begin{equation}
    \operatorname{poly}(L,{s},V_R)C_{\rm ev}^{L}
    \label{eq:classical-row-cost}
\end{equation}
arithmetic operations.  Uniformly for \(\rho\) and \(\mu\in\cU\),
\begin{align}
    \abs{F_\rho(\mu)-F_\rho^{[L]}(\mu)} & \le C_F\Jan2^{-L},
    \label{eq:classical-value-tail}\\
    \sum_a\abs{\partial_aF_\rho(\mu)-\partial_aF_\rho^{[L]}(\mu)}& \le C_J' {(2/3)^L}.
    \label{eq:classical-Jacobian-tail}
\end{align}
Consequently, to achieve approximation error at most \(\tau_F,\tau_J>0\) for $F_\rho(\mu)$ and $\nabla F_\rho(\mu)$ with $F_\rho^{[L]}(\mu)$ and $\nabla F_\rho^{[L]}(\mu)$ respectively, we can choose 
\begin{equation}
    L=\bigO(\log({eJ}/{\min\{\tau_F,J\tau_J\}})).
    \label{eq:classical-L-choice}
\end{equation}
\end{proposition}

\begin{proof}
Fix \(\rho=(k,\ell)\), and let \(\cA_\rho:=\{a:\supp(P_a)\subseteq B_R(k)\}\), as in \Cref{lem:row-derivatives}.  Define the interaction-overlap graph \(\mathcal G_{\rm ov}\) on \(\cA_\rho\) by joining distinct vertices \(a,b\) when \(\supp(P_a)\cap\supp(P_b)\ne\varnothing\).  By \Cref{eq:zeta-Dov}, its maximum degree is at most \(D_{\rm ov}-1\), and at most \(\zeta\) vertices correspond to interactions whose support contains the defect site \(k\).  We call a vertex set \(S\subseteq\cA_\rho\) \emph{anchored} if \(k\in\bigcup_{a\in S}\supp(P_a)\).

We enumerate the connected, anchored vertex sets \(S\) with \(|S|\le L\). For a graph of maximum degree \(D_{\rm ov}-1\), the number of connected sets of cardinality \(r\) containing a fixed vertex is at most \(C_{\rm conn}^{r}\), where \(C_{\rm conn}\) depends only on \(D_{\rm ov}\).  Since there are at most \(\zeta\) possible anchoring vertices, the number of sets considered at cardinality \(r\) is at most \(\zeta C_{\rm conn}^{r}\).  These sets can be generated by the standard recursive enumeration of connected vertex sets in a bounded-degree graph \cite{KomusiewiczSommer2021}.

For each such \(S\), use the patch charge \(N_{\rho,B}\) from \Cref{eq:patch-charge} and apply the \({s}\)-step ADHH recursion \Cref{eq:ADHH-recursion} to
\begin{equation}
    H_S(z):=z\sum_{a\in S}\mu_aP_a
    \label{eq:cluster-restricted-H}
\end{equation}
with all polynomial arithmetic performed in \(\mathbb C[z]/(z^{L+1})\), which discards all monomials of $z$ beyond degree $L$.  Let \(D_{\rho,S}^{(\le L)}(z)\) be the resulting degree-\(L\) truncation of \(D_{{s}}\), and define \(F_{\rho,S}^{(\le L)}(z):=\ell_\rho(D_{\rho,S}^{(\le L)}(z))\), where \(\ell_\rho\) is the transition functional in \Cref{eq:transition-functional}.  For every connected, anchored \(S\), following the standard subcluster-subtraction recursion for linked-cluster expansions~\cite[Sec.~1.3]{TangKhatamiRigol2013}, we define:
\begin{equation}
    W_{\rho,S}(z) :=F_{\rho,S}^{(\le L)}(z)-\sum_{\substack{T\subsetneq S\\T\text{ connected and anchored}}}W_{\rho,T}(z).
    \label{eq:cluster-Mobius-weight}
\end{equation}

By \Cref{lem:connected-word-expansion}, every degree-\(p\) term produced by the recursion is a linear combination of connected \(p\)-words in the sense of \Cref{def:p-word}.  Moreover, \Cref{eq:transition-functional-Xk} shows that applying \(\ell_\rho\) annihilates every operator whose support does not contain \(k\).  Thus, every nonzero Taylor monomial in \(F_\rho^{[L]}(\mu)\) has a connected, anchored set \(S\) of distinct interaction indices, with \(|S|\le p\le L\). For each \(F_{\rho,T}^{(\le L)}\), this monomial with interaction indices $S$ occurs if and only if \(S\subseteq T\).   Induction on \(|T|\) in \Cref{eq:cluster-Mobius-weight} shows that the coefficient of this monomial in \(W_{\rho,T}\) vanishes for every \(T\ne S\), whereas its coefficient in \(W_{\rho,S}\) equals its coefficient in the full patch expansion. Equivalently, \(W_{\rho,S}\) consists exactly of the monomials whose set of distinct interaction indices is \(S\). It follows that
\begin{equation}
    F_\rho^{[L]}(\mu)
    =\sum_{\substack{S\text{ connected, anchored}\\|S|\le L}}
    W_{\rho,S}(1).
    \label{eq:cluster-reconstruction}
\end{equation}
Repeated interaction indices in a connected word are included by polynomial multiplication in \(\mathbb C[z]/(z^{L+1})\), while the exponential conjugations in \Cref{eq:ADHH-recursion} generate all nested-commutator terms of degree at most \(L\).

Algorithmically, we process the connected, anchored sets \(S\) in increasing order of \(|S|\). For each \(S\), we initialize the ADHH recursion \Cref{eq:ADHH-recursion} with \(H_S(z)\) from \Cref{eq:cluster-restricted-H}, perform all operations modulo \(z^{L+1}\), and apply the transition functional \(\ell_\rho\) to the final charge-preserving interaction to obtain \(F_{\rho,S}^{(\le L)}(z)\). We then evaluate its degree-\(L\) representative at \(z=1\) and compute \(W_{\rho,S}(1)\) using \Cref{eq:cluster-Mobius-weight}; all required proper subcluster contributions have already been computed. Finally, summing these cluster contributions according to \Cref{eq:cluster-reconstruction} yields \(F_\rho^{[L]}(\mu)\). This allows us to only evaluate ADHH recursion on qubits involved in interaction terms in each $S$. Because each $S$ involves at most $L$ terms, this reduces the size of the subsystem we need to run ADHH recursion on from $\bigO(R^d)$ to at most $qL$.

It remains to bound the arithmetic cost.  By \Cref{eq:ball-prelim} and \Cref{eq:zeta-Dov}, the local index set \(\cA_\rho\) and the adjacency lists of \(\mathcal G_{\rm ov}\) can be constructed in \(\operatorname{poly}(V_R)\) operations.  A set \(S\) of cardinality \(r\) acts on at most \(qr\) qubits by \Cref{eq:locality-bounds-H}.  The charge projections and homological inverse preserve this union of supports by \Cref{lem:support-calculus}(i), so the restricted recursion can be evaluated on these at most \(qr\) qubits.  Dense matrix arithmetic for the degree-\(L\), \({s}\)-step recursion on this Hilbert space therefore requires \(\operatorname{poly}(L,{s})2^{\bigO(qr)} =\operatorname{poly}(L,{s})2^{\bigO(r)}\) operations, because \(q\) is fixed.  Computing \Cref{eq:cluster-Mobius-weight} examines at most \(2^r\) subsets.  Summing these costs over the at most \(\zeta C_{\rm conn}^r\) clusters of each cardinality \(r\le L\) proves \Cref{eq:classical-row-cost}, after enlarging \(C_{\rm ev}\). Finally, differentiating the preceding calculation with respect to the local coefficients \(\{\mu_a\}_{a\in\mathcal A_\rho}\) and applying the chain rule at each step yields all derivatives \(\partial_a F_\rho^{[L]}(\mu)\) with the same asymptotic arithmetic cost as evaluating \(F_\rho^{[L]}(\mu)\)~\cite{BaurStrassen1983}.

By \Cref{eq:classical-row-series}, the coefficients \(c_{\rho,p}(\mu)\) are the patch Taylor coefficients bounded in \Cref{eq:Taylor-coeff-bound}.  Hence \(\abs{c_{\rho,p}(\mu)}\le C_g\Jan2^{-p}\), and summing over \(p>L\) proves \Cref{eq:classical-value-tail}.

For the derivative tail, fix \(\mu\in\cU\) and \(\norm{\xi}_\infty\le1\), and consider \(F_\rho(z(\mu+w\xi))\), where \(w\) perturbs the Hamiltonian in the direction \(K_\xi\) defined in \Cref{eq:coefficient-direction}. The Taylor expansion remains centered at zero in the variable \(z\); the auxiliary variable \(w\) is used only to estimate the derivative of the remainder. We choose \(|z|\le3/2\) and \(|w|\le1/2\), so that \Cref{eq:coeff-to-local} gives
\begin{equation}
    \norm{z[H_B(\mu)+wK_\xi]}_{\kappa_0} \le |z|(2+|w|)J \le 15J/4<\Jan.
\end{equation}
Thus \Cref{thm:uniform-normal-form} and the transition-functional bound imply analyticity and the uniform bound \(\abs{F_\rho(z(\mu+w\xi))}\le C_g\Jan\) on this polydisc. Cauchy's coefficient estimate in \(z\), followed by summing the omitted terms, yields
\begin{equation}
    \sup_{|w|\le1/2} \abs{(F_\rho-F_\rho^{[L]})(\mu+w\xi)} \le C_g\Jan\sum_{p>L}(2/3)^p =2C_g\Jan(2/3)^L.
\end{equation}
Applying the first-derivative Cauchy estimate \Cref{eq:Cauchy-polydisc} in \(w\), on the circle of radius \(1/2\), therefore gives \(\abs{D(F_\rho-F_\rho^{[L]})(\mu)[K_\xi]} \le C_J'(2/3)^L\), with \(C_J'=4C_g\Jan\), uniformly for \(\norm{\xi}_\infty\le1\). The duality identity used in the proof of \Cref{eq:first-derivative-bound} then gives \Cref{eq:classical-Jacobian-tail}.  Requiring the right-hand sides of \Cref{eq:classical-value-tail,eq:classical-Jacobian-tail} to be at most \(\tau_F\) and \(\tau_J\), respectively, gives \Cref{eq:classical-L-choice}.
\end{proof}

Let $F^{[L]}(\mu):=\bigl(F_\rho^{[L]}(\mu)\bigr)_\rho$, $\mathsf J^{[L]}(\mu):=D F^{[L]}(\mu)$, and define the approximate gradient
\begin{equation}
    f^{[L]}(y,\mu) :=\frac1m\mathsf J^{[L]}(\mu)^{\mathsf T} \bigl(F^{[L]}(\mu)-y\bigr).
    \label{eq:classical-approx-gradient}
\end{equation}

\begin{lemma}[Gradient approximation errors]
\label{lem:classical-gradient-error}
We assume that for \(\mu\in\cU\) uniformly,
\begin{align}
    \norm{F^{[L]}(\mu)-F(\mu)}_\infty&\le\tau_F, \label{eq:classical-tauF}\\
    \sup_\rho\sum_a\abs{\mathsf J^{[L]}_{\rho a}(\mu)-\mathsf J_{\rho a}(\mu)} &\le\tau_J.
    \label{eq:classical-tauJ}
\end{align}
Then, uniformly on \(\cY\times\cU\),
\begin{equation}
    \norm{f^{[L]}(y,\mu)-f(y,\mu)}_\infty \le V_Rr_E\tau_J +\bigl(C_y+V_R\tau_J\bigr)\tau_F,
    \label{eq:classical-gradient-error-bound}
\end{equation}
where $V_R$ is the maximal volume of a radius-$R$ Manhattan distance ball, $r_E$ is defined in \Cref{eq:residual-bound}, and $C_y$ is defined in \Cref{eq:Cy}.
\end{lemma}

\begin{proof}
Set \(\mathsf E:=\mathsf J^{[L]}-\mathsf J\).  A fixed coefficient column can be nonzero in at most \(mV_R\) rows, so for every vector \(z\),
\begin{equation}
    \norm{\frac1m\mathsf E^{\mathsf T}z}_\infty \le V_R\tau_J\norm z_\infty.
    \label{eq:classical-E-transpose}
\end{equation}
Using \Cref{eq:Cy},
\begin{equation}
    \norm{\frac1m\bigl(\mathsf J^{[L]}\bigr)^{\mathsf T}z}_\infty
    \le\bigl(C_y+V_R\tau_J\bigr)\norm z_\infty.
    \label{eq:classical-JL-transpose}
\end{equation}
Now decompose
\begin{align}
    f^{[L]}-f
    ={}&\frac1m\mathsf E^{\mathsf T}(F-y)+\frac1m\bigl(\mathsf J^{[L]}\bigr)^{\mathsf T}
    \bigl(F^{[L]}-F\bigr).
\end{align}
The residual bound \Cref{eq:residual-bound} and the preceding two estimates
prove \Cref{eq:classical-gradient-error-bound}.
\end{proof}

Recall the fixed-preconditioner iteration \Cref{eq:iteration-map}. At each outer iteration, given the current coefficient vector \(\mu\), we evaluate \(F^{[L]}(\mu)\) and \(\mathsf J^{[L]}(\mu)\), form the approximate gradient \(f^{[L]}(y,\mu)\) defined in \Cref{eq:classical-approx-gradient}, approximately solve \(\Gamma d^{[L]}=f^{[L]}(y,\mu)\), and update \(\mu\leftarrow\mu-d^{[L]}\). It remains to show that this linear system can be solved efficiently. In other words we need to apply \(\Gamma^{-1}\) efficiently.  On the event that \Cref{eq:Gamma-close} holds, which happens with probability at least {\(1-\delta/2\) for \(m=\ceil{C_m{\log(4n/\delta)}}\)} by \Cref{thm:Walsh-conditioning}

\begin{equation}
    \Gamma_{aa}\ge (3/4)\cdot 3^{-q}, \quad \sum_{b\ne a}\abs{\Gamma_{ab}}\le (1/4)\cdot3^{-q}.
    \label{eq:classical-diagonal-dominance}
\end{equation}
Let \(\mathsf D_\Gamma=\operatorname{diag}(\Gamma)\).  Jacobi iteration for \(\Gamma x=b\),
\begin{equation}
    x^{(s+1)} =\mathsf D_\Gamma^{-1} \bigl[b-(\Gamma-\mathsf D_\Gamma)x^{(s)}\bigr],
    \label{eq:classical-Jacobi}
\end{equation}
has the uniform contraction bound
\begin{equation}
    \norm{\mathsf D_\Gamma^{-1} (\Gamma-\mathsf D_\Gamma)}_{\infty\to\infty} \le 1/3.
    \label{eq:classical-Jacobi-contraction}
\end{equation}
Moreover, \(\Gamma_{ab}=0\) whenever \(\supp(P_a)\cap\supp(P_b)=\varnothing\). Bounded interaction degree therefore ensures that the number of non-zero entries in $\Gamma$ is at most $\bigO(M)$.

As a result, each Jacobi step costs \(\bigO(M)\), and \(\bigO(\log(J/\tau_{\rm lin}))\) steps solve the system to \(\ell_\infty\)-error \(\tau_{\rm lin}\).

We will show that the error in solving this linear system only results in controllable error in the fixed point in the iteration. Let \(\widetilde s^{[L]}(y,\mu)\) denote the approximate solution of $\Gamma s=f^{[L]}(y,\mu)$ with solver error at most \(\tau_{\rm lin}\), and set
\begin{equation}
    \widetilde{\mathcal T}_y(\mu) :=\mu-\widetilde s^{[L]}(y,\mu).
    \label{eq:classical-implemented-map}
\end{equation}
By \Cref{lem:classical-gradient-error},
\begin{equation}
    \sup_{\mu\in\cU}
    \norm{\widetilde{\mathcal T}_y(\mu)-\mathcal T_y(\mu)}_\infty
    \le\tau_\Phi,
    \label{eq:classical-map-error}
\end{equation}
where
\begin{equation}
    \tau_\Phi :=C_\Gamma\left[V_Rr_E\tau_J +\bigl(C_y+V_R\tau_J\bigr)\tau_F\right] +\tau_{\rm lin},
    \label{eq:classical-tau-Phi}
\end{equation}
and \(C_\Gamma\) is defined in \Cref{eq:Gamma-inverse}.

We first choose an initializer that satisfies the initialization hypothesis of \Cref{thm:local-inverse}.  We work under the hypotheses of that theorem: {\Cref{eq:Gamma-close} holds}, \(\nu\ge\nu_{\rm opt}\), and \(r_y\) satisfies \Cref{eq:ry-condition}. Write \(e:=y-y_*\), assume \(0<\eps\le1\), and require
\begin{equation}
    \norm{e}_\infty\le\min\left\{r_y,{\eps}/{(2C_\Gamma C_y)}\right\}.
    \label{eq:classical-data-error}
\end{equation}
The first bound implies \(y\in\cY\) by \Cref{eq:data-neighborhood}, thereby verifying the remaining hypothesis of \Cref{thm:local-inverse}.  Set
\begin{equation}
    \widetilde\mu^{(0)}:=0.
    \label{eq:classical-zero-initializer}
\end{equation}
The normalization \(\norm{\lambda}_\infty\le1\) in \Cref{thm:main-para} and the definition \(\cU=\{\mu:\norm{\mu-\lambda}_\infty\le1\}\) in \Cref{eq:cU-defn} imply \(\widetilde\mu^{(0)}\in\cU\).

We next prove that every implemented outer iteration remains in the domain \(\cU\) on which the analytic and optimization estimates hold.  Under the conditions of \Cref{thm:local-inverse}, the exact map satisfies
\begin{equation}
    \sup_{\mu\in\cU}
    \norm{\mathcal T_y(\mu)-\lambda}_\infty\le5/6
\end{equation}
by \Cref{eq:self-map-proof-detailed}.  Combining this inequality with \Cref{eq:classical-map-error} shows that
\begin{equation}
    \sup_{\mu\in\cU} \norm{\widetilde{\mathcal T}_y(\mu)-\lambda}_\infty \le 5/6 +\tau_\Phi.
    \label{eq:classical-implemented-map-margin}
\end{equation}
The tolerance choice in \Cref{eq:classical-outer-choices} gives \(\tau_\Phi\le\eps/12\le1/12\).  Consequently, \Cref{eq:classical-implemented-map-margin} yields
\begin{equation}
    \sup_{\mu\in\cU}
    \norm{\widetilde{\mathcal T}_y(\mu)-\lambda}_\infty
    \le 11/12 <1,
\end{equation}
which proves \(\widetilde{\mathcal T}_y(\cU)\subseteq\cU\).  Induction starting from \Cref{eq:classical-zero-initializer} therefore proves that every computed iterate belongs to \(\cU\).

We now bound the coefficient error after a finite number of implemented outer iterations.  Define \(\widetilde\mu^{(t+1)} =\widetilde{\mathcal T}_y(\widetilde\mu^{(t)})\).  The contraction estimate \Cref{eq:Phi-contraction-detailed}, the fixed-point identity \(\mathcal T_y(\mu^\star(y))=\mu^\star(y)\), and \Cref{eq:classical-map-error} imply
\begin{equation}
    \norm{\widetilde\mu^{(t+1)}-\mu^\star(y)}_\infty
    \le\frac13
    \norm{\widetilde\mu^{(t)}-\mu^\star(y)}_\infty
    +\tau_\Phi.
 \label{eq:classical-inexact-recurrence}
\end{equation}
Both \(\widetilde\mu^{(0)}\) and \(\mu^\star(y)\) belong to \(\cU\), so iteration of \Cref{eq:classical-inexact-recurrence} gives
\begin{equation}
    \norm{\widetilde\mu^{(t)}-\mu^\star(y)}_\infty
    \le 2\cdot 3^{-t}+3\tau_\Phi / 2.
    \label{eq:classical-inexact-convergence}
\end{equation}
Choose
\begin{equation}
    \tau_\Phi\le {\eps}/{12},
    \quad
    N_{\rm out}:=\ceil{\log_3(16/\eps)}=\bigO(\ell_\eps).
    \label{eq:classical-outer-choices}
\end{equation}
Then \Cref{eq:classical-inexact-convergence} yields \(\norm{\widetilde\mu^{(N_{\rm out})}-\mu^\star(y)}_\infty \le\eps/4\).  Moreover, the data-error condition in \Cref{eq:classical-data-error} and the stability estimate \Cref{eq:stability} give \(\norm{\mu^\star(y)-\lambda}_\infty\le3\eps/4\).  Therefore the output \(\widehat\lambda:=\widetilde\mu^{(N_{\rm out})}\) satisfies the target accuracy \(\norm{\widehat\lambda-\lambda}_\infty\le\eps\) in \Cref{eq:intro-goal}.

We choose the forward map evaluation and linear-solver tolerances so that the map error condition in \Cref{eq:classical-outer-choices} holds.  One explicit choice is
\begin{equation}
    \tau_{\rm lin}:={\eps}/{36}, \quad
    \tau_J:=\min\left\{1/{V_R}, {\eps}/{(36C_\Gamma V_R(r_E+1))}\right\},\quad
    \tau_F:={\eps}/{(36C_\Gamma(C_y+1))}.
    \label{eq:classical-tolerance-choice}
\end{equation}
Indeed, each of the three terms in \Cref{eq:classical-tau-Phi} is at most \(\eps/36\).  Substitution of \Cref{eq:classical-tolerance-choice} into the truncation requirement \Cref{eq:classical-L-choice} gives $L=\bigO(\log({eJV_R}/{\eps})) =\bigO(\ell_\eps)$, because the patch choice following \Cref{eq:patch-locality} gives \(R=\bigO(\ell_\eps)\), and the definition following \Cref{eq:ball-prelim} gives \(V_R=\bigO(\ell_\eps^d)\).  Increasing the constant in the patch order \(K=\ceil{C\ell_\eps}\) (\Cref{eq:K_choice}), if necessary, makes \(L\le K\); the support statement in \Cref{lem:connected-word-expansion} then ensures that every connected set retained through degree \(L\) is contained in the patch \(B_R(k)\), where \(R=(K+1)r_0\).

We finally count the arithmetic operations required by the implemented iteration.  By \Cref{prop:classical-row-evaluation}, evaluating all \(nm\) forward-map values and their nonzero partial derivatives in each of the \(N_{\rm out}\) iterations costs
\begin{equation}
    nmN_{\rm out}\operatorname{poly}(L,{s},V_R)C_{\rm ev}^{L}
    \le nm\,\operatorname{poly}({s},\ell_\eps) (J/\eps)^{C_{\rm alg}},
    \label{eq:classical-row-total-cost}
\end{equation}
where \(C_{\rm alg}\) depends only on the fixed locality parameters listed after \Cref{eq:classical-cost-conclusion}, and \Cref{prop:classical-row-evaluation} is defined in \Cref{prop:classical-row-evaluation}.  The formula in \Cref{lem:local-Walsh} permits \(X\) and \(\Gamma=m^{-1}X^{\mathsf T}X\) to be assembled in \(\bigO(Mm)\) operations for fixed \(q\) and \(D_{\rm ov}\); this is absorbed by the right-hand side of \Cref{eq:classical-row-total-cost} because \(M=\bigO(n)\).  Each Jacobi step costs \(\bigO(M)\) by the sparsity of \(\Gamma\), and \Cref{eq:classical-Jacobi-contraction} requires \(\bigO(\log(J/\tau_{\rm lin}))=\bigO(\ell_\eps)\) steps per outer iteration.  Hence all \(N_{\rm out}\) linear solves cost \(\bigO(M\ell_\eps^2)\) operations.  Combining these two costs proves \Cref{eq:classical-cost-conclusion}. {On the event that \Cref{eq:Gamma-close} holds and that the robust frequency estimation results are sufficiently accurate, i.e., \Cref{eq:parallel-gap-accuracy} holds for every sampled $\mathcal{D},\ell,k$ if using the parallel protocol, or \Cref{eq:full-gap-data-error} holds for every sampled $\rho$ if using the sequential protocol, the data satisfy \Cref{eq:classical-data-error}, and the classical post-processing achieves the required accuracy in time polynomial in \(n\), \(J/\eps\), and \(\log(1/\delta)\). The failure budgets allocated above for the protocols of \Cref{thm:main,thm:main-para} ensures this event happens with probability at least \(1-\delta\).}

\section{Notation}
\label{app:notations}
For ease of reference, we summarize the main notation used throughout the paper.

\begingroup

\setlength{\tabcolsep}{4pt}
\renewcommand{\arraystretch}{1.12}
\setlength{\LTleft}{0pt}
\setlength{\LTright}{\fill}
\begin{longtable}{@{}
  >{\raggedright\arraybackslash}p{0.25\linewidth}
  >{\raggedright\arraybackslash}p{0.585\linewidth}
  >{\raggedright\arraybackslash}p{0.12\linewidth}@{}}
\caption{Principal notation, grouped by topic.}
\label{tab:notation}\\
\toprule
Notation & Meaning & Definition \\
\midrule
\endfirsthead
\multicolumn{3}{@{}l}{\tablename~\thetable\ (continued)}\\
\toprule
Notation & Meaning & Definition \\
\midrule
\endhead
\midrule
\multicolumn{3}{r@{}}{Continued on the next page}\\
\endfoot
\bottomrule
\endlastfoot

\multicolumn{3}{@{}l}{\textbf{Hamiltonian and locality}}\\*
\(\Lambda,n,d\)
& {lattice \(\Lambda\) from \Cref{sec:intro}}, number of qubits
  \(n=|\Lambda|\), and spatial dimension.
& \eqref{eq:intro-H} \\
\(\mathcal A,M,P_a\)
& Pauli-term index set, its size \(M=|\mathcal A|\), and the Pauli
  terms in \(H(\lambda)=\sum_{a\in\mathcal A}\lambda_aP_a\).
& \eqref{eq:intro-H} \\
\(\lambda,\mu,\widehat\lambda;\ \varepsilon\)
& True, trial, and estimated coefficient vectors; target accuracy
  \(\|\widehat\lambda-\lambda\|_\infty\le\varepsilon\).
& \eqref{eq:intro-goal}, \eqref{eq:cU-defn} \\
\(q,r_0;\ \zeta,D_{\rm ov}\)
& Maximum term support size and diameter; maximum number of terms
  incident on a site and overlapping a term (including itself).
& \eqref{eq:locality-bounds-H}, \eqref{eq:zeta-Dov} \\
\(B_R(k),V_R,H_{B_R(k)}\)
& Radius-\(R\) ball in Manhattan distance; maximum ball volume
  \(V_R=\max_k|B_R(k)|\); Hamiltonian restricted to terms supported
  entirely in that ball.
& \eqref{eq:ball-prelim}, \eqref{eq:patch-H} \\
\(\|Z\|_\kappa;\ \kappa_0,J\)
& Weighted local interaction norm, its reference locality parameter (arbitrary positive constant),
  and local interaction scale.
& \eqref{eq:local-norm}, \eqref{eq:Ham-local-norm-bound} \\
\(\mathcal U;\ J_{\rm an},\mathfrak B_{\rm an}\)
& Coefficient neighborhood \(\|\mu-\lambda\|_\infty\le1\);
  analytic scale \(J_{\rm an}=4J\) and interaction ball
  \(\|Z\|_{\kappa_0}\le J_{\rm an}\).
& \eqref{eq:cU-defn}, \eqref{eq:analytic-ball} \\

\midrule
\multicolumn{3}{@{}l}{\textbf{Prethermal construction}}\\*
\(G=\nu N+H;\ \nu,N\)
& Strong-field Hamiltonian, field strength, and integer-valued
  charge operator.
& \eqref{eq:generic-G} \\
\(\mathcal P_{N,\omega},\mathcal S_N\)
& Charge-\(\omega\) projection and homological inverse;
  \(\langle W\rangle_N=\mathcal P_{N,0}(W)\) is the
  charge-preserving average.
& \eqref{eq:charge-projection}, \eqref{eq:homological-inverse} \\
\(H_j,D_j,V_j,A_j\)
& ADHH interaction at step \(j\), its charge-preserving and
  charge-changing parts, and the conjugating generator.
  The initial split is \(D=D_0\), \(V=V_0\).
& \eqref{eq:ADHH-recursion} \\
{\(\kappa_j,n_*,s;\ \widehat D\)}
& {Locality parameter \(\kappa_j\); maximum allowed recursion
order \(n_*\); chosen sequential order \(1\le s\le n_*\).
The theorem for the maximal order uses \(D_{n_*}\).
The functions \(E_\rho^\Lambda\) and \(F_{\rho,R}\) use
\(D_s\), with the same \(s\) held fixed as the input varies.}
& \eqref{eq:kappa-j-defn}, \eqref{eq:new-n*},
  \eqref{eq:ADHH-expression} \\
\(W_{\mathbf{a,Q}},\mathcal C_{p,H}\)
& Connected $p$-words (more precisely their evaluation) and the span of such
  formal words. Connectedness refers to the support-overlap graph.
& \Cref{def:p-word} \\

\midrule
\multicolumn{3}{@{}l}{\textbf{Experiments and local forward maps}}\\*
\(\rho=(k,\ell),m\)
& Experiment index: defect site \(k\) and repetition \(\ell\in[m]\);
  \(m\) random field patterns per site.
& \eqref{eq:row-index} \\
\(s_j^\rho,\beta_j^\rho;\ \tau_{j,z}^\rho,Q_j^\rho\)
& Random sign bit and Pauli axis; signed Pauli operator
  \(\tau_{j,z}^\rho=(-1)^{s_j^\rho}\sigma_j^{\beta_j^\rho}\)
  and excitation projector \(Q_j^\rho=(I-\tau_{j,z}^\rho)/2\).
& \eqref{eq:random-pattern}--\eqref{eq:Q-row} \\
\(H_{\rm ctrl}^\rho,H_{\rm tot}^\rho;\ N_\rho,G_\rho\)
& Control profile and physical Hamiltonian
  \(H_{\rm tot}^\rho=H-\nu H_{\rm ctrl}^\rho\);
  {\(G_\rho=H+\nu N_\rho\) is the auxiliary Hamiltonian with
fields on all sites, used to define \(E_\rho^\Lambda\).
The experimental signal is generated by
\(G_\rho^B=H+\nu N_\rho^B\), which has fields only on \(B\)
and differs from \(H_{\rm tot}^\rho\) by a scalar.}
& \eqref{eq:Hctrl}--\eqref{eq:G-rho} \\
\(\ket{\Phi_0^\rho},\ket{\Phi_1^\rho},\ket{\Phi_+^\rho}\)
& Charge-zero and charge-one product states, and their equal
  superposition used for Ramsey preparation.
& \eqref{eq:Phi0_Phi1}, \eqref{eq:Phi-plus} \\
\(\tau_{k,x}^\rho,\tau_{k,y}^\rho;\ S_\rho(t)\)
& Measured quadratures and complex Ramsey signal; its ideal tone is
  \(e^{i[\nu+2E_\rho^\Lambda(\lambda)]t}\).
& \eqref{eq:X-Y-quadrature}, \eqref{eq:signal} \\
\(\ell_\rho;\ E_\rho^\Lambda,F_{\rho,R}\)
& Half-transition functional and its values on the full-lattice
  interaction \(\widehat D_\rho^\Lambda\) and patch interaction
  \(\widehat D_{\rho,R}\), respectively, both at order {\(s\le n_*(\nu)\)}.
  Write \(F_\rho=F_{\rho,R}\) for fixed \(R\).
& \eqref{eq:transition-functional}, \eqref{eq:full-gap},
  \eqref{eq:patch-gap} \\
\(K,R;\ \eta_{\rm patch},\eta_{\rm freq}\)
& Taylor cutoff and patch radius, with \(R\ge(K+1)r_0\);
  patch-approximation and measured half-transition-energy errors.
& \eqref{eq:patch-locality}, \eqref{eq:patch-data-model} \\

\midrule
\multicolumn{3}{@{}l}{\textbf{Coefficient recovery}}\\*
\(F(\mu),y,y_*,e\)
& Patch forward-map vector, measured data, exact patch data
  \(y_*=F(\lambda)\), and combined error \(e=y-y_*\).
& \eqref{eq:global-F}, \eqref{eq:patch-data-model} \\
\(F_0,X;\ \mathcal R\)
& Zeroth-order forward map \(F_0(\mu)=X\mu\), its constant
  Jacobian \(X\), and finite-field correction \(\mathcal R=F-F_0\).
& \eqref{eq:bare-and-correction}, \eqref{eq:row-coeff-set} \\
\(\mathsf J,\Gamma,\mathsf G,\Sigma\)
& Effective Jacobian \(\mathsf J=DF\); Gram matrices
  \(\Gamma=X^{\mathsf T}X/m\),
  \(\mathsf G=\mathsf J^{\mathsf T}\mathsf J/m\);
  population matrix \(\Sigma=\operatorname{diag}(p_a3^{-p_a})\),
  where \(p_a=|\operatorname{supp}(P_a)|\).
& \eqref{eq:Jacobian-def}, \eqref{eq:Gamma},
  \eqref{eq:Hessian-split} \\
\(\mathcal L,r,f,\mathsf H\)
& Least-squares loss \(\mathcal L=\|F(\mu)-y\|_2^2/(2m)\),
  residual \(r=F(\mu)-y\), gradient \(f=\nabla_\mu\mathcal L\),
  and Hessian \(\mathsf H=\nabla_\mu^2\mathcal L\).
& \eqref{eq:loss}--\eqref{eq:Hessian-split} \\
\(\mathcal T_y,\mu^\star(y)\)
& Preconditioned gradient map
  \(\mathcal T_y(\mu)=\mu-\Gamma^{-1}f(y,\mu)\)
  and the unique minimizer in \(\mathcal U\).
& \eqref{eq:iteration-map}, \Cref{thm:local-inverse} \\

\midrule
\multicolumn{3}{@{}l}{\textbf{Parallel experiments}}\\*
\(\mathcal D,L_{\mathcal D},P_{\mathcal D}\)
& Defect batch, minimum defect separation, and connected-word
  threshold \(P_{\mathcal D}=\lceil L_{\mathcal D}/r_0\rceil\).
& \eqref{eq:parallel-separation} \\
\(N_{\mathcal D,\ell}\)
& {Charge of the auxiliary Hamiltonian used to construct
\(\widehat D_{\mathcal D,\ell}(z)\): weight one at each defect
and two at every other site. The experimental charge
\(N_{\mathcal D,\ell}^{\mathcal B}\) agrees with it on
\(\mathcal B\) and vanishes outside \(\mathcal B\).}
& \eqref{eq:parallel-charge} \\
\(n_{\rm par};\ \widehat D_{\mathcal D,\ell}(z)\)
& Prescribed stopping order \(1\le n_{\rm par}\le n_*(\nu)\)
  and stopped interaction \(D_{n_{\rm par}}(zH(\lambda))\).
& \eqref{eq:parallel-stopping-order}, \eqref{eq:parallel-taylor} \\
\(\widehat D_{\mathcal D,\ell}^{(<P)}\)
& Taylor truncation of the stopped interaction, evaluated at \(z=1\),
  retaining degrees \(1,\ldots,P-1\).
& \eqref{eq:parallel-taylor-split} \\
\(E_{\mathcal D,\ell;k}^{(<P_{\mathcal D})},
  \ S_{\mathcal D,\ell;k}(t)\)
& Half transition energy at defect \(k\) for the truncated parallel
  interaction, and the physical parallel Ramsey signal with ideal
  tone \(e^{i[\nu+2E_{\mathcal D,\ell;k}^{(<P_{\mathcal D})}]t}\).
& \eqref{eq:parallel-signal}, \eqref{eq:parallel-ideal-tone} \\
\(F_{\rho,R}^{\rm par}(\mu)\)
& Patch half-transition-energy map computed at the same stopping
  order \(n_{\rm par}\) as the parallel experiment.
& \eqref{eq:parallel-patch-map} \\
\end{longtable}
\endgroup

\end{document}